\documentclass[%
acmsmall,%
screen,%
natbib=false%
]{acmart}

\RequirePackage[
datamodel=acmdatamodel,
style=acmnumeric,
sortcites=true,
]{biblatex}
\usepackage{tikz}
\usetikzlibrary{calc}
\usepackage{keyrustydl,keybook,listings-rust}
\usepackage{zref-clever}
\usepackage{csquotes}
\usepackage[shortlabels,inline]{enumitem}
\usepackage{acronym}
\usepackage{stmaryrd}
\usepackage{multicol}

\zcRefTypeSetup{section}{cap=true}
\zcRefTypeSetup{appendix}{cap=true}
\zcRefTypeSetup{figure}{cap=true,abbrev=true}
\zcRefTypeSetup{example}{cap=true}
\zcRefTypeSetup{table}{cap=true}
\zcRefTypeSetup{listing}{cap=true}
\zcRefTypeSetup{definition}{cap=true}
\zcRefTypeSetup{theorem}{cap=true}
\zcRefTypeSetup{lemma}{cap=true}
\zcRefTypeSetup{page}{abbrev=true}

\zcRefTypeSetup{rule}{
  Name-sg = Rule ,
  name-sg = rule ,
  Name-pl = Rules ,
  name-pl = rules ,
  refbounds = {,(,),} ,
}

\zcRefTypeSetup{case}{
  Name-sg = Case ,
  name-sg = case ,
  Name-pl = Cases ,
  name-pl = case ,
}

\zcLanguageSetup{english}{
  type=appendix,
  Name-sg-ab = App.,
  name-sg-ab = app.,
  Name-pl-ab = App.,
  name-pl-ab = app.
}

\zcLanguageSetup{english}{
  type=page,
  Name-sg-ab = P.,
  name-sg-ab = p.,
  Name-pl-ab = Pp.,
  name-pl-ab = pp.
}

\AtEndPreamble{%
  \theoremstyle{acmdefinition}
  \newtheorem{case}{Case}[section]}

\AddToHook{env/conjecture/begin}{%
  \zcsetup{countertype={theorem=conjecture}}}
\AddToHook{env/proposition/begin}{%
  \zcsetup{countertype={theorem=proposition}}}
\AddToHook{env/lemma/begin}{%
  \zcsetup{countertype={theorem=lemma}}}
\AddToHook{env/corollary/begin}{%
  \zcsetup{countertype={theorem=corollary}}}
\AddToHook{env/example/begin}{%
  \zcsetup{countertype={theorem=example}}}
\AddToHook{env/definition/begin}{%
  \zcsetup{countertype={theorem=definition}}}

\acrodef{FLS}{ferrocene language specification}
\acrodef{MIR}{mid-level intermediate representation}
\acrodef{HIR}{high-level intermediate representation}
\acrodef{LAGC}{locally abstract, globally concrete}
\acrodef{DL}{dynamic logic}
\acrodef{CFG}{control flow graph}
\acrodef{SOS}{structural operational semantics}
\acrodef{BOS}{big-step operational semantics}
\acrodef{LIFO}{last in, first out}
\acrodef{FR}{featherweight \Rust}
\acrodef{LLBC}{low-level borrow calculus}
\acrodef{LTL}{linear temporal logic}

\usepackage{todonotes}

\setcopyright{acmlicensed}
\copyrightyear{2026}
\acmYear{2026}
\acmDOI{XXXXXXX.XXXXXXX}

\acmJournal{TOPLAS}
\acmVolume{XXX}
\acmNumber{YYY}
\acmArticle{ZZZ}
\acmMonth{9}

\usepackage{lagc}

\begin{document}
\title{A Complete, Formal Semantics for \Rust Source Code}

\author{Daniel Drodt}
\email{daniel.drodt@tu-darmstadt.de}
\orcid{0000-0003-3036-8220}
\affiliation{%
  \institution{Technical University of Darmstadt}
  \city{Darmstadt}
  \state{Hessen}
  \country{Germany}
}


\begin{abstract}
  Formally reasoning about \Rust programs requires a rigorous formal
  semantics, especially in the context of deductive verification and
  concurrent programming. We present a modular, flexible semantics for
  a significant subset of (close to) source code level \Rust, based on
  the recent \emph{locally abstract, globally concrete} semantics
  framework, separating local evaluation of expressions from their
  composition into concrete traces. The semantics is extended to model
  \Rust's asynchronous programming features and \Rust's most popular
  async runtime, Tokio. Based on our more abstract formalization, we
  establish the fairness of Tokio's scheduler. Further, we show the
  applicability of our semantics to deductive verification of \Rust by
  providing soundness proofs for a \Rust program logic.
\end{abstract}



\maketitle

\section{Introduction}
\label{sec:introduction}

\Rust is a modern, flexible, systems-level programming language of
growing importance. It is used in many crucial application domains,
most notably the Linux kernel~\cite{LGYWX2024}. This importance and
usage increases the need for formal analysis and formal verification
of \Rust programs, both of which are well supported by \Rust's
language design. Accordingly, several verification tools for Rust are
being
developed~\autocite{ABFGMMPS2022,DJM2022,DH2026,HP2022,LHCBSZHPH2023,BBFHKSS2025}. With
its complex ownership model and reference system, \Rust provides a
powerful yet unusually restrictive type system that ensures memory
safety and the absence of data races~\cite{KN2022}.

The lack of an official, formal, source code level \Rust semantics
makes validation of these tools problematic. Without such a semantics,
interoperability of verification tools is difficult, as tools may not
be based on the same, or even compatible, formalisms. Similarly,
verification of the \Rust compiler needs a formal semantics for each
of \Rust's intermediate representations.

The semantics problem is partially addressed by formal models of parts
of \Rust's type system or intermediate representations, typically the
\ac{MIR}~\autocite{AMF,HFP2024,JJKD2018,MiniRust}. These semantics
provide usable and rigorous formalizations of \Rust's borrow checker
and operational semantics of subsets of the language, each offering
different trade-offs and specializations. Nevertheless, they are
limited in at least two respects. First, being low-level and focused
on the \ac{MIR}, they offer little for reasoning about \Rust programs
at a level closer to the source code. A high-level representation
crucially allows, e.g., comprehension of the complex proof
state~\autocite{DH2026}.

Second, no formal semantics for \Rust's concurrent features
(\enquote{async Rust}) has been developed. As \enquote{fearless
  concurrency} is a major feature and promise of
\Rust~\autocite{KN2022}, this is an unfortunate limitation. Without
formal semantics, it remains difficult to reason about concurrent and
parallel \Rust programs or to prove them correct, free of deadlocks,
and similar properties. Furthermore, it remains elusive to model
concrete schedulers or runtimes, which is especially relevant as \Rust
allows different asynchronous runtimes to be used. Here too, a
high-level view seems more promising than semantics of lower-level
representations.

As \Rust has matured and its type system, borrow checking, and memory
model have been examined, refined, and formalized, we argue that a
sensible next step is to examine the high-level features, building on
these closely-studied foundations.

Hence, we see a crucial gap left open: a formal semantics of \Rust
suitable for formalizing high-level properties of \Rust
programs. Ideally, such a semantics should
\begin{enumerate*}[(1)]
\item be highly flexible and compositional to accommodate future
  language extensions;
\item formalize \Rust at the source code level (or close to it);
\item allow the formalization of asynchronous \Rust code;
\item permit modeling of and reasoning about asynchronous runtimes;
  and
\item ease deductive verification of \Rust programs.
\end{enumerate*}

To meet these goals, \emph{modularity} of the semantics is paramount:
Extendability is only feasible if adding a new language feature does
not require to significantly rewrite the existing semantics
rules. Additionally, a modular semantics eases verification effort, as
calculus proof rules are also made modular in the sense that it is
possible to extend the calculus with new rules and constructs without
revising the rest of the calculus.

However, semantics in the style of \ac{SOS} struggle with modularity.
As Peter \citeauthor{M2004} put it, \enquote{conventional \ac{SOS}
  descriptions of programming languages have quite poor modularity,}
which leads to complex semantics, as \textcquote{M2004}{\textins{w}hen
  extending a pure functional language with concurrency primitives
  and/or references, for instance, the \ac{SOS} rules for all the
  functional constructs had to be completely reformulated}. While
\textcite{M2004} provides a modification of \ac{SOS} that improves
modularity, recent developments in trace-based semantics offer an even
more flexible, complete, and modular
path~\autocite{B2007,K2019,BGHS2023,HEJ2009}.

The present paper remedies this situation by providing a
comprehensive, formal semantics for most parts of safe \Rust's
\ac{HIR} based on the \ac{LAGC} semantics framework~\cite{LAGC23},
which was developed with concurrent and parallel programs and
deductive verification in mind. \ac{LAGC} achieves a modular and
flexible semantics by separating the local evaluations of expressions
from the global state and concerns such as scheduling. This approach
is modular in the sense that each syntactic structure---expressions or
statements---have one local rule that is independent of the global
context. Adding new language features consists largely of formulating
new local rules. Furthermore, the approach is well-suited for
concurrency by delegating all scheduling concerns to the global layer
and building hybrid traces over states and events for each program
execution. These traces contain all relevant information of state
changes and scheduling relevant events. Combining these elements,
\ac{LAGC} defines the semantics of a specific program as the set of
all well-formed traces generated by the global composition rules.

\ac{LAGC} has been successfully applied to a number of programming
languages and settings, such as ProMeLa and active object
languages~\autocite{dDFHHJKRSSY2017} like
ABS~\autocite{LAGC23,DHJPT2017}, demonstrating its
flexibility. Further examples include Real-Time ABS~\autocite{T2022a}
and Timed Rebeca~\autocite{HJT2025}. The \ac{LAGC} approach has also
been used, in a slightly different form, in the context of deductive
verification~\autocite{BGHS2023}. These examples demonstrate that
\ac{LAGC}, with its modular approach, separation of local and global
context, and hybrid traces, is more suitable than the traditional
\ac{SOS} semantics for purposes of verification and the modeling of
concurrent programming languages.

For these reasons, \ac{LAGC} is a viable semantics framework to
achieve the goals outlined above and it allows a semantics of
higher-level \Rust. Hence, the presented semantics is based on
\ac{LAGC}. Our detailed contributions are:
\begin{enumerate}
\item We provide a semantics for most parts of safe \Rust's \ac{HIR}
  in the \ac{LAGC} framework and demonstrate the semantics' usability
  with a number of examples. To demonstrate the applicability of the
  semantics, we support complex language features and control flows,
  such as loops with breaks and continues.
\item We extend this semantics to async \Rust, modeling any possible
  scheduling of \lstinline|async| functions. Notably, to our
  knowledge, this is the first formal semantics of async \Rust in the
  literature.
\item Based on the asynchronous semantics, we formalize \Rust's most
  widely-used runtime, the \enquote{current thread runtime} of
  Tokio~\autocite{tokio}, and prove its fairness. We stress that our
  approach is straightforwardly adaptable to other runtimes, with the
  async \Rust semantics carrying over.
\item Finally, we show how our semantics can be used to prove
  soundness of the calculus for \RustyDL, a program logic for
  \Rust~\autocite{DH2026} based on \ac{DL}~\autocite{P1976},
  demonstrating how the modular rules of \ac{LAGC} ease such
  proofs. Thereby, we also highlight the interaction of \ac{DL} and
  \ac{LAGC}.
\end{enumerate}

\zcref[S]{sec:basics} introduces the necessary background and
definitions on which our semantics is built, as well as the general
\ac{LAGC} approach. We present the semantics for synchronous \Rust in
\zcref{sec:lagc-semant-rust} and extend it to async \Rust in
\zcref{sec:asynchr-rust-semant}. \zcref[S]{sec:fair-tokio} then builds
upon the async semantics to represent one of Tokio's runtimes in
\ac{LAGC} and proves its fairness guarantees. The calculus of the
program logic \RustyDL~\autocite{DH2026} is proven sound in
\zcref{sec:sound-dl}. We present related work in
\zcref{sec:related-work} and conclude the paper in
\zcref{sec:future-work-concl}.


\section{Prerequisites}
\label{sec:basics}

The goal of the \ac{LAGC} approach~\autocite{LAGC23} is to represent a
program run as a \emph{trace} consisting of states and events and to
represent a program's entire behavior as the set of all possible
traces starting from a concrete state. As a principle idea to tackle
the complexity of all possible executions, \ac{LAGC} consists of many
\emph{local} rules for each syntactic structure, which consider
abstract states and traces, and \emph{global} rules, which consider
concrete states and traces and are concerned with the correct
composition of such traces. Events in traces aid this effort by
providing more structural information, e.g., that a function was
invoked or a panic occurred.

Building an \ac{LAGC} semantics consists of four steps:
\begin{enumerate*}[(1)]
\item the local evaluation, mapping expressions to (abstract)
  continuation traces, without knowledge of the surrounding context;
\item a well-formedness predicate, restricting the set of allowed
  traces and enforcing, e.g., proper ordering of events;
\item trace composition rules, defining how the local traces are
  stitched together to concrete traces; and
\item global program semantics, defining the initial state and allowed
  end states.
\end{enumerate*}
This framework yields a set of concrete traces describing all possible
executions for a given program from an initial state and achieves a
modular semantics: Each kind of expression has exactly one local rule
irrespective of context.

This section gives the necessary background for states, events, and
traces, mainly following the definitions in~\cite{LAGC23}, modified to
adapt to \Rust's more functional, expression-based structure. This
modification in fact simplifies some parts of the \ac{LAGC} approach.

\subsection{Values and States}
\label{sec:states}

Due to the high-level concerns of our semantics, we do not map \Rust
constructs to low-level or even machine-level domains; rather, we map
to mathematical domains: \Rust's concrete, machine word integer types,
such as \lstinline|i32| and \lstinline|u8|, are mapped to mathematical
integers $\mathbb{Z}$, while \lstinline|bool| is mapped to boolean
values $\trueSem$ and $\falseSem$. We proceed similarly with compound
types like arrays and structs. This abstraction is not a limitation of
the semantics. We ensure the correct behavior of \Rust types as part
of the evaluation of the relevant expressions.

\begin{definition}[(Starred) Semantic Values]\label{def:val}
  We write $\Var$ for the set of program variables with typical
  elements $\pv{x}, \pv{y}, \pv{z}$, etc. Further, we define a
  \emph{place} $p\in\Place$ as an abstract representation of memory
  locations. Such a place is either a program variable~\pv{x}, an
  array field $p[n]$ of another place~$p$ where $n$ is a natural
  number, or a field access $p.\mathit{identifier}$.
  \begin{align*}
    \pv{x}\in\Var&::=\mathit{identifier}\\
    p\in\Place &::= \pv{x} \mid p[n] \mid p.\mathit{identifier}\\
    n\in\mathbb{N} &::= 0 \mid 1 \mid \ldots
  \end{align*}

  The set of \emph{semantic values} $\Val$, with typical elements
  $v,v_1,v_2,\ldots$, includes boolean, integer, reference, function,
  tuple, array, struct, and enum variant values. The latter compound
  types are defined as their \Rust counterparts, where all fields are
  values. Reference values are either \emph{shared} (denoted by $s$)
  or \emph{mutable} (denoted by $m$) and contain a place~$p\in\Place$
  and a \emph{borrow identifier} $b\in\BId$; e.g.,
  $\mRefVal{\pv{x}[1]}{b_0}$ is a mutable reference of
  place~$\pv{x}[1]$ with borrow identifier~$b_0$. Function values are
  represented as function names.
  \begin{align*}
    v\in\Val &::= \trueSem \mid \falseSem \mid z \mid \sRefVal{p}{b} \mid \mRefVal{p}{b} \mid \rust{($\many{v}$)} \mid \rust{[$\many{v}$]} \mid \textit{structs} \mid \textit{enums} \mid \mathit{functions} \\
    z\in\mathbb{Z} &::= n \mid -1 \mid -2 \mid \ldots
  \end{align*}

  The set of \emph{starred values} $\SVal \supset \Val$, with typical
  element $\sval$, extends values with the starred value~$\ast$,
  representing unknown, \emph{symbolic} values.
  \begin{align*}
    \sval\in\SVal &::= v \mid \ast
  \end{align*}
\end{definition}

\begin{example}
  The semantic value to which the \Rust expression
  \lstinline|(true, 4 * 2, &mut n)|, where \pv{n} is a program
  variable of type \lstinline|i32|, evaluates is
  $(\trueSem, 8, \mRefVal{\pv{n}}{b_0})$, where $b_0\in\BId$ is
  unique.
\end{example}
  
To associate variables with values, we use a \emph{symbolic state}
$\sigma$, modelled as a partial mapping
\(\sigma:\Var\rightharpoonup\SVal\) from variables to starred
values. We express a modification of a variable~\pv{x} with the
(starred) value~$\sval$ as the \emph{state update} of $\sigma$:
\[
  \update{\sigma}{\pv{x}}{\sval}(\pv{y})=
  \begin{cases}
    \sval &\pv{x}=\pv{y}\\
    \sigma(\pv{y}) &\pv{x}\ne\pv{y}\ .
  \end{cases}
\]

As values can be nested in arrays, tuples, structs, and enums, we
define syntactic sugar for updating not just variables but places
$p\in\Place$. The base case of $p=\pv{x}$ for some variable~\pv{x} is
handled by state updates. Updates to an array place $p[n]$ is resolved
by updating place~$p$ with a changed array value.
\[
  \update{\sigma}{p[n]}{\sval}=\update{\sigma}{p}{[\sval_0,\ldots,\sval_{n-1},\sval,\sval_{n+1},\ldots,\sval_m]}
  \quad \text{where}~n\le m \wedge \sigma(p) =
  [\sval_0,\ldots,\sval_m]
\]
The case for updating field places $p.\mathit{identifier}$ is
analogous: $\update{\sigma}{p.\mathit{identifier}}{\sval}$.

The \emph{symbolic variables} of state $\sigma$ are defined as
\(\symb(\sigma)=\{\pv{X}\in\Var\mid\sigma(\pv{X})=\ast\}\) and are
used to abstract from concrete evaluations, e.g., for function
parameters. We typically denote them with capital letters. If
$\symb(\sigma)=\emptyset$, we call $\sigma$ \emph{concrete.}

As our semantics may introduce new variables, we require a notion of
\emph{state extension.}

\begin{definition}[State Extension]\label{def:state-extension}
  Let $\sigma_1,\sigma_2$ be two states. We say $\sigma_2$
  \emph{extends} $\sigma_1$---written as
  $\sigma_1\subseteq\sigma_2$---iff
  $\dom(\sigma_1)\subseteq\dom(\sigma_2)$ and
  $\forall\pv{x}\in\dom(\sigma_1).\
  \sigma_1(\pv{x})=\sigma_2(\pv{x})$.
\end{definition}

\subsection{Traces and Events}
\label{sec:traces}

Traces are sequences of states and events, the latter of which give a
structure and additional information to traces. Events cannot change
the state; we add the relevant events for \Rust as we go along. Like
states, traces can be symbolic.

\begin{definition}[Symbolic Trace]\label{def:trace}
  We define \emph{symbolic traces} inductively by the following rules,
  where $\epsilon$ is the empty trace:\\
  \[
    \begin{array}{r@{\;::=\;}l@{\qquad}l}
      \tau & \epsilon\mid\consTr{t}{\tau} & \text{(trace)}\\
      t & \sigma\mid\textit{ev}(\overline{\mathit{ep}}) & \text{(trace element)}\\
      \mathit{ep} & \pv{x}\mid v & \text{(event parameter)}
    \end{array}
  \]
  where $\textit{ev}(\overline{\mathit{ep}})$ is an \emph{event
    marker} and $\many{\mathit{ep}}$ are the \emph{event parameters,}
  i.e., variables or values.
\end{definition}

\begin{example}\label{ex:trace}
  Let $\sigma_0:=[\pv{X}\mapsto\ast,\pv{y}\mapsto\falseSem]$ be a
  symbolic state and $\mathit{ev}_0$ an event marker with one
  parameter. Then,
  $\tau_0:=\consTr{\update{\sigma_0}{\pv{y}}{\trueSem}}{\consTr{\mathit{ev}_0(\pv{X})}{\consTr{\sigma_0}{\epsilon}}}$
  is a symbolic trace, representing an execution starting in
  (symbolic) state~$\sigma_0$, then event $\mathit{ev}_0(\pv{X})$, and
  ending in the updated state where $\pv{y}$'s value is changed to
  false.
\end{example}

We take care that each event in a trace is surrounded on either side
by a state. This ensures that that functions over traces like
$\last(\tau)$ and $\lastEv(\tau)$, which evaluate to the last state
and event of trace~$\tau$, respectively, are always defined for
non-empty traces.

Often, it is necessary to construct a trace~$\tau$ from two given
traces $\tau_1,\tau_2$. We call this process \emph{trace composition.}
We write $\mkTr{\sigma}$ for the \emph{singleton trace}
$\consTr{\sigma}{\epsilon}$ and $\concatTr{\tau_1}{\tau_2}$ for the
\emph{concatenation} of two symbolic traces $\tau_1,\tau_2$. For
example, we write the trace from \zcref{ex:trace} as
\[
  \tau_0=\consTr{\update{\sigma_0}{\pv{y}}{\trueSem}}{\consTr{\mathit{ev}_0(\pv{X})}{\mkTr{\sigma_0}}}\ .
\]

Traces $\tau_1$ and $\tau_2$, where $\tau_1$ ends in the state with
which $\tau_2$ begins, can be \enquote{stitched together} with the
\emph{semantic chop} operation\footnote{The name \enquote{chop} stems
  from temporal interval logic~\autocite{HMM1983}.}:
\[\chopTrSem{\tau_1}{\tau_2}=\concatTr{\tau}{\tau_2}~\text{where}~\tau_1=\consTr{\sigma}{\tau}~\text{and}~\tau_2=\concatTr{\mkTr{\sigma}}{\tau'}\
  .\]
Note that the chop is a partial operation: If $\tau_1$ is empty or if
the last state of~$\tau_1$ and first state of $\tau_2$ are
incompatible, the chop is undefined.

We also require a way to concretize symbolic traces into concrete
traces. To this end, we work with concretization mappings~$\rho$.

\begin{definition}[Concretization Mapping]\label{def:concr}
  A \emph{concretization mapping} $\rho:\Var\rightharpoonup\Val$ is a
  partial mapping for state $\sigma$ if
  $\dom(\rho)\cap\dom(\sigma)=\symb(\sigma)$. We denote the
  \emph{concretization of $\sigma$ with $\rho$} as
  \[\mapComp{\rho}{\sigma}=\rho\cup\{\pv{x}\mapsto\sigma(\pv{x}) \mid \pv{x}\in\dom(\sigma)\setminus\dom(\rho)\}\ .\]
\end{definition}

\begin{proposition}
  Let $\sigma$ be a state and $\rho$ a $\sigma$ concretization
  mapping. Then $\mapComp{\rho}{\sigma}$ is a concrete state.
\end{proposition}

\begin{definition}[Trace Concretization]\label{def:concr-tr}
  We extend the notion of \emph{concretization mapping} to traces in
  the natural way: A partial mapping~$\rho$ is a concretization
  mapping for $\tau$ if it is a concretization mapping for all states
  in $\tau$. The \emph{concretization} $\rho(\tau)$ of $\tau$ is
  defined as:
  \begin{align*}
    \rho(\epsilon)&=\epsilon\\
    \rho(\consTr{t}{\tau})&=\consTr{\rho(t)}{\rho(\tau)}\\
    \rho(\sigma)&=\mapComp{\rho}{\sigma}\\
    \rho(\textit{ev}(\many{\mathit{ep}}))&=\textit{ev}(\rho(\many{\mathit{ep}}))\\
    \rho(v)&=v
  \end{align*}
\end{definition}

\begin{example}
  The trace
  $\tau_0=\consTr{\update{\sigma_0}{\pv{y}}{\trueSem}}{\consTr{\mathit{ev}_0(\pv{X})}{\mkTr{\sigma_0}}}$
  in \zcref{ex:trace} is symbolic. With $\rho_0:=[\pv{X}\mapsto 3]$ we
  concretize it to
  $\mapComp{\rho_0}{\sigma_0}=[\pv{X}\mapsto 3, \pv{y}\mapsto
  \falseSem]$,
  $\mapComp{\rho_0}{\update{\sigma_0}{\pv{y}}{\trueSem}}=[\pv{X}\mapsto
  3, \pv{y}\mapsto
  \trueSem]=\update{\rho(\sigma_0)}{\pv{y}}{\trueSem}$, and
  $\rho(\tau_0):=\consTr{\update{\rho(\sigma_0)}{\pv{y}}{\trueSem}}{\consTr{\textit{ev}_0(3)}{\mkTr{\rho(\sigma_0)}}}$.
\end{example}

Finally, we require \emph{continuation markers,} as placeholders
marking the unevaluated remainder of an expression.

\begin{definition}[Continuation Marker]\label{def:cont}
  Let $e$ be a \Rust expression. The \emph{continuation marker}
  $\cont{e}$ states that a trace is extended by the traces resulting
  from evaluating~$e$.
\end{definition}

The set of symbolic traces with continuation markers, i.e., all
$\concatTr{\tau}{\cont{e}}$, is denoted as $\ctraces$. We call the
elements of $\ctraces$ \enquote{continuation traces.}

Finally, to ease reasoning over traces, they must adhere to a certain
structure, e.g., agree on the set of symbolic variables. Additionally,
our semantics should not generate any symbolic traces, only concrete
ones. \zcref[S]{def:wf-trace} introduces well-formed and concrete (or
\enquote{shining}) traces, which our semantics generates.

\begin{definition}[Well-Formed and Shining Trace]\label{def:wf-trace}
  Let $\tau$ be a symbolic trace and
  $S=\bigcup_{\sigma\in\tau}\symb(\sigma)$. Trace~$\tau$ is
  \emph{well-formed} if conditions
  \zcref[noname]{eq:wf-tr-stays-symb,eq:wf-tr-ev-symb,eq:wf-tr-ev-ext}
  hold.
  \begin{align}
    \forall\sigma\in\tau.\, (\dom(\sigma)\setminus\symb(\sigma))\cap S=\emptyset \label{eq:wf-tr-stays-symb}\\
    \forall\textit{ev}(\overline{\mathit{ep}})\in\tau.\, \vars(\overline{\mathit{ep}})\subseteq S\label{eq:wf-tr-ev-symb}\\
    \forall\textit{ev}(\overline{\mathit{ep}}),\tau_1,\tau_2,\sigma_1,\sigma_2.\ (\tau=\concatTr{\consTr{\consTr{\consTr{\sigma_2}{\textit{ev}(\overline{\mathit{ep}})}}{\sigma_1}}{\tau_1}}{\tau_2}\rightarrow \sigma_1\subseteq\sigma_2) \label{eq:wf-tr-ev-ext}
  \end{align}
  Condition~\zcref[noname]{eq:wf-tr-stays-symb} ensures that each
  variable is either symbolic in all states or in no state.
  Condition~\zcref[noname]{eq:wf-tr-ev-symb} holds if all variables in
  events are symbolic. Finally,
  condition~\zcref[noname]{eq:wf-tr-ev-ext} states that the state
  succeeding an event is an extension of the state preceding the
  event. $\sigma_1\subseteq\sigma_2$ expresses that $\sigma_2$ is an
  extension of $\sigma_1$ (see \zcref{def:state-extension}); hence,
  events may not alter the state but only introduce new variables to
  it.
  
  If a well-formed trace contains only concrete states (i.e.,
  $S=\emptyset$), we call it \emph{concrete} or \emph{shining.}
  Shining traces are denoted by $\sh$.
\end{definition}

\begin{proposition}
  The concretization of a well-formed trace is a well-formed shining
  trace.
\end{proposition}


\section{LAGC Semantics for \Rust}
\label{sec:lagc-semant-rust}

This section presents an \ac{LAGC} semantics for synchronous, safe
\Rust. Concretely, we cover the expressions and statements in
\zcref{fig:syn-rs}. Type, trait, and module definitions are left out,
as they have no evaluation and, therefore, require no \ac{LAGC}
treatment. Note that, due to our focus on the \ac{HIR}, macros are
already expanded (though we still denote panic expressions, i.e.,
forceful errors that stop the program's execution, as
\lstinline|panic!()| for readability) and all loops are transformed
into the singular \lstinline|loop| expression. For clarity, we assume
a number of straightforward normalizations of the \ac{HIR}
expressions:
\begin{enumerate*}[(i)]
\item all loops, \lstinline|break|, and \lstinline|continue|
  expressions have unique labels;
\item all blocks, and \lstinline|break| expressions contain
  expressions, e.g., \lstinline|break 'l 2 * 3| where \lstinline|'l|
  is the label and \lstinline|2 * 3| is the \lstinline|break|'s
  expression (this is achieved by inserting the unit expression
  \enquote{\lstinline|()|} where necessary);
\item all \lstinline|if| expressions have an \lstinline|else| branch,
  and
\item all functions have at least one explicit \lstinline|return|
  expression.
\end{enumerate*}

Observe that most syntactic structures that would be statements in
other languages are expressions in \Rust, e.g., \lstinline|if| and
\lstinline|loop| expressions. The intuitive semantics of such
expressions is explained alongside the formal semantics but we refer
to the \Rust book~\cite{KN2022}, the \Rust reference~\cite{RustRef},
and the \ac{FLS}~\cite{FLS} for details.

\begin{figure}[h]
  \centering
  \begin{array}[t]{r@{\hspace{2pt}}r@{\hspace{2pt}}l}
    e\in\Expr &::=& v \mid  \rust{x} \mid \rust{!$e$} \mid \rust{-$e$} \mid \rust{*$e$} \mid \rust{$e$ $\mathit{op}$ $e$} \mid \rust{$\mathit{ae}$ = $e$} \mid \rust{[$\many{e}$]} \mid \rust{$(\many{e})$} \mid e\rust{($\many{e}$)} \mid \rust{$e$.$\mathit{identifier}$($\many{e}$)} \\
              & & \quad \mid \kwMatch~\rust{$e$ \{$\many{arm}$\}} \mid \rust{$e$.$\mathit{identifier}$} \mid \rust{$e$[$e$]} \mid \rust{\&$e$} \mid \rust{\&\kwMut $e$} \\
              & & \quad \mid \rust{\kwBreak~'$l$ $e$} \mid \rust{\kwContinue '$l$} \mid \rust{\kwReturn $e$} \mid \rust{x \{$\many{\rust{$\mathit{identifier}$: $e$}}$\}} \mid \rust{[$e$; $n$]} \mid \rust{\{$\overline{s}$ $e$\}} \\
              & & \quad \mid \rust{\kwIf $e$ $e$ \kwElse $e$} \mid \rust{'$l$: \kwLoop $e$} \mid \rust{\panic()} \mid \rust{?$i$} \\
    s\in\Stmt &::=& e \mid \rust{\kwLet $\mathit{pat}$ = $e$;} \\
    l\in\Lbl &::=& \mathit{identifier} \\
    i\in\CId &::=& \mathit{call identifier} \\
    \mathit{op}\in\mathit{Op} &::=& \rust{+} \mid \rust{-} \mid \rust{*} \mid \ldots \mid \rust{+=} \mid \rust{-=} \mid \rust{*=} \mid \ldots \\
    \mathit{arm}\in\mathit{Arm} &::=& \rust{$\mathit{pat}$ => $e$}\\
    \mathit{pat}\in\mathit{Pat} &::=& \rust{\_} \mid v \mid \pv{x} \mid \rust{x \{$\many{\mathit{identifier}:\;\mathit{pat}}$\}} \mid \rust{x($\many{\mathit{pat}}$)} \mid \rust{($\many{\mathit{pat}}$)} \mid (\rust{$\mathit{pat}$ | $\mathit{pat}$}) \\
              & & \quad \mid \rust{$e$..$e$} \mid \rust{$e$..=$e$} \mid \rust{$e$..} \mid \rust{$e$..=} \mid \rust{..$e$} \mid \rust{..=$e$} \\
    \mathit{ae}\in\AssignExpr &::=& \mathit{pe} \mid \rust{\_} \mid \rust{[$\many{ae}$]} \mid \rust{x \{$\many{\rust{$\mathit{identifier}$: $\mathit{ae}$}}$\}} \mid \rust{($\many{ae}$)} \mid \rust{x($\many{ae}$)}\\
    \mathit{pe}\in\PlaceExpr &::=& \pv{x} \mid \rust{$\mathit{pe}$[$e$]} \mid \rust{$\mathit{pe}$.$\mathit{identifier}$} \mid \rust{*$\mathit{pe}$} \\
    \pv{foo}\in\Fn &::=& \rust{\kwFn $\mathit{identifier}$($\many{x}$) $e$}\\
    r\in\Program &::=& \many{\pv{foo}}\;\rust{\kwFn main() $e$}
  \end{array}
  \Description{Definition of expressions, operators, statements,
    patterns, functions, and programs.}
  \caption{The syntax of \Rust's \ac{HIR} covered in this section.}
  \label{fig:syn-rs}
\end{figure}

The syntax of expressions also contains an additional construct needed
to model function calls ergonomically: \rust{?$i$}, where $i \in \CId$
is a \emph{call identifier}, serves as a stand-in for the value that
is to be returned by the function call with identifier~$i$. The
precise, formal semantics are outlined below.

Some rules below, especially regarding assignments, also consider
important subsets of $\Expr$, which we enumerate above: \emph{Place
  expressions} ($\PlaceExpr \subset \Expr$) are expressions that
represent a memory location, such as local variables, index
expressions, dereferences, and field expressions. \emph{Assignee
  expressions} ($\AssignExpr\subset\Expr$) are expressions that can
appear on the left-hand side of assignments. Aside from place
expressions and the underscore expression (i.e., an assignment to
nowhere), these are used for destructuring. If a place expression
$\mathit{pe}$ contains only values in the index and no dereferences,
then $\mathit{pe}\in\Place$. Further, if every place expression in an
assignee expression $\mathit{ae}$ is in $\Place$, it is an element of
the \emph{assignee value} set $\AssignVal \subset \Expr$.

The considered subset covers almost all of safe \Rust, the only
limitation being closures, which can be translated to normal
functions. Async \Rust's features are discussed in
\zcref{sec:asynchr-rust-semant}.
We are interested in an operational semantics of \Rust, hence we
abstract away from memory layout and typing. All expressions and
functions in programs must be well-typed and all types must be known.

\begin{example}[Running Example]\label{ex:sync-run}
  We consider a running example in this section which highlights many
  important \Rust features considered here, see
  \zcref{lst:running}. It contains three functions: \lstinline|fac|,
  computing the factorial recursively, \lstinline|set_all| setting all
  fields of an array passed via a mutable reference to a given
  value. The \lstinline|main| function calls the other two with
  example arguments.
\begin{lstlisting}[caption={A simple \Rust program adhering to the simplifications above.},label=lst:running,numbers=left,multicols=2,xleftmargin=2em]
fn fac(n: u64) -> u64 {
  return if n <= 1 { 1 } %\label{line:if}%
  else { %\label{line:block-start}%
    let left = n; let right = fac(n - 1); %\label{line:rec-call-fac}%
    left * right }; () %\label{line:block-end}%
}
fn set_all(arr: &mut [u64; 3], v: u64) {
  let mut i = 0;
  'l1: loop { if i == 3 { break 'l1 () } else {()}
    *a[i] = v; i += 1 }
  return (); ()
}
fn main() {
  let b = false;
  let c = !b; %\label{line:not-b}%
  let f = fac(4); %\label{line:init-call-fac}%
  let arr = [8; 3];
  set_all(&mut arr, f); %\label{line:call-set-all}%
  return (); ()
}
\end{lstlisting}
\end{example}

\subsection{Local Evaluation}
\label{sec:local-evaluation}

We define a \emph{local evaluation function}
$\valDnoargs: \Expr \rightarrow 2^\ctraces$, mapping \Rust expressions
to a set of symbolic continuation traces relative to
state~$\sigma$. Each such trace contains state changes, emitted
events, and the continuation marker, denoting the expression or value
with which we continue. We define one local rule per kind of
expression (syntax in \zcref{fig:syn-rs}). Due to the rich syntax, the
number of rules is quite large. Since \Rust lacks a full, official
specification~\autocite{HKMB2023}, our local rules follow the \Rust
reference~\autocite{RustRef} and the unofficial yet \enquote{adopted}
\ac{FLS}~\autocite{FLS}. Observe that for the subset we consider,
every program has a deterministic execution.

Values are already fully evaluated. Hence, we define a value's local
evaluation as the empty set.
\evalrule{val}{
  \valD{v} & = & \emptyset
}

Reading a local variable produces a set with a single continuation
trace, starting with the state~$\sigma$ and a continuation with
\pv{x}'s value in $\sigma$.
\evalrule{pv}{
  \valD{\pv{x}} & = & \{\concatTr{\mkTr{\sigma}}{\cont{\sigma(\pv{x})}}\}
}

\subsubsection{Evaluation of Sub-Expressions}
\label{sec:eval-sub-expr}

A common pattern in the following rules and expressions is that an
expression~$e$ contains sub-expressions, which need to be evaluated
before $e$'s top-level operator can be evaluated. For instance, for
the expression \lstinline|(5 + 1) / 3|, the division cannot be
evaluated until the sub-expression \lstinline|5 + 1| is fully
evaluated. Therefore, we require a way to obtain for a given
expression~$e$ the next sub-expression~$e'$ to be evaluated. However,
in \Rust, what this sub-expression is, depends on the context: If we
evaluate \rust{a[1 + 1]} to a value, then \enquote{\pv{a}} is the next
sub-expression. However, evaluating it in an \emph{assignment
  context}, e.g., \lstinline|a[1 + 1] = 4 / 2|, the behavior is
different: Because the left-hand side of an assignment is evaluated
not to values but to assignee values ($\AssignVal$), the next
sub-expression is \lstinline|1 + 1|.

For that reason, we also require the context of the expression in
addition to the next sub-expression to be evaluated. This allows us to
evaluate the next sub-expression and then replace the old
sub-expression with the result, leaving the rest of the expression
unaltered. For this, we introduce \emph{hole expressions}. Let
$\HoleExpr$ with typical element~$\mathit{he}$ be defined analogous
$\Expr$ but with a new base case, a \emph{hole}~$\exprHole$, similar
to holes in types and terms in functional
programming~\autocite{W2018,M1998}. We further require that
$\exprHole$ occurs exactly once per $\mathit{he}\in\HoleExpr$. For
example, \rust{$\exprHole$ / 3} is a hole expression. This allows us
to extract a sub-expression from an expression and substitute it
without altering the rest. We denote filling the hole in
$\mathit{he}\in\HoleExpr$ with expression~$e$ as
$\exprFill{\mathit{he}}{e}$. The result is an element of $\Expr$.

We now obtain the next sub-expression with the function
$\nextExprSym:\Expr\times2^\Expr\rightarrow\Expr\times\HoleExpr$,
which accurately models \Rust's evaluation order using the
sub-expression's context as second argument. For example,
$\nextExpr{\rust{(5 + 1) / 3}}{\Val}$ is expected to yield
$(\rust{5 + 1}, \rust{$\exprHole$ / 3})$, which allows us to evaluate
the next sub-expression \lstinline|5 + 1| to the value $6$ which can
then be inserted into the hole. We get
$\exprFill{(\rust{$\exprHole$ / 3})}{\rust{6}}=\rust{6 / 3}$. The
$\nextExprSym$ function must be defined such that for all
$e, e', \mathit{he}$, and $M$ such that
$\nextExpr{e}{M}=(e',\mathit{he})$, we have
$\exprFill{\mathit{he}}{e'}=e$.

The function $\nextExprSym$ is defined recursively and, as mentioned
before, it carries the context as its second argument. The recursion
stops when reaching an element of this set.  Typically it is the set
$\Val$ or $\AssignVal$.  By default we are the context $\Val$.  For
example, to evaluate the expression \lstinline|a[1 + 1] = 4 / 2|, we
determine the next sub-expression by
$\nextExpr{\rust{a[1 + 1] = 4 / 2}}{\Val}$. In this case, first
\lstinline|4 / 2| is evaluated, then \lstinline|a[1 + 1]|.  However,
when we have \lstinline|a[1 + 1] = 2|, we must consider the left-hand
side by calling $\nextExpr{\rust{a[1 + 1]}}{\AssignVal}$, because do
not want to evaluate \lstinline|a[1 + 1]| to a value but to get the
assignee value (in this case, it is a place) $\pv{a}[2]$. The latter
call yields $(\rust{1 + 1}, \rust{a[$\exprHole$]})$.

We need to define $\nextExprSym$ for all \Rust expressions. For
instance, the next sub-expression of the negation
expression~\rust{!$e$} is defined as follows: If $e$ also has a
further sub-expression that is to be evaluated, we return
it. Otherwise, if $e$ is not in $\Val$, it is the next sub-expression
and should be evaluated to a value. We always pass $\Val$ to the
recursive call, because negations cannot be part of an assignee
expression.
\[
  \nextExpr{\rust{!$e$}}{M} =
  \begin{cases}
    (\rust{!$e$}, \exprHole)&\text{if}~e\in \Val\\
    (e',\rust{!$\mathit{he}$}) &\text{if}~e\not\in\Val\wedge(e',\mathit{he})=\nextExpr{e}{\Val}
  \end{cases}
\]

The expression \rust{\kwIf~$e_1$ $e_2$ \kwElse~$e_3$} behaves
similarly. Because the branches are only conditionally evaluated
depending on $e_1$'s value, they should not be considered as
sub-expressions to be evaluated. Hence, we only consider the condition
and obtain the same cases and results as for negation.
\[
  \nextExpr{\rust{\kwIf~$e_1$ $e_2$ \kwElse~$e_3$}}{M} =
  \begin{cases}
    (\rust{\kwIf~$e_1$ $e_2$ \kwElse~$e_3$}, \exprHole)&\text{if}~e_1\in \Val\\
    (e_1',\rust{\kwIf~$\mathit{he}$ $e_2$ \kwElse~$e_3$}) &\text{if}~e_1\not\in\Val\wedge(e_1',\mathit{he})=\nextExpr{e_1}{\Val}
  \end{cases}
\]

An assignment \rust{$\mathit{ae}$ = $e$} is evaluated by first
evaluating the right-hand side~$e$ to a value and then the left-hand
side~$\mathit{ae}$ to an assignee value. Hence, the second case, where
$e$ is not a value and may have further sub-expressions, behaves as
above. If $e$ is a value, we treat $\mathit{ae}$ similarly, but denote
that $\mathit{ae}$ is to be evaluated to an assignee value. Observe
that $\mathit{ae}$ may contain sub-expression that must be evaluated
to values, e.g., \lstinline|1 + 1| in
\lstinline|a[1 + 1] = 4 / 2|. Because assignments occur only on the
right-hand side, we do not consider~$M$.
\[
  \nextExpr{\rust{$\mathit{ae}$ = $e$}}{M} =
  \begin{cases}
    (\rust{$\mathit{ae}$ = $e$},\exprHole)&\text{if}~e\in\Val\wedge\mathit{ae}\in\AssignVal\\
    (e',\rust{$\mathit{ae}$ = $he$}) &\text{if}~e\not\in\Val\wedge(e',\mathit{he})=\nextExpr{e}{\Val}\\
    (\mathit{ae}',\rust{$\mathit{he}$ = $e$})&\text{if}~e\in\Val\wedge\mathit{ae}\not\in\AssignVal\wedge(\mathit{ae}',\mathit{he})=\nextExpr{\mathit{ae}}{\AssignVal}
  \end{cases}
\]

The second parameter matters when the expression can be an assignee
expression, as in that case sub-expressions should not be considered
if they are \emph{places} rather than values. For the index expression
\rust{$e_1$[$e_2$]}, if we are in an assignment context, $e_1$ is not
a sub-expression if it is a place, e.g., \pv{x} or
$\pv{x}[1]$. However, the indexing expression~$e_2$ must be evaluated
to a value.
\[
  \nextExpr{\rust{$e_1$[$e_2$]}}{M} =
  \begin{cases}
    (\rust{$e_1$[$e_2$]},\exprHole)&\text{if}~e_1\in M \wedge e_2\in\Val\\
    (e'_1,\rust{$\mathit{he}$[$e_2$]}) &\text{if}~e_1\not\in M \wedge (e'_1,\mathit{he})=\nextExpr{e_1}{M}\\
    (e'_2,\rust{$e_1$[$\mathit{he}$]}) &\text{if}~e_1\in M \wedge e_2 \not\in\Val \wedge (e'_2,\mathit{he})=\nextExpr{e_2}{\Val}
  \end{cases}
\]

We do not enumerate all the other definitions for the remaining
expressions here. See \zcref{sec:add-def-next} for the full
definitions.

This helper function allows us to express the pattern of evaluating
sub-expressions concisely as a local evaluation rule. We evaluate the
first sub-expression one step, attach the trace generated by the
sub-expression, and rebuild the expression with the sub-expression
replaced by its continuation. \zcref[S]{eq:local-hole} considers only
the cases where $\nextExpr{e}{\Val}=(e', \mathit{he})$ and
$\mathit{he}\ne\exprHole$; for all other expressions, the rule is not
applicable.
\evalrule{hole}{
  \valD{e} & = & \{\concatTr{\tau}{\cont{\exprFill{\mathit{he}}{e''}}} \mid \concatTr{\tau}{\cont{e''}}\in\valD{e'}\} \quad \text{where}~\nextExpr{e}{\Val}=(e', \mathit{he})\wedge\mathit{he}\ne\exprHole
}

With this rule in place, the remaining rules can assume that all
sub-expressions are (assignment) values, simplifying the presentation.

\subsubsection{Arithmetic and Logic Operators}
\label{sec:arithm-logic-oper}

For clarity, we split the rules for operators, which are typically
overloaded via traits, into multiple definitions for different
types. For \lstinline|bool|, negation is defined as:
\evalrule{not-bool}{
  \valD{\rust{!$v$}} & = &
  \begin{cases}
    \{\concatTr{\mkTr{\sigma}}{\cont{\falseSem}}\} &\text{if}~v=\trueSem \\
    \{\concatTr{\mkTr{\sigma}}{\cont{\trueSem}}\} &\text{if}~v=\falseSem
  \end{cases}
}

We present bitwise negation only for the case where $e$ is of type
\lstinline|u8|; all signed and unsigned integer types are
analogous. Here, $\concrMinus$ represents the mathematical
subtraction operation over $\mathbb{Z}$ rather than the \Rust
operator.
\evalrule{not-u8}{
  \valD{\rust{!$v$}} & = & \{\concatTr{\mkTr{\sigma}}{\cont{2^8 \concrMinus v \concrMinus 1}}\} 
}

Finally, for all other types, negation is resolved by calling the
\lstinline|not| function of the \lstinline|Not| trait, as in
\zcref{eq:local-not-other}. The trait is in module
\lstinline|core::ops|, hence, we add the module to the call. We stress
that the rules are modular, as which rule is applied depends strictly
on the local expression and its type, not the wider context.
\evalrule{not-other}{
  \valD{\rust{!$v$}} & = & \valD{\rust{core::ops::\trait{Not}::not($v$)}}
}

\begin{example}\label{ex:local-neg}
  Assume the program in \zcref{lst:running}. Let $\sigma_1$ be the
  state $[\pv{b}\mapsto \falseSem]$. The local evaluation of \pv{b} in
  $\sigma_1$ yields
  $\valB{\sigma_1}{\pv{b}} =
  \{\concatTr{\mkTr{\sigma_1}}{\cont{\falseSem}}\}$.  Locally
  evaluating \rust{!$\falseSem$} in $\sigma_1$ gives
  $\{\concatTr{\mkTr{\sigma_1}}{\cont{\trueSem}}\}$. The evaluation of
  \rust{!\pv{b}} in \zcref{line:not-b} in $\sigma_1$ is
  $\valB{\sigma_1}{!\pv{b}} =
  \{\concatTr{\tau}{\cont{\rust{!$e$}}} \mid \concatTr{\tau}{\cont{e}}
  \in \valB{\sigma_1}{\pv{b}}\}$, which we simplify to
  $\{\concatTr{\mkTr{\sigma_1}}{\cont{\rust{!$\falseSem$}}}\}$.
\end{example}

Sign negation is similar to bitwise negation, in that it is
overloaded. We present the version for \lstinline|i8| in
\zcref{eq:local-neg-i8}; other integer types are similar. If the
result of the negation would be outside the type's range (in
\lstinline|i8|'s case, $[-128, 127]$), a panic occurs.
\evalrule{neg-i8}{
  \valD{\rust{-$v$}} & = &
  \begin{cases}
    \{\concatTr{\mkTr{\sigma}}{\cont{\concrMinus v}}\} &\text{if}~{-128} < v < 128 \\
    \{\concatTr{\mkTr{\sigma}}{\cont{\rust{\panic()}}}\} &\text{if}~v=-128
  \end{cases}
}

For all other types, the \lstinline|Neg| trait's function is invoked.
\evalrule{neg-other}{
  \valD{\rust{-$v$}} & = & \valD{\rust{core::ops::\trait{Neg}::neg($v$)}}
}

Addition is handled similarly to negation above: If the sub-values
$v_1,v_2$ can be added without exceeding the bounds of the type, we
continue with this value; otherwise we emit a panic. For non-integer
types, the \lstinline|Add| trait is used. \zcref[S]{eq:local-add}
shows the case for \lstinline|u8|, where ${\concrPlus}$ again stands
for the mathematical addition over~$\mathbb{Z}$.
\evalrule{add}{
  \valD{\rust{$v_1$ + $v_2$}} & = &
  \begin{cases}
    \{\concatTr{\mkTr{\sigma}}{\cont{v_1 \concrPlus v_2}}\} &\text{if}~v_1 \concrPlus v_2 \le 255\\
    \{\concatTr{\mkTr{\sigma}}{\cont{\rust{\panic()}}}\} &\text{if}~v_1 \concrPlus v_2 > 255
  \end{cases}
}

Compound assignment proceeds by invoking the appropriate function of
the \lstinline|AddAssign| trait. Observe that the left-hand side is a
place and the right-hand side a
value.
\evalrule{assign-add}{
  \valD{\rust{$p$ += $v$}} & = & \{\concatTr{\mkTr{\sigma}}{\cont{\rust{\trait{AddAssign}::add\_assign(\&\kwMut $p$, $v$)}}}\}
}

\subsubsection{References}
\label{sec:references}

References are either \emph{shared} (also called \enquote{immutable})
or \emph{mutable}. As discussed in \zcref{sec:basics}, we represent
shared references as $\sRefVal{p}{b}$ and mutable references as
$\mRefVal{p}{b}$, where $p$ is the borrowed place and $b$ is a borrow
identifier.

\zcref[S]{eq:local-borrow-sh,eq:local-borrow-mut} describe the
evaluation of borrowing from a place~$p$ to create a shared and
mutable reference, respectively. Note that
\zcref{eq:local-borrow-sh,eq:local-borrow-mut} create \emph{fresh}
borrow identifiers, so that two references created by different
expressions are distinct.
\evalrule{borrow-sh}{
  \valD{\rust{\&$p$}} & = & \{\concatTr{\mkTr{\sigma}}{\cont{\sRefVal{p}{b}}} \mid b \in \BId, b~\text{fresh}\}
}
\evalrule{borrow-mut}{
  \valD{\rust{\&\kwMut $p$}} & = & \{\concatTr{\mkTr{\sigma}}{\cont{\mRefVal{p}{b}}} \mid b \in \BId, b~\text{fresh}\}
}

\begin{example}\label{ex:local-mref}
  Continuing the running example, evaluating the mutable reference
  expression in \zcref{line:call-set-all} in a state
  $\sigma_2=[\pv{arr}\mapsto[8,8,8]]$ yields
  $\valD{\rust{\&\kwMut
      arr}}=\{\concatTr{\mkTr{\sigma_2}}{\cont{\mRefVal{\pv{arr}}{b}}}
  \mid b \in \BId\}$. Observe, that we do not evaluate $\pv{arr}$ to
  create the mutable reference: a variable is a fully evaluated place
  and the reference need not know the borrowed variable's value. The
  generated set is infinite if we assume $\BId$ to be infinite.
\end{example}

Dereferencing an expression (\rust{*$v$}) has multiple possible
evaluations. If $v$ is of a pointer type (\lstinline|&T| or
\lstinline|&mut T|), dereferencing evaluates to the referenced
value. However, this is only legal if \lstinline|T| implements the
\lstinline|Copy| trait.  \zcref[S]{eq:local-deref} handles this case.
\evalrule{deref}{
  \valD{\rust{*$v$}} & = & \{\concatTr{\mkTr{\sigma}}{\cont{\sigma(p)}}\} \quad\text{if}~v=\sRefVal{p}{b}\vee v=\mRefVal{p}{b}
}

If $v$'s type is not a pointer type, dereferencing is evaluated via
the \lstinline|Deref| trait. As the trait's function takes a shared
reference, we also add a borrow expression.
\evalrule{deref-trait}{
  \valD{\rust{*$v$}} & = &
    \{\concatTr{\mkTr{\sigma}}{\cont{\rust{*core::ops::\trait{Deref}::deref(\&$v$)}}}\}
}

\subsubsection{Assignments}
\label{sec:assignments}

Assignments encompass a range of possible evaluations depending on the
left-hand side (the \enquote{assignee expression}). We split them into
different rules for clarity.  \zcref[S]{eq:local-assign-underscore}
handles the case where we assign to \rust{\_}, i.e., to nowhere. This
assignment does not affect the state and evaluates to the unit value.
\evalrule{assign-underscore}{
  \valD{\rust{\_ = $v$}} & = & \{\concatTr{\mkTr{\sigma}}{\cont{\rust{()}}}\}
}
Assigning to a program variable simply updates the state at this
variable.
\evalrule{assign-pv}{
  \valD{\rust{x = $v$}} & = & \{\concatTr{\consTr{\update{\sigma}{\pv{x}}{v}}{\mkTr{\sigma}}}{\cont{\rust{()}}}\}
}

An assignment to the field on an array of type \lstinline|[T; N]|,
where \lstinline|T| is the array's element type and \lstinline|N| the
array's size, has two cases: If the index is less than the array's
length, the state is updated at that location; otherwise, a panic
occurs. The index must be of type \lstinline|usize|, hence it cannot
be negative (see \zcref{eq:local-assign-index}). Index assignments to
other types typically go through the \lstinline|IndexMut| trait. In
\zcref{eq:local-assign-index}, $v_1,v_2\in\Val$ and $p\in\Place$ is a
place of type \lstinline|[T; N]|, i.e., a variable of the correct
type, a field access like \lstinline|x.f|, or an inner array
\lstinline|x.f[2]|. \zcref[S]{eq:local-assign-field} work similarly
for struct and tuple fields, except no panic can occur.
\evalrule{assign-index}{
  \valD{\rust{$p$[$v_1$] = $v_2$}} & = &
  \begin{cases}
    \{\concatTr{\consTr{\update{\sigma}{p[n]}{v}}{\mkTr{\sigma}}}{\cont{\rust{()}}}\} &\text{if}~n<\pv{N}\\
    \{\concatTr{\mkTr{\sigma}}{\cont{\rust{\panic()}}}\} &\text{otherwise}
  \end{cases}
}
\evalrule{assign-field}{
  \valD{\rust{$p$.$\mathit{ident}$ = $v$}} & = & \{\concatTr{\consTr{\update{\sigma}{p.\mathit{ident}}{v}}{\mkTr{\sigma}}}{\cont{\rust{()}}}\}
}

The dereferencing expression evaluates a reference to a place~$p$
which can then be assigned a value by the expression \rust{*$r$ =
  $v$}, where $r$ is a reference. Because assigning to a dereferenced
expression is only permitted when the reference is a mutable
reference, we know that $r = \mRefVal{p}{b}$ for some
$p\in\Place, b\in\BId$. \zcref[S]{eq:local-assign-deref} simply updates
the corresponding place.
\evalrule{assign-deref}{
  \valD{\rust{*$\mRefVal{p}{b}$ = $v$}} & = & \{\concatTr{\consTr{\update{\sigma}{p}{v}}{\mkTr{\sigma}}}{\cont{\rust{()}}}\}
}

Next, we deal with the \emph{destructuring assignments}, where parts
of complex values can be assigned to multiple locations at once. Since
the assignment is correctly typed, we know that the value on the right
fits the assignee value on the left. Rules
\zcref[noname]{eq:local-assign-destr-array,eq:local-assign-destr-tuple,eq:local-assign-destr-struct,eq:local-assign-destr-tuple-struct}
deal with arrays, tuples, structs, and tuple structs,
respectively. All are evaluated by introducing a block with
assignments to the individual assignee values. Let
$T_\text{assign}:=\{\concatTr{\mkTr{\sigma}}{\cont{\rust{\{$\mathit{av}_1$=$v_1$;
      $\ldots$; $\mathit{av}_n$=$v_n$\}}}}\}$ in the following, as all
these assignments lead to the same set of traces.
\begin{align}
  \valD{\rust{[$\mathit{av}_1$, $\ldots$, $\mathit{av}_n$] = [$v_1,\ldots,v_n$]}} & = T_\text{assign} \label{eq:local-assign-destr-array} \\
  \valD{\rust{($\mathit{av}_1$, $\ldots$, $\mathit{av}_n$) = ($v_1,\ldots,v_n$)}} & = T_\text{assign} \label{eq:local-assign-destr-tuple} \\
  \valD{\rust{x \{$\mathit{id}_1$:$\mathit{av}_1$, $\ldots$, $\mathit{id}_n$:$\mathit{av}_n$\} = x \{$\mathit{id}_1$:$\mathit{v}_1$, $\ldots$, $\mathit{id}_n$:$\mathit{v}_n$\}}} & = T_\text{assign} \label{eq:local-assign-destr-struct} \\
  \valD{\rust{x($\mathit{av}_1$, $\ldots$, $\mathit{av}_n$) = x($v_1,\ldots,v_n$)}} & = T_\text{assign} \label{eq:local-assign-destr-tuple-struct}
\end{align}

\begin{example}\label{ex:local-assign}
  We demonstrate the state changes incurred by an assignment by
  continuing \zcref{ex:local-neg}: The assignment \lstinline|c = !b|
  (\zcref{line:not-b}) in state~$\sigma_1$ gives us the traces
  $\{\concatTr{\tau}{\cont{\rust{c =
        $e$}}} \mid \concatTr{\tau}{\cont{e}} \in
  \valD{\rust{!b}}\}=\{\concatTr{\mkTr{\sigma_1}}{\cont{\rust{c =
        !$\falseSem$}}}\}$. The assignment \lstinline|c = true|, the
  result of multiple local evaluations, yields:
  $\{\concatTr{\consTr{\update{\sigma_1}{\pv{c}}{\trueSem}}{\mkTr{\sigma_1}}}{\cont{\rust{()}}}\}$.
\end{example}

\subsubsection{Compound Types}
\label{sec:compound-types}

Array expressions are evaluated by evaluating their elements
left-to-right. This is entirely handled by the rules in
\zcref{sec:eval-sub-expr}. The same approach works for tuples, struct,
and enum expressions.

An array repeat expression creates an array value with the inner value
repeated $n$ times.
\evalrule{repeat}{
  \valD{\rust{[$v$; $n$]}} & = & \{\concatTr{\mkTr{\sigma}}{\cont{\rust{[$\overbrace{v, \ldots, v}^{n~\text{times}}$]}}}\}
}

To access a struct field or tuple field, we simply take the field's
value from the state.
\evalrule{field}{
  \valD{\rust{$v$.$\mathit{ident}$}} & = & \{\concatTr{\mkTr{\sigma}}{\cont{v'}}\} \quad\text{where}~v=\rust{x \{$\mathit{ident}: v', \ldots$\}}
}

Accessing the index of an array is similar to a struct field, besides
the necessary bounds check: if the index is at least the array's
length, a panic occurs. Since we know that $v'$ is of type
\lstinline|usize|, we need not check for negative values. We only show
the case for arrays; all other types are translated to a trait
function call.
\evalrule{index}{
  \valD{\rust{$v$[$v'$]}} & = &
  \begin{cases}
    \{\concatTr{\mkTr{\sigma}}{\cont{v_{v'}}}\} &\text{if}~v=[v_0, \ldots, v_n], v' \le n\\
    \{\concatTr{\mkTr{\sigma}}{\cont{\rust{\panic()}}}\} &\text{otherwise}
  \end{cases}
}

\subsubsection{Basic Control Flow}
\label{sec:control-flow}

Evaluating a conditional expression such as \rust{\kwIf $v$ $e_1$
  \kwElse $e_2$} is straightforward: If condition $v$ is the true
value, we continue with the then expression $e_1$; if $v$ is false, we
continue with $e_2$.
\evalrule{if-else}{
  \valD{\rust{\kwIf $v$ $e_1$ \kwElse $e_2$}} & = &
  \begin{cases}
    \{\concatTr{\mkTr{\sigma}}{\cont{e_1}}\} &\text{if}~v=\trueSem\\
    \{\concatTr{\mkTr{\sigma}}{\cont{e_2}}\} &\text{if}~v=\falseSem
  \end{cases}
}

\begin{example}\label{ex:local-if}
  When evaluating the \lstinline|if| in \zcref{line:if} in a
  state~$\sigma_3=[\pv{n}\mapsto 4]$, we evaluate the condition. Once
  the condition is simplified to $\falseSem$, we have
  $\valD{\rust{\kwIf
      $\falseSem$ \{1\} \kwElse \{\kwLet left = \ldots\}}} =
  \{\concatTr{\mkTr{\sigma_3}}{\cont{\rust{\{\kwLet left =
        \ldots\}}}}\}$.
\end{example}

A match expression is evaluated step-wise. First, we consider the
first arm. This arm contains a pattern~$\mathit{pat}$ and a
body~$e$. If the pattern's structure fits the matched value (the
\emph{scrutinee}), we want to continue with this arm. However, pattern
matching may introduce new variable bindings: A match arm like
\lstinline|x => {x + 2}|, where \pv{x} is some integer variable, does
not compare the scrutinee's value with the value of an existing
variable~\pv{x}. Instead, a new variable~\pv{x} is added to the scope
and it is bound to the scrutinee's value.

To accommodate this binding pattern, we use the partial function
$\mathrm{patMatch}_\sigma: \mathit{Pat} \times \Val \rightharpoonup
\mathit{State}$, which takes a pattern and a value and returns an
extension of state~$\sigma$, that is, new values may be introduced but
no existing values changed.

We cover a few important cases of the function
$\mathrm{patMatch}_\sigma$ below. The wildcard pattern
\enquote{\rust{\_}} always matches, values (i.e., literal patterns)
match if they are exactly the scrutinee's value, binding patterns
(\rust{x}) introduce new variables, and the alternative pattern
considers both sub-patterns and returns the first successful
match. The other patterns follow the same approach. Note that we use
the symbol \enquote{$\undefSym$} to denote undefinedness.
\begin{equation*}
  \begin{array}{@{}r@{\,}c@{\,}l@{}}
    \patMatch{\rust{\_}}{v} &=& \sigma\\
    \patMatch{\rust{$v_1$}}{v_2} &=&
                                     \begin{cases}
                                       \sigma &\text{if}~v_1=v_2\\
                                       \undefSym &\text{otherwise}
                                     \end{cases}\\
    \patMatch{\rust{x}}{v} &=& \update{\sigma}{\pv{x}}{v}\\
    \patMatch{\rust{$\mathit{pat}_1$ | $\mathit{pat}_2$}}{v}  &=&
                                                               \begin{cases}
                                                                 \sigma_1 &\text{if}~\patMatch{\mathit{pat}_1}{v}=\sigma_1\\
                                                                 \sigma_2 &\text{if}~\patMatch{\mathit{pat}_1}{v}\undefSym\wedge\patMatch{\mathit{pat}_2}{v}=\sigma_2\\
                                                                 \undefSym &\text{otherwise}
                                                               \end{cases}\\
                            &\ldots&
  \end{array}
\end{equation*}

If the first pattern does not match, we remove its arm from the match
expression and continue with the rest. The \Rust compiler ensures that
the match arms are exhaustive, hence we need not cover the case of all
arms being removed.
\evalrule{match}{
  \valD{\rust{\kwMatch $v$ \{$\mathit{pat}$ => $e$, $\many{\mathit{arm}}$\}}} 
&=& \begin{cases}
    \{\concatTr{\consTr{\sigma'}{\mkTr{\sigma}}}{\cont{e}}\} &\text{if}~\patMatch{\mathit{pat}}{v}=\sigma'\\
    \{\concatTr{\mkTr{\sigma}}{\cont{\rust{\kwMatch $v$ \{$\many{\mathit{arm}}$\}}}}\} &\text{if}~\patMatch{\mathit{pat}}{v}\undefSym
  \end{cases}
}

\subsubsection{Function and Method Calls}
\label{sec:funct-meth-calls}

Evaluating functions and methods introduces additionally complexity
that must be managed by the semantics. To model the lifecycle of a
call---invocation, execution, computation, and reaction to the
computation---it is helpful to directly express these points in the
semantics and the generated traces. Hence the \emph{events} we allow
in our traces. We make use of four call events in this section, one
for each point of the lifecycle of a call. These four events follow
the four event semantics of \textcite{DDJO2012}. Including these
events directly in the trace permits easy, declarative restrictions on
the order of events to ensure correct evaluation order.

To evaluate calls, we emit an \emph{invocation event} with the
function identifier, a unique but arbitrary call
identifier~$i\in\CId$, where $\CId$ is the set of call identifiers,
and the call arguments. We discuss below, how we ensure uniqueness of
the call identifiers. Because the trace composition rules take care of
actually executing the called function, the local evaluation rule
needs to put a placeholder in the continuation marker that denotes
that this expression will evaluate to the value returned by the call
with identifier~$i$. For that, we introduce the \emph{get} expression
\rust{?$i$} with $i\in\CId$. This is necessary to mark that this
expression cannot be resolved until the called function is completely
executed and to access the called function's return value.
\evalrule{call}{
  \valD{\rust{$v$($\many{v}$)}}   &=& \{\concatTr{\consTr{\sigma}{\consTr{\invEv{v}{i}{\many{v}}}{\mkTr{\sigma}}}}{\cont{\rust{?$i$}}} \mid i\in\CId\} 
}

Method calls, e.g., \lstinline|x.m(1 + 1)| behave similarly to regular
function calls but we make use of the method lookup function
$\mathrm{method2Fn}(v,m)$, which yields the function name for the
method~$m$ on callee~$v$. Assume in \zcref{eq:local-method-call} for
brevity that $\mathit{foo}=\mathrm{method2Fn}(v,m)$. Observe that the
callee is added as the first argument of the call identifier. This
matches \Rust's treatment of methods~\autocite{RustRef}.
\evalrule{method-call}{
  \valD{\rust{$v$.$m$($\many{v}$)}} & = &
    \{\concatTr{\consTr{\sigma}{\consTr{\invEv{\mathit{foo}}{i}{v,\many{v}}}{\mkTr{\sigma}}}}{\cont{\rust{?$i$}}}\mid i \in\CId\} 
}

To evaluate a function's body, we extend $\valDnoargs$ to a function
declaration, emitting an \emph{invocation reaction event}, introducing
symbolic variables~$\many{\pv{X}}$ for the parameter's values, and
continuing with the body where the parameters $\many{\pv{x}}$ have
been substituted with $\many{\pv{X}}$. We use the notation
$e\subst{\many{\pv{x}}}{\many{\pv{X}}}$ for syntactically substituting
all $\pv{x} \in \many{\pv{x}}$ with the corresponding
$\pv{X} \in \many{\pv{X}}$ in $e$.
\evalrule{fn-body}{
  \valD{\rust{\kwFn foo($\many{\pv{x}}$) $e$}} &=& \{\concatTr{\consTr{\sigma[\many{\pv{X}\mapsto\ast}]}{\consTr{\invREv{\pv{foo}}{i}{\many{\pv{X}}}}{\mkTr{\sigma}}}}{\cont{e\subst{\many{\pv{x}}}{\many{\pv{X}}}}} \mid i\in\CId,\many{\pv{X}}~\text{fresh}\}
}

\begin{example}\label{ex:local-call}
  Consider the call in \zcref{line:init-call-fac} and assume that the
  callee is already evaluated to the function
  identifier~$\mathit{fac}$. Both callee and argument are already
  values, hence, we have
  $\valB{\sigma_1}{\rust{$\mathit{fac}$(4)}}=\{\concatTr{\consTr{\sigma_1}{\consTr{\invEv{\mathit{fac}}{i}{1}}{\mkTr{\sigma_1}}}}{\cont{\rust{?$i$}}}
  \mid i\in\CId\}$. The get expression \rust{?$i$} creates a
  computation reaction event and we must ensure, as we do below, that
  this is only evaluated directly after the function call with id~$i$
  is completed. Evaluation of the body gives us
  \[\valB{\sigma_1}{\rust{\kwFn fac(n)
      $e$}}=\{\concatTr{\consTr{\update{\sigma_1}{\pv{X}}{\ast}}{\consTr{\invREv{\pv{fac}}{i}{\pv{X}}}{\mkTr{\sigma}}}}{\cont{e\subst{\pv{n}}{\pv{X}}}}
  \mid i\in\CId,\pv{X}~\text{fresh}\}\] where $e$ is the body of
\lstinline|fac| in \zcref{lst:running}.
\end{example}

An expression \rust{\kwReturn~$v$} emits a \emph{return event} with
the returned value. The return expression evaluates to the unit value,
because the returned value is not to be considered for the surrounding
context. Instead, the event is used to transmit the return value to
the caller.
\evalrule{return}{
  \valD{\rust{\kwReturn $v$}} & = & \{\concatTr{\consTr{\sigma}{\consTr{\retEv{v}}{\mkTr{\sigma}}}}{\cont{\rust{()}}}\}
}

The get expression emits a \emph{computation reaction event} and
introduces a fresh symbolic variable~\pv{X} which represents the value
returned by call~$i$.
\evalrule{call-get}{
  \valD{\rust{?$i$}} = \{\concatTr{\consTr{\update{\sigma}{\pv{X}}{\ast}}{\consTr{\compREv{i}{\pv{X}}}{\mkTr{\sigma}}}}{\cont{\pv{X}}}\mid\pv{X}~\text{fresh}\}
}

\begin{example}\label{ex:local-comp}
  Combining the results of
  \zcref{ex:sync-run,ex:local-neg,ex:local-mref,ex:local-assign,ex:local-call},
  we obtain the following traces:
  \begin{align*}
    \valB{\sigma_0}{\rust{b = \rsFalse}} =\,& \{\concatTr{\consTr{\sigma_1}{\mkTr{\sigma_0}}}{\cont{\rust{()}}}\}\\
    \valB{\sigma_1}{\rust{c = !b}} =\,& \{\concatTr{\mkTr{\sigma_1}}{\cont{\rust{c = !$\falseSem$}}}\}\\
    \valB{\sigma_1}{\rust{c = !$\falseSem$}} =\,& \{\concatTr{\mkTr{\sigma_1}}{\cont{\rust{c = $\trueSem$}}}\}\\
    \valB{\sigma_1}{\rust{c = $\trueSem$}} =\,& \{\concatTr{\consTr{\sigma_2}{\mkTr{\sigma_1}}}{\cont{\rust{()}}}\}\\
    \valB{\sigma_2}{\rust{f = fac(4)}} =\,& \{\concatTr{\consTr{\sigma_2}{\consTr{\invEv{\pv{fac}}{i}{4}}{\mkTr{\sigma_1}}}}{\cont{\rust{f = ?$i$}}} \mid i \in \CId\}\\
    \valB{\sigma_3}{\rust{arr = [8; 3]}} =\,& \{\concatTr{\mkTr{\sigma_3}}{\cont{\rust{arr = [8, 8, 8]}}}\}\\
    \valB{\sigma_3}{\rust{arr = [8, 8, 8]}} =\,& \{\concatTr{\consTr{\sigma_4}{\mkTr{\sigma_3}}}{\cont{\rust{()}}}\}\\
    \valB{\sigma_4}{\rust{set\_all(\&\kwMut~arr, f)}} =\,& \{\concatTr{\mkTr{\sigma_4}}{\cont{\rust{set\_all(\mRefVal{\pv{arr}}{b}, f)}}} \mid b\in\BId\}
  \end{align*}
  with
  $\dom(\sigma_0)=\emptyset,
  \sigma_1=\update{\sigma_0}{\pv{b}}{\falseSem},
  \sigma_2=\update{\sigma_1}{\pv{c}}{\trueSem},
  \sigma_3=\update{\sigma_2}{\pv{f}}{24},
  \sigma_4=\update{\sigma_3}{\pv{arr}}{[8, 8, 8]}$. Clearly, the
  resulting trace of this series of computations must be of the form
  \begin{align*}
    \consTr{\sigma_4}{\consTr{\sigma_3}{\consTr{\cdots}{\consTr{\invEv{\pv{fac}}{i}{4}}{\consTr{\sigma_2}{\consTr{\sigma_1}{\mkTr{\sigma_0}}}}}}}\ .
  \end{align*}
  What is missing so far is
  \begin{enumerate*}[(1)]
  \item local rules dealing with variable declarations as in
    \zcref{lst:running}, not just assignments,
  \item local rules dealing with blocks, and
  \item a way to compose these individual local rules into one single
    trace.
  \end{enumerate*}
  We deal with the first two items immediately and return to the third
  one in \zcref{sec:trace-comp-rules}.
\end{example}

\subsubsection{Block Expressions}
\label{sec:block-expressions}

When a block is fully evaluated and, hence, removed from the
continuation, we want to \emph{drop} all local variables this block
introduced, i.e., remove them from the state.  This requires
bookkeeping of which variables are introduced by a block. Thus, we add
a syntactic extension called \emph{scoped block expressions}:
\rust{\rib{$\overline{\pv{x}}$}{$\overline{s}$ $e$}} stands for the
block expression \rust{\{$\overline{s}$ $e$\}}, which has so far
introduced the local variables~$\overline{\pv{x}}$. For brevity, we
also write \rust{\{$\overline{s}$ $e$\}} for the scoped block
expression \rust{\rib{$\emptyset$}{$\overline{s}$ $e$}}.

Expressions and statements in blocks are also evaluated
piecemeal. Declarations with initial value are split into the
declaration and the assignment. For simplicity, we only show rules for
declaring variables.
\evalrule{let-assign}{
  \valD{\rust{\rib{$\overline{\pv{x}}$}{\kwLet y = $e_1$; $\overline{s}$ $e_2$}}} & = & \valD{\rust{\rib{$\overline{\pv{x}}$}{\kwLet y; y = $e_1$; $\overline{s}$ $e_2$}}}
}

A declaration of program variable~\pv{y} is evaluated by adding~\pv{y}
to the declared variables in the block. Additionally, to avoid name
clashes in our state caused by, e.g., evaluating two function bodies
in parallel, we rename~\pv{y} to a fresh program variable~$\pv{y}'$.
As there are potentially infinitely many fresh~$\pv{y}'$,
\zcref{eq:local-let} may introduce infinitely many traces; which of
these is chosen is covered in \zcref{sec:trace-comp-rules}.
\evalrule{let}{
  \valD{\rust{\rib{$\overline{\pv{x}}$}{\kwLet y; $\overline{s}$ $e$}}} & = & \{\concatTr{\mkTr{\sigma}}{\cont{\rust{\rib{$\overline{\pv{x}},\pv{y}$}{$\overline{s}$ $e$}}\subst{\pv{y}}{\pv{y}'}}} \mid \pv{y}'~\text{fresh}\}
}

The first expression in a block is either a value or a complex
expression. The latter case is handled by
\zcref{sec:eval-sub-expr}. In the former case, we can remove it as in
\zcref{eq:local-block-expr-stmt}.
\evalrule{block-expr-stmt}{
  \valD{\rust{\rib{$\overline{\pv{x}}$}{$v$; $\overline{s}$ $e$}}} & = & \{\concatTr{\mkTr{\sigma}}{\cont{\rust{\rib{$\overline{\pv{x}}$}{$\overline{s}$ $e$}}}}\}
}

A block containing only the final value~$v$ is removed, all variables
introduced by the block are dropped, i.e., removed from the state, and
we continue with $v$ only.
\evalrule{block-expr}{
  \valD{\rust{\rib{$\overline{\pv{x}}$}{$v$}}} & = & \{\concatTr{\consTr{(\sigma\setminus\many{\pv{x}})}{\mkTr{\sigma}}}{\cont{v}}\} 
}

Panics, in turn, are handled by emitting a \emph{panic event} and
leaving the unit value behind. Abrupt termination is modeled during
composition.
\evalrule{panic}{
  \valD{\rust{\panic()}} & = & \{\concatTr{\consTr{\sigma}{\consTr{\panicEv{}}{\mkTr{\sigma}}}}{\cont{\rust{()}}}\}
}

\subsubsection{Loops}
\label{sec:loops}

The more complex control flow options of a loop, including
\lstinline|break| and \lstinline|continue| prove to be a challenge, as
they require the context of the surrounding loops, which complicates
control flow, especially when labeled jumps are involved.  We model
loops by extending the syntax with \emph{loop frames,} taking
inspiration from \emph{loop scopes} in deductive program
verification~\autocite{SW2017}. These frames permit tracking of the
surrounding loop context while leaving the original loop---to which
the program might jump again---unaltered. A loop frame
$\loopFrame{l}{e}$ states that we are evaluating the loop with
label~$l$ and that the current expression to be evaluated of the
loop's body is~$e$. Once the loop body is fully evaluated, the loop is
evaluated from the beginning again. \zcref[S]{fig:loop} demonstrates
the possible relevant control flow options in a loop. Breaks jump out
of the loop addressed by its label; continues jump to the next
iteration of its corresponding loops. Note that the loop with label
\lstinline|'l1| has an integer value (\pv{i}, once the corresponding
\lstinline|break| in \zcref{line:break-l2} is executed) and the loop
\lstinline|'l2| has the unit value \lstinline|()|.

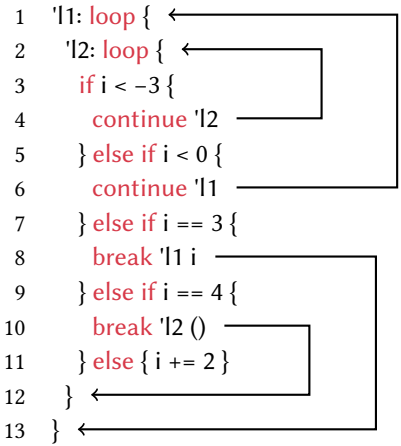
\begin{figure}[h]
  \centering
\begin{lstlisting}[numbers=left]
'l1: loop {%\tikz[remember picture] \node [] (l1) {};%
  'l2: loop {%\tikz[remember picture] \node [] (l2) {};%
    if i < -3 {
      continue 'l2%\tikz[remember picture] \node [] (c2) {};%
    } else if i < 0 {
      continue 'l1%\tikz[remember picture] \node [] (c1) {};%
    } else if i == 3 {
      break 'l1 i%\tikz[remember picture] \node [] (bl1) {};% %\label{line:break-l2}%
    } else if i == 4 {
      break 'l2 ()%\tikz[remember picture] \node [] (bl2) {};%
    } else { i += 2 }
  }%\tikz[remember picture] \node [] (el2) {};%
}%\tikz[remember picture] \node [] (el1) {};%
\end{lstlisting}
  \begin{tikzpicture}[remember picture, overlay]
    \node [right = 4cm of el1] (el1p4) {};
    \node [right = 2cm of bl1] (bl1p2) {};
    \draw[->, thick] (bl1) -- (bl1p2.center) -- ($(el1p4)!(bl1p2.center)!(el1)$) -- (el1);
    \node [right = 4cm of el2] (el2p4) {};
    \node [right = 1cm of bl2] (bl2p1) {};
    \draw[->, thick] (bl2) -- (bl2p1.center) -- ($(el2p4)!(bl2p1.center)!(el2)$) -- (el2);
    \node [right = 4cm of l2] (l2p4) {};
    \node [right = 1cm of c2] (c2p1) {};
    \draw[->, thick] (c2) -- (c2p1.center) -- ($(l2p4)!(c2p1.center)!(l2)$) -- (l2);
    \node [right = 4cm of l1] (l1p4) {};
    \node [right = 2cm of c1] (c1p2) {};
    \draw[->, thick] (c1) -- (c1p2.center) -- ($(l1p4)!(c1p2.center)!(l1)$) -- (l1);
  \end{tikzpicture}
  \Description{Two nested loops. The inner loop contains a long
    if-else-chain with continues and breaks to both of the
    loops. Arrows denote the different possible jumps from the breaks
    and continues.}
  \caption{Possible control-flow paths of a loop expression in \Rust.}
  \label{fig:loop}
\end{figure}

To start a loop's evaluation, it is translated to a loop frame.
\evalrule{loop}{
  \valD{\rust{'$l$: \kwLoop $e$}} & = & \valD{\loopFrame{l}{e}}
}

In \Rust, \lstinline|break| expressions contain an (optional)
expression. These expressions mark the value of the loop or block to
which the \lstinline|break|'s label refers. A \lstinline|loop|
expression's value thus depends on the \lstinline|break| expression
contained therein.

A \lstinline|break| with value $v$ in a loop is resolved in one of two
ways: If the label matches the active loop frame's label, we continue
with the expression.  Otherwise, we replace the entire frame with the
break expression. The following rules consider a break expression in
the context of a loop frame and a block. The latter may contain
additional statements and an expression after the
\lstinline|break|. These are discarded once the \lstinline|break| is
evaluated.
\evalrule{break}{
  \valD{\loopFrame{l_1}{\rust{\rib{$\many{\pv{x}}$}{\kwBreak~'$l_2$ $v$; $\many{s}$; $e$}}}} 
  &=&\begin{cases}
    \{\concatTr{\mkTr{\sigma}}{\cont{\rust{\rib{$\many{\pv{x}}$}{$v$}}}}\} &l_1=l_2\\
    \{\concatTr{\mkTr{\sigma}}{\cont{\rust{\rib{$\many{\pv{x}}$}{\kwBreak~'$l_2$ $v$}}}}\} &l_1\ne l_2
  \end{cases}
}

When we encounter a \lstinline|continue| expression with label $l_2$
in a loop frame, we evaluate the loop with label $l_2$, using the
function \enquote{$\getloopSym$,} which maps labels to loops. We
maintain the surrounding block to maintain the scope, which preserves
proper dropping of the declared variables.
\evalrule{continue}{
  \valD{\loopFrame{l_1}{\rust{\rib{$\many{\pv{x}}$}{\kwContinue '$l_2$; $\many{s}$; $e$}}}} & = & \{\concatTr{\mkTr{\sigma}}{\cont{\rust{\{\rib{$\many{\pv{x}}$}{()}; $\getloop{'l_2}$\}}}}\}
}

If the loop frame's expression is not a continue or break, we
re-insert the loop of the active label, repeating execution. The value
of the loop's body is ignored, as only break expressions affect the
loop's value.
\evalrule{loop-step}{
  \valD{\loopFrame{l}{v}} & =&
  \{\concatTr{\mkTr{\sigma}}{\cont{\getloop{'l}}}\}
}

\begin{example}
  Evaluating a loop expression
  \lstinline|'l loop { if n == 0 { break 'l true; } else {()} n -= 1; }| in
  state $\sigma_2=[\pv{n}\mapsto 1]$ leads to
  $\valB{\sigma_2}{\rust{'l \kwLoop
      $\ldots$}}=\valB{\sigma_2}{\loopFrame{\pv{l}}{\rust{\{\kwIf
        $\ldots$\}}}}$, which we further evaluate to
  $\{\concatTr{\mkTr{\sigma_2}}{\cont{\loopFrame{\pv{l}}{\rust{\{\kwIf 1
          == 0 $\ldots$\}}}}}\}$. Evaluating
  $\valB{\sigma_3}{\loopFrame{\pv{l}}{\rust{\{\kwBreak~'l \rsTrue;\}}}}$,
  where $\sigma_3=\update{\sigma_2}{\pv{n}}{0}$, yields
  $\{\concatTr{\mkTr{\sigma_3}}{\cont{\rust{\{\rsTrue\}}}}\}$.
\end{example}

\subsection{Well-Formedness of Traces}
\label{sec:well-formedness}

Without any restrictions, the traces generated by our semantics do not
necessarily describe possible program behavior, as they encompass
traces with undesired behavior, e.g., a function body evaluation where
the parameters do not match the call arguments. The well-formedness
predicate $\mathit{wf}$ restricts the set of generated traces to
ensure a certain ordering and matching between events. This predicate
is defined over shining traces and hence defines additional
restrictions on top of \zcref{def:wf-trace}. We define the
well-formedness predicate recursively: An empty trace is always
well-formed, as is the trace $\mkTr{\sigma}$. States do not influence
the well-formedness, and a panic event can only occur as the final
event and it is followed by exactly one state.
\noindent\begin{equation*}
  \begin{array}{@{}r@{\,}c@{\,}l@{}}
    \wf{\epsilon} &=& \mathit{true} \\
    \wf{\consTr{\sigma}{\epsilon}} &=& \mathit{true} \\
    \wf{\consTr{\sigma}{\consTr{t}{\sh}}} &=& \wf{\consTr{t}{\sh}} \wedge \panicEv\not\in\sh\\
    \wf{\consTr{\panicEv}{\sh}} &=& \wf{\sh}
  \end{array}
\end{equation*}

Every invocation event must have a unique call identifier. For the
invocation reaction, we require that the last event prior is an
invocation event with matching function name, call id, and
arguments. This requirement also implies that each invocation is
reacted to at most once.
\noindent\begin{align}
  \wf{\consTr{\invEv{\mathit{foo}}{i}{\many{v}}}{\sh}} &\,= \wf{\sh} \wedge \nexists\mathit{foo}',\many{v}'.\ \invEv{\mathit{foo}'}{i}{\many{v}'}\in\sh  \label{eq:wf-invEv}\\
  \wf{\consTr{\invREv{\mathit{foo}}{i}{\many{v}}}{\sh}} &\,= \wf{\sh} \wedge \big(\invEv{\mathit{foo}}{i}{\many{v}}=\lastEv(\sh)\big) \label{eq:wf-invREv}
\end{align}

No restrictions are placed on the return event, while the computation
reaction event---similarly to the invocation's case---must follow
immediately its corresponding return event.
\noindent\begin{align}
  \wf{\consTr{\retEv{v}}{\sh}} &= \wf{\sh} \label{eq:wf-retEv} \\
  \wf{\consTr{\compREv{\mathit{foo}}{v}}{\sh}} &= \wf{\sh} \wedge \big(\compEv{\mathit{foo}}{v}=\lastEv(\sh)\big) \label{eq:wf-compREv}
\end{align}

\begin{example}
  The trace
  $\consTr{\sigma_0}{\consTr{\invEv{\pv{foo}}{\mathit{cid}_0}{3}}{\consTr{\sigma_0}{\consTr{\panicEv}{\mkTr{\sigma_0}}}}}$
  is not well-formed, as the panic event is succeeded by an invocation
  event. However, the trace
  $\consTr{\update{\sigma_0}{\pv{x}}{3}}{\consTr{\invREv{\pv{foo}}{\mathit{cid}_0}{3}}{\consTr{\sigma_0}{\consTr{\invEv{\pv{foo}}{\mathit{cid}_0}{3}}{\mkTr{\sigma_0}}}}}$
  is well-formed as the invocation reaction event matches the
  invocation event and follows it immediately.
\end{example}

\subsection{Trace Composition Rules}
\label{sec:trace-comp-rules}

So far, our semantics provides the local, symbolic traces for each
expression and a definition of \enquote{correct} (well-formed) traces.
The next step in an \ac{LAGC} semantics is to compose together the
local traces generated in \zcref{sec:local-evaluation} into
well-formed, shining traces (see \zcref{def:wf-trace}). We define
\emph{trace composition rules}, which govern how traces are combined.

Trace composition rules are transition rules and define the relation
$\to$ between two \emph{configurations}. In general, a configuration
consists of a shining trace and some data structures over continuation
markers. The former is the trace generated so far, the latter
represents in some form the remaining program to evaluate. Thus, the
configuration expresses the global state of the program evaluation as
well as its history in the form of the trace.

In the synchronous \Rust case, we aim to model the call stack created
by function and method calls. Hence, we choose as configuration
\[\sh,\left[\many{k}\right]\] where $\many{k}$ are continuations and
the left-most continuation is the top element.

All composed traces must be well-formed according to
$\mathit{wf}$. This restricts both the local traces that may be chosen
and the applicability of rules.

Assume an expression~$e$ and a shining trace~$\sh$. The trace
composition \zcref{eq:rule-sync-base} defines the derivation from the
configuration $\sh,[\cont{e},\many{k}]$ to
$\chopTrSem{\sh}{\rho(\tau)},[\cont{e'},\many{k}]$, where $e$ is the
expression in the top continuation marker of the continuation
stack. We take the last state of $\sh$ to evaluate $e$ one step
further. This yields a set of continuation traces, from which we pick
one trace, $\concatTr{\tau}{\cont{e'}}$, arbitrarily. The trace~$\tau$
may be symbolic and we use concretization mapping~$\rho$ to concretize
it. Then, we use chop to stitch the concretized trace~$\rho(\tau)$ to
$\sh$ and continue with $e'$ in the top continuation marker. The lower
levels of the stack are unchanged.

\gcondrule{sync-base}{
  \sigma=\last(\sh) \quad \concatTr{\tau}{\cont{e'}} \in \valD{e} \quad
  \rho~\text{concretizes}~\tau \quad \wf{\chopTrSem{\sh}{\rho(\tau)}}
}{
  \sh, \left[\cont{e},\many{k}\right] \to \chopTrSem{\sh}{\rho(\tau)}, \left[\cont{e'},\many{k}\right]
}

\begin{example}\label{ex:sync-comp-1}
  We continue the running example. Assume a state~$\sigma_0$ with
  $\dom(\sigma_0)=\emptyset$ and let $e_0$ be the \lstinline|main|
  function's block in \zcref{lst:running}. Let further
  $\sigma_1:= \update{\sigma_0}{\pv{x}}{\falseSem}$ and
  $\sigma_2:=\update{\sigma_1}{\pv{y}}{\trueSem}$. We start in
  configuration $\mkTr{\sigma_0},[\cont{e_0}]$. As the states are
  concrete, we can use the concretization mapping $\rho_0$ with
  $\dom(\rho_0)=\emptyset$ throughout. By repeated applications of
  \zcref{eq:rule-sync-base} we obtain the following derivation
  \begin{align*}
    &\mkTr{\sigma_0},[\cont{e_0}]
      \to \mkTr{\sigma_0},[\cont{\rust{\rib{x}{x = \rsFalse; $\ldots$}}}] \\
    \to & \consTr{\sigma_1}{\mkTr{\sigma_0}},[\cont{\rust{\rib{x}{\kwLet c = !x; $\ldots$}}}]
          \to \consTr{\sigma_1}{\mkTr{\sigma_0}},[\cont{\rust{\rib{x,y}{y = !x; $\ldots$}}}] \\
    \to & \consTr{\sigma_1}{\mkTr{\sigma_0}},[\cont{\rust{\rib{x,y}{y = !$\falseSem$; $\ldots$}}}]
          \to  \consTr{\sigma_1}{\mkTr{\sigma_0}},[\cont{\rust{\rib{x,y}{y = $\trueSem$; $\ldots$}}}] \\
    \to & \consTr{\sigma_2}{\consTr{\sigma_1}{\mkTr{\sigma_0}}},[\cont{\rust{\rib{x,y}{\kwLet f = fac(4);$\ldots$}}}] \\
    \to & \consTr{\sigma_2}{\consTr{\sigma_1}{\mkTr{\sigma_0}}},[\cont{\rust{\rib{x,y}{\kwLet f; f = fac(4);$\ldots$}}}]\\
    \to & \consTr{\sigma_2}{\consTr{\sigma_1}{\mkTr{\sigma_0}}},[\cont{\rust{\rib{x,y,z}{z = fac(4);$\ldots$}}}] \\
    \to & \consTr{\sigma_2}{\consTr{\invEv{\pv{fac}}{i_0}{4}}{\consTr{\sigma_2}{\consTr{\sigma_1}{\mkTr{\sigma_0}}}}},[\cont{\rust{\rib{x,y,z}{z = ?$i_0$;$\ldots$}}}]
  \end{align*}
  It is not possible to apply \zcref{eq:rule-sync-base} further, as
  the well-formedness predicate prohibits adding the computation
  reaction event that evaluation of \rust{?$i_0$} would yield. Observe
  that we could have many other derivations similar to the one above,
  the only difference being that different concrete variable namings
  and call identifiers are chosen.
\end{example}

\zcref[S]{eq:rule-sync-base} ensures that all constructed traces are
well-formed and concrete. The composition is also stuck if the last
event of $\sh$ is a \emph{panicEv} event, as we cannot continue past
such a fatal error. Similarly, it is not applicable if $e$ is a value,
because values generate no traces.

\zcref[S]{eq:rule-sync-push} executes a function body and pushes the
resulting continuation onto the continuation stack. The rule may pick
any function from the program at any point, leading to possible
spurious function execution. However, the predicate $\mathit{wf}$
ensures with \zcref{eq:wf-invREv} that \zcref{eq:rule-sync-push} is
only applicable when appropriate, i.e., immediately after its
corresponding invocation event, and with the correct function. This
restriction ensures correct applications of the rule and prohibits
spurious function executions.
\gcondrule{sync-push}{
  \sigma=\last(\sh) \quad \rust{\kwFn foo($\many{\pv{x}}$) $e$}\in \Fn\\
  \concatTr{\tau}{\cont{e'}} \in \valD{\rust{\kwFn foo($\many{x}$) $e$}} \quad    \rho~\text{concretizes}~\tau \quad \wf{\chopTrSem{\sh}{\rho(\tau)}} 
}{
  \sh, \left[\many{k}\right] \to \chopTrSem{\sh}{\rho(\tau)}, \left[\cont{e'},\many{k}\right]
}
    
\begin{example}\label{ex:sync-comp-2}
  Continuing \zcref{ex:sync-comp-1} we start from configuration
  $\consTr{\sigma_2}{\consTr{\invEv{\pv{fac}}{i_0}{4}}{\consTr{\sigma_2}{\consTr{\sigma_1}{\mkTr{\sigma_0}}}}},[\cont{\rust{\rib{x,y,z}{z=?$i_0$;$\ldots$}}}]$
  by executing function \pv{fac} and putting the result to the top of
  the continuation stack. By \zcref{eq:rule-sync-push} we derive
  $\consTr{\sigma_3}{\consTr{\invREv{\pv{fac}}{i_0}{4}}{\consTr{\sigma_2}{\consTr{\invEv{\pv{fac}}{i_0}{4}}{\consTr{\sigma_2}{\consTr{\sigma_1}{\mkTr{\sigma_0}}}}}}},[\cont{\rust{\rib{p}{\kwReturn
        if p <= 1
        $\ldots$}}}, \cont{\rust{\rib{x,y,z}{z=?$i_0$;$\ldots$}}}]$
  where we have $\sigma_3:=\update{\sigma_2}{\pv{p}}{4}$. Note how the
  trace composition applies the concretization: The symbolic variables
  introduced by \zcref{eq:local-fn-body} are immediately concretized
  and the well-formedness predicate ensures that it is concretized to
  the actual call arguments. With
  \zcref{eq:rule-sync-base,eq:rule-sync-push} we can also evaluate the
  next call of \pv{fac} with argument~$3$ in \zcref{line:rec-call-fac}
  but we cannot, as yet, continue after a function returns.
\end{example}
    
Once a function call completes, we require a rule to pop the top
continuation off the continuation stack, effectively returning control
to the calling function. Crucially, we cannot do this step once the
top continuation marker contains a value, i.e., is fully evaluated,
because a \lstinline|return|-expression may not be the last or only
expression in a function body; it can occur at any point. Because this
\lstinline|return| emits a return event and immediately stops the
function's evaluation, we want to pop the top continuation marker off
the stack at this point.

Observe that, by
\zcref{eq:rule-sync-push,eq:local-call,eq:wf-invEv,eq:wf-invREv}, any
continuation marker but the top one must contain as the next
sub-expression to be evaluated a get expression~\rust{?$i$}---i.e.,
any continuation marker except the top one is waiting on the function
call above it to finish. By \zcref{eq:local-call-get,eq:wf-compREv},
such a get expression \rust{?$i$} can be evaluated iff the call~$i$
has just returned. This condition precisely models when we can clean
up the continuation marker for the call~$i$.

Hence, the structure of \zcref{eq:rule-sync-pop}: We remove the top
continuation marker $\cont{e}$ if we can evaluate the continuation
marker $\cont{e'}$ below it and the result is well-formed according to
\zcref{eq:wf-compREv}. We know that $e'$'s next sub-expression to be
evaluated is a get expression waiting on $\cont{e}$ to finish. If the
result is well-formed, the top continuation marker must have just
returned and must therefore be removed from the stack.
\gcondrule{sync-pop}{
  \sigma=\last(\sh) \quad \concatTr{\tau}{\cont{e''}} \in \valD{e'} \quad
  \rho~\text{concretizes}~\tau \quad \wf{\chopTrSem{\sh}{\rho(\tau)}}
}{
  \sh, \left[\cont{e},\cont{e'},\many{k}\right] \to \chopTrSem{\sh}{\rho(\tau)}, \left[\cont{e''},\many{k}\right]
}

\begin{example}\label{ex:sync-comp-3}
  We continue \zcref{ex:sync-comp-2}. After the call with
  identifier~$i_0$, we have a configuration of the form
  $\consTr{\sigma_3}{\consTr{\retEv{24}}{\consTr{\cdots}{\consTr{\invEv{\pv{fac}}{i_0}{4}}{\cdots}}}},
  [\cont{\rust{()}},\cont{\rust{\rib{x,y,z}{z=?$i_0$;$\ldots$}}}]$. Observe
  that the called function's continuation marker does not contain the
  return value; this value is transmitted only via events and is not
  passed to the surrounding context of the return.
  \zcref[S]{eq:rule-sync-pop} allows to remove the stack's top and
  continue local evaluation of the main block with
  \zcref{eq:local-call-get}:
  $\consTr{\sigma_3}{\consTr{\compREv{i_0}{24}}{\consTr{\sigma_3}{\consTr{\retEv{24}}{\cdots}}}},
  [\cont{\rust{\rib{x,y,z}{z=24;$\ldots$}}}]$.
\end{example}

\subsection{Global Program Semantics}
\label{sec:glob-progr-semant}

It remains to define the starting point of trace composition and the
semantics of a \Rust program as a whole. \zcref[S]{def:prog-sem}
further defines a partial semantics, which gives meaning to a \Rust
expression relative to a state~$\sigma$. Partial semantics are crucial
to the soundness proofs of \zcref{sec:sound-dl}.

\begin{definition}[(Partial) Program Semantics]\label{def:prog-sem}
  Let $e$ be a \Rust expression and $\sigma$ a concrete state. Assume
  a maximal sequence of applications of
  \zcref{eq:rule-sync-base,eq:rule-sync-push,eq:rule-sync-pop}
  \[\sh_0,[\cont{e}]\rightarrow\sh_1,[\many{k}]\to\cdots\]
  with $\sh_0=\mkTr{\sigma}$. If the sequence is finite, it has the
  form
  \[\mkTr{\sigma},[\cont{e}]\to^\ast\shi,[\many{k}']\]
  with $\to^\ast$ being the transitive closure of $\to$. If the
  sequence is infinite, we define $\shi=\lim_{i\to\infty}\sh_i$. We
  denote the set of all such traces, whether finite or infinite, as
  $\traces{e,\sigma}$.

  For a \Rust program $r$ with the main function \rust{\kwFn main() $e$},
  we define the set of traces for $r$ as
  $\traces{r}:=\traces{\rust{main()},\sigma}$, where $\sigma$ is the
  state with $\dom(\sigma)=\emptyset$.
\end{definition}

\begin{example}\label{ex:global-neg}
  Continuing \zcref{ex:local-neg}, we start in configuration
  $\mkTr{\sigma_1},\cont{\rust{!b}}$. With the local evaluation, we
  derive $\mkTr{\sigma_1},\cont{\rust{!$\falseSem$}}$, which is the
  only possible derivation. We then compose this with the local
  evaluation of \rust{!$\falseSem$} to get
  $\mkTr{\sigma_1},\cont{\trueSem}$. This trace is the only element of
  $\traces{\rust{!b},\sigma_1}$.
\end{example}

\begin{example}\label{ex:global-if}
  Let $\pv{x}, \pv{y}$ be program variables of type \lstinline|u32|,
  \pv{b} a variable of type \lstinline|bool|, $\sigma_0$ the state
  that assigns \pv{x}, \pv{y} the value~$0$ and
  \pv{b}~$\falseSem$. For the \Rust program
  \lstinline|x = if b {y = 3; y - 2} else {y = 4; y + 1}|, we have the
  following derivation:
  \begin{align*}
    &\mkTr{\sigma_0},\cont{\rust{x = \kwIf b \{y = 3; y - 2\} \kwElse \{y = 4; y + 1\}}}\\
    \to\,&\mkTr{\sigma_0},\cont{\rust{x = \kwIf $\falseSem$ \{y = 3; y - 2\} \kwElse \{y = 4; y + 1\}}}\\
    \to\,&\mkTr{\sigma_0},\cont{\rust{x = \{y = 4; y + 1\}}}
           \to \consTr{\update{\sigma_0}{\pv{y}}{4}}{\mkTr{\sigma_0}},\cont{\rust{x = \{(); y + 1\}}}\\
    \to\,&\consTr{\update{\sigma_0}{\pv{y}}{4}}{\mkTr{\sigma_0}},\cont{\rust{x = \{y + 1\}}}
           \to \consTr{\update{\sigma_0}{\pv{y}}{4}}{\mkTr{\sigma_0}},\cont{\rust{x = \{4 + 1\}}}\\
    \to\,&\consTr{\update{\sigma_0}{\pv{y}}{4}}{\mkTr{\sigma_0}},\cont{\rust{x = \{5\}}}
           \to \consTr{\update{\sigma_0}{\pv{y}}{4}}{\mkTr{\sigma_0}},\cont{\rust{x = 5}}\\
    \to\,&\consTr{\update{\update{\sigma_1}{\pv{y}}{4}}{\pv{x}}{5}}{\mkTr{\sigma_0}},\cont{\rust{()}}
  \end{align*}
\end{example}

\begin{example}\label{ex:global-block}
  Consider the expression $e=\;$\lstinline|{let x; x = 1; x}|. Assume
  a concrete state~$\sigma_2$; we start with
  $\mkTr{\sigma_2},\cont{\rust{\{\kwLet x; $\ldots$\}}}$. Per
  \zcref{eq:local-let}, the local evaluation yields the set
  $\{\concatTr{\mkTr{\sigma_2}}{\cont{\rust{\rib{\pv{y}}{y = 1; y}}}}
  \mid \pv{y}~\text{fresh}\}$. Hence, \zcref{eq:rule-sync-base} can
  pick any of these infinite traces to compose. We choose
  variable~\pv{z} as an example, yielding
  $\mkTr{\sigma_2},\cont{\rust{\rib{\pv{z}}{z = 1; z}}}$. Evaluating
  twice with both cases of \zcref{eq:local-block-expr-stmt}, we have
  $\consTr{\update{\sigma_2}{\pv{z}}{1}}{\mkTr{\sigma_2}},\cont{\rust{\rib{\pv{z}}{z}}}$. Using
  \zcref{eq:local-block-expr,eq:local-pv}, we have
  $\consTr{\update{\sigma_2}{\pv{z}}{1}}{\mkTr{\sigma_2}},\cont{\rust{\rib{\pv{z}}{1}}}$,
  which we finally evaluate with \zcref{eq:local-block-expr} to
  $\consTr{\sigma_2}{\consTr{\update{\sigma_2}{\pv{z}}{1}}{\mkTr{\sigma_2}}},\cont{1}$.
  The set $\traces{e,\sigma_2}$ is the set containing all traces of
  this form with different instantiations for \pv{y}. The first and
  final states are the same for all these traces.
\end{example}

\begin{example}\label{ex:panic}
  Consider a very simple program like
  \lstinline|fn main() { panic!(); let x = 1; () }|. Evaluating this
  program in state~$\sigma_0$ with $\dom(\sigma_0)=\emptyset$, we
  first evaluate the panic expression and obtain the configuration
  \(
  \consTr{\sigma_0}{\consTr{\panicEv}{\mkTr{\sigma_0}}},[\rust{\{();
    \kwLet~x = 1; ()\}}] \). We cannot continue from this
  configuration, as the well-formedness predicate allows exactly one
  state after a panic event and nothing else. Hence, this is the
  maximal application.
\end{example}

\begin{example}\label{ex:non-term}
  Assume the program \lstinline|fn main() { 'l: loop { () } }|, which
  is obviously non-terminating. Due to our use of chop, where
  $\chopTrSem{\mkTr{\sigma_0}}{\mkTr{\sigma_0}}=\mkTr{\sigma_0}$, we
  obtain the following infinite derivation:
  \begin{align*}
    &\sigma_0, \left[\cont{\rust{\{'l: \kwLoop~\{()\}\}}}\right]\\
    \to\,&\mkTr{\sigma_0}, \left[\cont{\rust{\{\loopFrame{\pv{l}}{\{()\}}\}}}\right]\\
    \to\,&\mkTr{\sigma_0}, \left[\cont{\rust{\{\loopFrame{\pv{l}}{()}\}}}\right]\\
    \to\,&\mkTr{\sigma_0}, \left[\cont{\rust{\{\getloop{\pv{l}\}}}}\right]\\
    \to\,&\mkTr{\sigma_0}, \left[\cont{\rust{\{'l: \kwLoop~\{()\}\}}}\right]\\
    \to\,&\ldots
  \end{align*}
  Hence, the traces for the program consist just of a singleton trace:
  $\{\mkTr{\sigma_0}\}$. If the loop has an effect on the state, e.g.,
  \lstinline|'l: loop { x += 1 }|, where \pv{x} is originally set to
  $0$, the resulting trace would be an infinite trace of the form
  $\consTr{\ldots}{\consTr{[\pv{x}\mapsto 2]}{\consTr{[\pv{x}\mapsto
        1]}{\mkTr{\sigma_0}}}}\ .$ See \autocite{LAGC23} for a
  discussion of this modeling of infinite computations.
\end{example}

We want to state some general properties over the set of generated
concrete traces for a state and expression $\traces{e,\sigma}$. In
particular, we notice that for the synchronous subset of \Rust
presented in \zcref{fig:syn-rs} the generated traces differ only in
the names chosen for local variables, e.g., in \zcref{eq:local-let},
and chosen borrow identifiers. Hence,
\zcref{def:eq-mod-ren-state,def:eq-mod-ren-trace} introduce the notion
of relative equality. Intuitively, two states may have different names
for some variables but they must agree on a fixed set of variables.

\begin{definition}[Renaming]\label{def:renaming}
  Let $\kappa:\Var\rightarrow\Var$ and $\iota:\BId\rightarrow\BId$ be
  two bijective functions. Then we call $\kappa$ a \emph{variable
    renaming}, $\iota$ a \emph{borrow renaming}, and $(\kappa,\iota)$
  a \emph{renaming}. Variable renamings $\kappa$ are lifted to places
  in the natural manner:
  \[
    \kappa(p)=
    \begin{cases}
      \kappa(\pv{x}) &\text{if}~p=\pv{x} \\
      \kappa(p')[n] &\text{if}~p=p'[n] \\
      \kappa(p').\mathit{identifier} &\text{if}~p=p'.\mathit{identifier}
    \end{cases}
  \]
  A renaming $(\kappa,\iota)$ can be applied on a reference value with
  the function \enquote{$\renameSym$} as follows:
  \begin{equation*}
    \begin{array}{r@{\,}c@{\,}l}
      \rename{\kappa,\iota}{\sRefVal{p}{b}} & = & \sRefVal{\kappa(p)}{\iota(b)} \\
      \rename{\kappa,\iota}{\mRefVal{p}{b}} & = & \mRefVal{\kappa(p)}{\iota(b)}
    \end{array}
  \end{equation*}
  The \enquote{$\renameSym$} function is extended to values in the
  usual manner; every reference contained in a value is renamed as
  above.
\end{definition}

\begin{definition}[State Equality Modulo
  Renaming]\label{def:eq-mod-ren-state}
  We define a relation $\eqModRenSym{V}$ over LAGC states and a set of
  variables $V\subseteq\Var$. Two states $\sigma_1,\sigma_2$ are
  \emph{equal modulo renaming relative to $V$}---written as
  $\eqModRen{V}{\sigma_1}{\sigma_2}$---iff
  $\dom(\sigma_1)\cap V=\dom(\sigma_2)\cap V$,
  $\sigma_1(\pv{x})=\sigma_2(\pv{x})$ for all
  $\pv{x}\in\dom(\sigma_1)\cap V$, and there exists a renaming
  $(\kappa,\iota)$ such that
  \(\sigma_1(\pv{x})=\rename{\kappa^{-1},\iota^{-1}}{\sigma_2(\kappa(\pv{x}))}\)
  for all $\pv{x}\in\dom(\sigma_1)\setminus V$ and
  $\sigma_2(\pv{x})=\rename{\kappa,\iota}{\sigma_1(\kappa^{-1}(\pv{x}))}$
  for all $\pv{x}\in\dom(\sigma_2) \setminus V$.
\end{definition}

\begin{proposition}
  For every set $V\subseteq\Var$, the relation $\eqModRenSym{V}$ over
  LAGC states is an equivalence relation.
\end{proposition}

\zcref[S]{def:eq-mod-ren-trace} extends this notion to traces in the
natural way.

\begin{definition}[Trace Equality Modulo
  Renaming]\label{def:eq-mod-ren-trace}
  We extend the relation $\eqModRenSym{V}$ over concrete traces for
  the set $V\subseteq\Var$ with the following inductive definition:
  \begin{itemize}
  \item $\eqModRen{V}{\epsilon}{\epsilon}$,
  \item
    $\eqModRen{V}{\consTr{\sigma_1}{\sh_1}}{\consTr{\sigma_2}{\sh_2}}$
    iff $\eqModRen{V}{\sh_1}{\sh_2}$ and $\eqModRen{V}{\sigma_1}{\sigma_2}$,
  \item
    $\eqModRen{V}{\consTr{\evTrioP{1}{\overline{v}_1}}{\sh_1}}{\consTr{\evTrioP{2}{\overline{v}_2}}{\sh_2}}$
    iff $\eqModRen{V}{\sh_1}{\sh_2}$ and
    $\evTrioP{1}{\overline{v}_1}=\evTrioP{2}{\overline{v}_2}$.
  \end{itemize}
  We again pronounce it as \emph{equality modulo renaming relative to
    $V.$}
\end{definition}

\zcref[S]{lem:determ} can now define the determinism of our \Rust
fragment in terms of generates traces.

\begin{lemma}\label{lem:determ}
  For each \ac{LAGC} state~$\sigma$ and each expression~$e$, the
  following propositions hold:
  \begin{align}
    \left|\traces{e,\sigma}/{\eqModRenSym{\dom(\sigma)}}\right|=1\label{eq:tr-mod-renam}\\
    (\forall\shi\in\traces{e,\sigma}.\ \shi~\text{is finite})\vee(\forall\shi\in\traces{e,\sigma}.\ \shi~\text{is infinite})\label{eq:tr-all-fin}\\
    \forall\shi\in\traces{e,\sigma}.\ \shi~\text{is finite} \rightarrow \dom(\sigma)=\dom(\last(\sh))\label{eq:tr-last}
  \end{align}
  \zcref[S]{eq:tr-mod-renam} states that all traces generated for $e$
  starting in $\sigma$ are equal modulo renaming relative to the
  variables defined in $\sigma$. As a direct consequence,
  \zcref{eq:tr-all-fin} states that either all generated traces are
  finite or all are infinite. Similarly, \zcref{eq:tr-last} states
  that for all generated traces, the first and final state (if one
  exists) share their domain. Together,
  \zcref{eq:tr-mod-renam,eq:tr-all-fin,eq:tr-last} give us the fact
  that either all generated traces are infinite or they share the
  final state.
\end{lemma}


\section{Asynchronous \Rust Semantics}
\label{sec:asynchr-rust-semant}

\ac{LAGC} is suitable for asynchronous and parallel code. Here, we
present a semantics for asynchronous constructs in \Rust and their use
in a concurrent, but not parallel, setting with the most permissive
scheduler possible. \zcref[S]{sec:fair-tokio} shows how these
definitions can be used to model a realistic async runtime: Tokio,
\Rust's most widely-used async runtime.

\subsection{Background: Async \Rust}
\label{sec:bg-async-rust}

\Rust's use of \lstinline|async|, \lstinline|await|, and futures seems
to be similar to other programming languages such as JavaScript and
active object languages like ABS~\cite{JHSSS2011} but it differs in
crucial respects.

Asynchronous \Rust revolves around the declaration and use of
\emph{futures,} i.e., representations of computations that have yet to
be completed. Unlike other languages, \Rust does not provide a single
type for futures, but instead defines a trait, \lstinline|Future|,
which can be implemented for many types.

A simplified definition of the \lstinline|Future| trait can be seen in
\zcref{lst:future}. It requires an implementor to define the output
type which the future is to compute and a \lstinline|poll| method,
which returns the \lstinline|Poll| enum. This type is either
\lstinline|Pending| or \lstinline|Ready(Output)|, the latter
containing the computed value. Calling \lstinline|poll| on a future
runs the future until its next suspension point. Such a point may
either be the end of the function body or another \lstinline|await|
expression. In the former case, the function returns \lstinline|Ready|
with the return value and is finished. In the latter case, we say the
future yields (or \enquote{suspends}) and returns \lstinline|Pending|.

\begin{lstlisting}[caption={A simplified definition of \Rust's \lstinline|Future| trait.}, label={lst:future}]
trait Future {
  type Output;
  fn poll(self) -> Poll<Self::Output>;
}
\end{lstlisting}

Hence, a \Rust future is a state machine that advances from the
initial state to its next yield and after a number of calls may reach
its final state and return a value. How the future exactly progresses
and when and how it computes its final value is left up to the
implementation. \Rust merely defines the basic operations and defines
suitable abstractions on top of them. \Rust futures are
\emph{passive}, that is, they are not automatically scheduled and need
to be actively \emph{polled} to be (partially) executed. On top of
these asynchronous primitives, libraries can provide async runtimes,
which allow scheduling of futures and can make these scheduled futures
behave like \enquote{active} futures, as we discuss below.

With an expression \rust{f.\kwAwait}, where \pv{f} is a future, a
future is repeatedly polled until it is completed. If the awaited
future is still pending, the await expression also suspends. Hence,
\lstinline|await| provides a suitable abstraction over polling of
futures. However, the \lstinline|await| expression is only allowed in
functions declared with the \lstinline|async| keyword, as in
\zcref{lst:async}. Assume that the async function
\enquote{\lstinline|get|} returns the HTML of a URL as a
\lstinline|String|. Calling \enquote{\lstinline|get(url)|,} hence,
returns a future with output type \lstinline|String|. In other
languages, the created future would automatically be scheduled. In
\Rust, however, this is not the case; the future remains passive.

\begin{lstlisting}[caption={An example async function in \Rust.}, label={lst:async}]
async fn check(url: &str) -> bool {
  let f1 = get(url); let f2 = get(url);
  let x = f1.await; let y = f2.await;
  return x.len() == y.len();
}
\end{lstlisting}

Functions marked with \lstinline|async| are only syntactic sugar. An
asynchronous function like \lstinline|check| is actually transformed
into a regular function that returns something that implements \Rust's
\lstinline|Future| trait. The signature of the function then looks as
follows
\begin{center}
  \lstinline|fn check(url: &str) -> impl Future<Output=bool> { ... }|
\end{center}

Hence, when the future of function \enquote{\lstinline|check|} polls
\enquote{\lstinline|get|} and the latter's future is still pending,
\enquote{\lstinline|check|}'s future also returns
\lstinline|Pending|. Awaiting futures is the only way to suspend a
future.\footnote{Technically, there are other possibilities for
  suspension, e.g., locks, but we leave these aside for brevity.}

With futures as the basic building blocks in place, a number of
different schedulers and abstractions can be built. Rather than
provide a single runtime for asynchronous programs, \Rust leaves
runtimes up to libraries and developers. A variety of runtimes have
been developed with different features and specializations:
Glommio~\autocite{G2024} is specialized for I/O operations and is
based on Linux's \lstinline|io_uring|; smol~\autocite{smol} is a
small, simple runtime with little overhead. By far the most prevalent
runtime library is Tokio~\autocite{tokio}.

With a runtime set up, futures can be handed to the runtime to be
scheduled; typically, this happens by calling the runtime's
\enquote{\lstinline|spawn|} function. It is the runtime's
responsibility to keep track of these futures and calling
\lstinline|poll| on scheduled futures to advance them. These
\enquote{spawned} or \enquote{scheduled} futures then are
\emph{active} futures and behave differently. They may be executed
whenever the runtime deems fit, even in parallel to other
futures. Hence, scheduled futures behave more analogously to futures
in active object languages~\autocite{dDFHHJKRSSY2017}. Such spawned
futures are often called \enquote{handles} (\lstinline|JoinHandle| in
Tokio's lingo).

Any non-scheduled future, even when created by a scheduled future,
behaves passively, as described above: They only progress when polled
and this polling may suspend the scheduled future. We call these
non-scheduled futures \enquote{local.}

A \Rust developer is unlikely to use the \lstinline|Future| trait
directly and is typically guided by \Rust's design and the existing
libraries to use \lstinline|async| functions and \lstinline|await|
expressions as the principal elements. Therefore, our semantics
abstracts away from \lstinline|poll| calls and instead focuses on
awaiting local (or unscheduled) and scheduled futures. This is not due
to a fundamental restriction of the approach. Modeling such
lower-level steps leads to a somewhat more verbose semantics and may
be justified in other contexts. Focusing on \lstinline|await| is an
important design decision and potential limitation but we believe it
to be a sensible abstraction for a high-level semantics, which allows
us to achieve the results laid out in the following sections.

\begin{example}[Running Example]\label{ex:async-run}
  In this section, we demonstrate our semantics with running examples
  based on \zcref{lst:async-prog}. It consists of two asynchronous
  functions, \lstinline|two_power| and \lstinline|bar|, both taking
  one parameter of type \lstinline|u32| and returning a value of the
  same type. The former function recursively computes $n\mapsto 2^n$
  for $n\in\mathbb{N}$, the latter computes $m\mapsto 2^{m+1}$ for
  $m\in\mathbb{N}$, using \lstinline|two_power|. The \lstinline|main|
  function schedules futures of both functions and awaits them later
  on, while the functions \lstinline|two_power| and \lstinline|bar|
  create and await only local futures.
  \begin{lstlisting}[caption={An asynchronous program.},label=lst:async-prog,numbers=left,multicols=2]
async fn two_power(n: u32) -> u32 {
  return if n == 0 {
    1
  } else {
    2 * two_power(n - 1).await %\label{line:local-two-power}%
  }; ()
}
async fn bar(m: u32) -> u32 {
  return two_power(m).await * 2; () %\label{line:local-bar}%
}
fn main() {
  let f1 = spawn(two_power(1)); %\label{line:call-two-power}%
  let f2 = spawn(bar(2)); %\label{line:call-bar}%
  let x = f1.await + f2.await; %\label{line:await-spawned}%
  return (); ()
}
\end{lstlisting}

  Observe that the main function here is only used for ease of
  presentation: In truth, it must set up the runtime and does not have
  access to \lstinline|await|, as it is not declared as
  \lstinline|async|. But the pattern is not uncommon. See how Tokio,
  for instance, sets up the \lstinline|tokio::main| macro to achieve
  the same effect~\autocite{TokioDoc}.
\end{example}

\subsection{Local Evaluation}
\label{sec:async-local}

We extend the values as defined in \zcref{def:val} with \emph{future
  identifiers} $\FId \subset \Val$, standing for created futures. The
syntax of \zcref{fig:syn-rs} is extended with additional syntactic
structures: await expressions like \rust{$e$.\kwAwait} and
asynchronous functions declared as %
\lstinline|async fn foo($\pv{x}_1,\ldots,\pv{x}_n$) $e$|, where
$\pv{x}_1,\ldots,\pv{x}_n$ are the parameters and $e$ the function
body. From this point on, the evaluation function is tagged with a
future identifier $F \in \FId$: $\valDnoargs^F$, which denotes the
currently active future.

Evaluating a call to async function \pv{foo} adds a \emph{future
  creation event} $\createEv{\pv{foo}}{\pv{F}'}{\many{v}}$, denoting
that a future of \pv{foo} is called with parameters
$\many{v}$. $\pv{F}'$ is a fresh, symbolic variable standing for the
created future. The symbolic variable will be concretized with a
concrete future ID during composition.
\evalrule{async-call}{
  \valP{F}{\rust{foo($\many{v}$)}}&=&\{\concatTr{\consTr{\update{\sigma}{\pv{F}'}{\ast}}{\consTr{\createEv{\pv{foo}}{\pv{F}'}{\many{v}}}{\mkTr{\sigma}}}}{\cont{\pv{F}'}} \mid \pv{F}'~\text{fresh}\}  
}

\begin{example}\label{ex:async-run-create}
  Calling \lstinline|two_power| in \zcref{line:call-two-power} of
  \zcref{lst:async-prog}, creates a future. The evaluation is
  $\valBP{\sigma_0}{f_0}{\rust{two\_power(1)}}
  \{\concatTr{\consTr{\update{\sigma_0}{\pv{F}}{\ast}}{\consTr{\createEv{\pv{two\_power}}{\pv{F}}{1}}{\mkTr{\sigma_0}}}}{\cont{\pv{F}}}
  \mid \pv{F}~\text{fresh}\}$, assuming we are in some
  state~$\sigma_0$ and the current future is~$f_0$.
\end{example}

A future is scheduled, or spawned, with the \lstinline|spawn|
function. We mark this with a \emph{spawn event} containing the
future's identifier and the symbolic variable \pv{H}, standing for a
\emph{handle identifier}. The latter is needed, as we formally
differentiate between local futures and handles, i.e., wrappers around
futures which represent the tasks. The set of all handle identifiers
is called $\HId \subset \FId$.
\evalrule{async-spawn}{
  \valP{F}{\rust{spawn($v$)}}&=&
    \{\concatTr{\consTr{\update{\sigma}{\pv{H}}{\ast}}{\consTr{\spawnEv{v}{\pv{H}}}{\mkTr{\sigma}}}}{\cont{\pv{H}}} \mid \pv{H}~\text{fresh}\}  
}

\begin{example}\label{ex:async-run-spawn}
  The futures created in the \lstinline|main| function in
  \zcref{lst:async-prog} are spawned. If we assume that the future
  created in \zcref{ex:async-run-create} is concretized to
  future~$f_t$ and the future of \lstinline|bar| is~$f_b$, we have the
  spawn events $\spawnEv{f_t}{h_t}$ and $\spawnEv{f_b}{h_b}$ for some
  fresh handle identifiers $h_t$ and $h_b$.
\end{example}

Evaluating an async function emits an \emph{evaluation event} with the
function's name, the currently executing future and fresh symbolic
variables $\many{\pv{X}}$ for the call arguments. Then, the function's
body is evaluated with parameters $\many{\pv{x}}$ substituted with
$\many{\pv{X}}$, just as in \zcref{eq:local-fn-body}.
\evalrule{async-body}{
  \valP{F}{\rust{\kwAsync~\kwFn foo($\many{\pv{x}}$) $e$}} = \{\concatTr{\consTr{\sigma[\many{\pv{X}\mapsto\ast}]}{\consTr{\evalEv{\pv{foo}}{F}{\many{\pv{X}}}}{\mkTr{\sigma}}}}{\cont{e\subst{\many{\pv{x}}}{\many{\pv{X}}}}} \mid \many{\pv{X}}~\text{fresh}\}
}

\begin{example}\label{ex:async-run-body}
  If we evaluate function \lstinline|bar| of \zcref{lst:async-prog} in
  state~$\sigma_2$ and future~$h_b$, we have the traces
  $\{\concatTr{\consTr{\update{\sigma_2}{\pv{X}}{\ast}}{\consTr{\evalEv{\pv{bar}}{h_b}{\pv{X}}}{\mkTr{\sigma_2}}}}{\cont{\rust{\{\kwReturn
        two\_power(X).\kwAwait * 2;\}}}} \mid \many{\pv{X}}~\text{fresh}\}$.
\end{example}

A \lstinline|return| expression now emits a computation event with the
active future and the returned value. We replace
\zcref{eq:local-return} with the following:
\evalrule{async-return}{
  \valP{F}{\rust{\kwReturn $v$}} &=& \{\concatTr{\consTr{\sigma}{\consTr{\compEv{F}{v}}{\mkTr{\sigma}}}}{\cont{\rust{()}}}\}
}

Awaiting a local future emits a \emph{local await event}
$\localAwEv{v}$, where $v \in \FId$ is the awaited future. Since this
event starts computation of the function represented by $v$, we model
it similarly to a local synchronous call (\zcref{eq:local-call}) and
leave a \enquote{get expression} for the awaited future. This
expression can now be applied to call identifiers or local
futures. However, if $v$ is a handle identifier, i.e., refers to a
spawned future, we directly emit a \emph{computation reaction event}
with a symbolic result value. The well-formedness predicate ensures
that an \lstinline|await| on a spawned future is only executed once
this future is completed.
\evalrule{await}{
  \valP{F}{\rust{$v$.\kwAwait}}&=&
  \begin{cases}
    \{\concatTr{\consTr{\sigma}{\consTr{\localAwEv{v}}{\mkTr{\sigma}}}}{\cont{\rust{?$v$}}}\}&\text{if}~v\in\FId \setminus \HId\\
    \{\concatTr{\consTr{\update{\sigma}{\pv{V}}{\ast}}{\consTr{\compREv{v}{\pv{V}}}{\mkTr{\sigma}}}}{\cont{\pv{V}}} \mid \pv{V}~\text{fresh}\}&\text{if}~v\in\HId
  \end{cases}
}

\begin{example}\label{ex:async-run-await}
  The futures created by the calls in \zcref{lst:async-prog},
  \zcref{line:local-two-power,line:local-bar} are local futures, which
  are not spawned; hence, their \lstinline|await| expressions in those
  same lines \enquote{end up} in the first case of
  \zcref{eq:local-await}. On the other hand, the \lstinline|await|s in
  \zcref{line:await-spawned} are handled by \zcref{eq:local-await}'s
  second case, as the futures were spawned and are thus handle
  identifiers.
\end{example}

\subsection{Well-Formedness}
\label{sec:async-wf}

We adjust the well-formedness predicate as follows: A trace ending in
a creation event is well-formed if the preceding trace is well-formed
and does not contain a creation event with the same future identifier.
The create event cannot create a handle, as those must be created via
spawning.
\begin{equation}\label{eq:wf-create}
  \wf{\consTr{\createEv{\pv{foo}}{f}{\many{v}}}{\sh}} = \wf{\sh} \wedge f\not\in\HId \wedge \nexists\pv{foo}',\many{v}'.\ \createEv{\pv{foo}'}{f}{\many{v}'} \in \sh
\end{equation}

Spawn events must follow a creation event of the spawned future and
this future must not have been spawned or locally awaited already. The
created handle identifier must also be unique. A future can either be
local or spawned, never both. Hence, no future can have both a local
await and a spawn event in the same trace.
\begin{equation}\label{eq:wf-spawn}
  \begin{aligned}
    \wf{\consTr{\spawnEv{f}{h}}{\sh}} = \,&\wf{\sh} \wedge \exists \pv{foo},\many{v}.\ \createEv{\pv{foo}}{f}{\many{v}}\in\sh \wedge \localAwEv{f}\not\in\sh \\
                                          &\wedge \nexists h'.\ \spawnEv{f}{h'} \in \sh \wedge \nexists f'.\ \spawnEv{f'}{h} \in \sh
  \end{aligned}
\end{equation}

Local await events are only valid if the future has already been
created but not awaited or spawned before:
\begin{align*}
  \wf{\consTr{\localAwEv{f}}{\sh}} =\, &\wf{\sh} \wedge f\not\in\HId \wedge 
                                         \exists\pv{foo},\many{v}.\ \createEv{\pv{foo}}{f}{\many{v}} \in \sh\\
                                       &\wedge \localAwEv{f}\not\in\sh \wedge \nexists h.\
                                         \spawnEv{f}{h}\in\sh
\end{align*}

An evaluation event can only occur once per future. Additionally, it
must follow a local await event or a spawn event. If the future~$f$ is
local, it must have been created before being evaluated and the
function name and parameter must match, ensuring a pairing between
call and execution. If the future is a spawned future, it must have
been spawned from a future~$f'$ that was created with the correct
function name and arguments.
\begin{equation}\label{eq:wf-eval}
  \begin{aligned}
    \wf{\consTr{\evalEv{\pv{foo}}{f}{\many{v}}}{\sh}} =\, &\wf{\sh} \wedge \nexists\pv{foo}',\many{v}'.\ \evalEv{\pv{foo}'}{f}{\many{v}'} \in \sh \\
                                                          &\wedge (\localAwEv{f} \in \sh \wedge \createEv{\pv{foo}}{f}{\many{v}}\in\sh \\
                                                          &\quad \vee f\in\HId \wedge \exists f'.\ \spawnEv{f'}{f}\in\sh \wedge \createEv{\pv{foo}}{f'}{\many{v}} \in \sh)
  \end{aligned}
\end{equation}

We place no restrictions on computation events, as we can only reach a
return after evaluating a function, which would already have taken
care of the necessary events. Computation reaction events must follow
a corresponding computation event, which leads to a case distinction:
If the event is over a future, it must follow a computation event of
that future. If, however, the event is over a call identifier, it was
a synchronous call and, hence, the future in the computation event is
the future in which the function was called; the computation reaction
event must also directly succeed the computation event. There can be
at most one computation reaction event per identifier. In the
following, we use $\mathit{fi}\in\CId \uplus \FId$ for a future or
call identifier.
\begin{equation*}
  \begin{array}{@{}r@{\,}c@{\,}l@{}}
    \wf{\consTr{\compEv{f}{v}}{\sh}} & = & \wf{\sh} \\
    \wf{\consTr{\compREv{\mathit{fi}}{v}}{\sh}} & = & \wf{\sh} \wedge \nexists v'.\ \compREv{\mathit{fi}}{v'} \in \sh \\
                                     & & \wedge \, (\mathit{fi}\in\FId \rightarrow \compEv{\mathit{fi}}{v} \in \sh) \\
                                     & & \wedge \, (\mathit{fi}\in\CId \rightarrow \lastEv(\sh)=\compEv{\mathit{fi}}{v})
  \end{array}
\end{equation*}

\subsection{Trace Composition Rules}
\label{sec:async-comp}

We define trace composition rules for async \Rust that generate all
possible schedules of futures or \enquote{tasks.} They are defined
over configurations consisting of a bag of tasks~$Q$. A task, thus,
consists of a stack of continuation markers as in
\zcref{sec:trace-comp-rules} with two differences:
\begin{enumerate*}[(a)]
\item the continuation markers are tagged with future identifiers,
  e.g., $\contP{f_0}{\rust{x + y}}$, and
\item the elements of the stack represent both called functions and
  locally evaluated, unscheduled futures.
\end{enumerate*}
A configuration is written as $\asCfg{\sh}{\mathit{active}}{Q}$, where
$\sh$ is the current shining trace, $\mathit{active}$ is the active
task, i.e., the task currently being executed, or $\Idle$ if there is
none, and $Q$ is the current bag of tasks. This representation ensures
that an active task is not prematurely preempted.

The first rule executes the top continuation marker of the active task
if possible. It pops the first continuation of the stack, takes one
step of the local evaluation, extends the trace by the concretization
of the result if it is well-formed, and pushes the new continuation
onto the stack. This rule corresponds to \zcref{eq:rule-sync-base}
extended with a bag of stacks. We \highlightRule{highlight} which
parts of the rules are changed.
\gcondrule{Eval}{
  \sigma=\last(\sh) \quad \concatTr{\tau}{\cont{e'}}\in\valP{\highlightRule{f}}{e} \quad
  \rho\text{ concretizes }\tau \quad \wf{\chopTrSem{\sh}{\rho(\tau)}}
}{
  \asCfgH{\sh}{\left[\contP{\highlightRule{f}}{e},\many{k}\right]}{Q} \to \asCfgH{\chopTrSem{\sh}{\rho(\tau)}}{\left[\contP{\highlightRule{f}}{e'},\many{k}\right]}{Q}
}

\begin{example}\label{ex:async-run-eval}
  We continue our running \zcref{ex:async-run}. Assume we are in a
  configuration
  $\asCfg{\consTr{\sigma_0}{\sh_0}}{[\contP{h_t}{\rust{\{\kwReturn~\kwIf
        n == 0 $\ldots$\}}}]}{Q_0}$ with $\sigma_0(\pv{n})=1$. Then,
  we are executing function \lstinline|two_power| in the active task,
  $Q_0$ must contain the task for \lstinline|main| and, possibly, the
  task for \lstinline|bar| on handle~$h_b$ (the latter may already be
  finished). With \zcref{eq:rule-Eval}, we execute the body one step
  further to
  $\asCfg{\consTr{\sigma_0}{\sh_0}}{[\contP{h_t}{\rust{\{\kwReturn~\kwIf
        1 == 0 $\ldots$\}}}]}{Q_0}$. Clearly, this is the semantics we
  expect, behaving similar to \zcref{sec:trace-comp-rules}.
\end{example}

\zcref[S]{eq:rule-PushCall} extends the active task with a new call
context, executing a locally called, synchronous function, just as
\zcref{eq:rule-sync-push} does in the synchronous case. The rules are
identical except for the addition of $Q$ and the tag~$f$.
\gcondrule{PushCall}{
  \sigma=\last(\sh) \quad f\in\FId \quad \rust{\kwFn foo($\many{\pv{x}}$) $e$}\in\Fn \\
  \concatTr{\tau}{\cont{e'}}\in\valP{\highlightRule{f}}{\rust{\kwFn
      foo($\many{x}$) $e$}} \quad \rho\text{ concretizes }\tau \quad
  \wf{\chopTrSem{\sh}{\rho(\tau)}} }{
  \asCfgH{\sh}{\left[\many{k}\right]}{Q} \to
  \asCfgH{\chopTrSem{\sh}{\rho(\tau)}}{\left[\contP{\highlightRule{f}}{e'},\many{k}\right]}{Q}
}

A locally awaited future is handled exactly the same as a synchronous
in call \zcref{eq:rule-PushCall}, except in this case, an active task
exists and the new continuation is pushed on
top. \zcref[S]{eq:rule-PushTask} has no corresponding rule in
\zcref{sec:trace-comp-rules}.
\gcondrule{PushTask}{
  \sigma=\last(\sh) \quad f\in\FId\setminus\HId \quad \rust{\kwAsync~\kwFn foo($\many{x}$) $e$} \in\Fn\\
  \concatTr{\tau}{\cont{e'}}\in\valP{f}{\rust{\kwAsync~\kwFn foo($\many{x}$) $e$}} \quad
  \rho\text{ concretizes }\tau \quad \wf{\chopTrSem{\sh}{\rho(\tau)}}
}{
  \asCfg{\sh}{\left[\many{k}\right]}{Q} \to \asCfg{\chopTrSem{\sh}{\rho(\tau)}}{\left[\contP{f}{e'},\many{k}\right]}{Q}
}

\begin{example}\label{ex:async-run-push-task}
  \zcref[S]{eq:rule-PushTask} is used, for instance, for the recursive
  future in \lstinline|two_power|: After the local future is awaited,
  the rule becomes applicable, according to $\mathit{wf}$. Like a
  synchronous function call, the stack is extended with the new future
  as top and takes the form
  $[\contP{f_3}{\rust{\{\kwReturn if k == 0
      $\ldots$\}}},\contP{h_t}{\rust{?$f_3$}}]$.
\end{example}

Similar to \zcref{eq:rule-sync-pop}, we pop the top continuation of a
stack, returning control to the calling function. As before, we do not
require $e$ to be a value. Rather, the structure of the continuation
stack and the well-formedness predicate ensures that $e'$ is a get
expression, which can only be executed once the called function,
represented by $\contP{f}{e}$, returns by the argument for
\zcref{eq:rule-sync-pop}.
\gcondrule{Pop}{
  \sigma=\last(\sh) \quad \concatTr{\tau}{\cont{e''}} \in \valP{\highlightRule{f'}}{e'} \quad \rho\text{ concretizes }\tau \quad \wf{\chopTrSem{\sh}{\rho(\tau)}}
}{
  \asCfgH{\sh}{\left[\contP{\highlightRule{f}}{e},\contP{\highlightRule{f'}}{e'},\many{k}\right]}{Q} \to
  \asCfgH{\sh}{\left[\contP{\highlightRule{f'}}{e''},\many{k}\right]}{Q}
}

\begin{example}
  After the recursive future in \zcref{ex:async-run-push-task} was
  pushed onto the stack and it was fully evaluated as in
  \zcref{ex:async-run-eval}, the \lstinline|return| emitted a
  computation event. According to the $\mathit{wf}$ predicate the this
  event must follow, which can only be done by the \rust{?$f_3$}
  expression in the continuation marker below the top. Thus, only
  \zcref{eq:rule-Pop} is applicable and the stack is simplified to
  $[\contP{h_t}{1}]$.
\end{example}

If the stack is empty, we have no more active task and mark the
configuration as idle.
\grule{Empty}{
  \asCfg{\sh}{[]}{Q} \to \asCfg{\sh}{\Idle}{Q}
}

When the next sub-expression of an expression~$e$ is an \kwAwait on a
scheduled future that has no corresponding computation event in
trace~$\sh$, we call $e$ \emph{suspended in $\sh$}, written
$\suspended{e}{\sh}$. We define this predicate as
\[
  \suspended{e}{\sh}:=\exists h\in\HId.\
\nextExpr{e}{\Val}=(\rust{$h$.\kwAwait}, \ldots) \wedge \nexists v\in\Val.\ \compEv{h}{v}\in\sh\ .
\]

If the top continuation marker of the active task contains a suspended
expression, we no longer have an active task and the suspended task is
added to the bag of tasks. This models \Rust's behavior of an
\lstinline|await| only yielding if the awaited future is not yet
ready.
\gcondrule{Suspend}{
  \suspended{e}{\sh}
}{
  \asCfg{\sh}{\left[\contP{f}{e},\many{k}\right]}{Q} \to \asCfg{\sh}{\Idle}{Q \uplus \left\{\left[\contP{f}{e},\many{k}\right]\right\}}
}

A scheduled future can be executed if there is no active task. Then an
arbitrary function is taken, evaluated, and, if well-formed, appended
to the trace. The well-formedness ensures that only scheduled
functions are executed. The configuration now has an active task with
the new continuation marker as its only element.
\gcondrule{SchedTask}{
  \sigma=\last(\sh) \quad f\in\HId \quad \rust{\kwAsync~\kwFn foo($\many{x}$) $e$}\in\Fn\\
  \concatTr{\tau}{\cont{e'}}\in\valP{f}{\rust{\kwAsync~\kwFn foo($\many{x}$) $e$}} \quad
  \rho\text{ concretizes }\tau \quad \wf{\chopTrSem{\sh}{\rho(\tau)}}
}{
  \asCfg{\sh}{\Idle}{Q} \to \asCfg{\chopTrSem{\sh}{\rho(\tau)}}{\left[\contP{f}{e'}\right]}{Q}
}

\begin{example}
  \zcref[S]{eq:rule-SchedTask} may be applied to functions
  \lstinline|two_power| and \lstinline|bar| after futures of them are
  spawned in \zcref{line:call-two-power,line:call-bar}, as soon as
  there is no active task. Observe how the well-formedness predicate
  ensures that their identifiers are unique and that parameters match.
\end{example}

Alternatively, if we are idle, we can pick a task from the bag $Q$ and
add it as the active task. The chosen task is also removed from $Q$.

\gcondrule{ReschedTask}{
  \left[\contP{f}{e},\many{k}\right] \in Q
}{
  \asCfg{\sh}{\Idle}{Q} \to \asCfg{\sh}{\left[\contP{f}{e},\many{k}\right]}{Q \setminus \left\{\left[\contP{f}{e},\many{k}\right]\right\}}
}

\subsection{Global Program Semantics}
\label{sec:async-global}

Let $r$ be a \Rust program with main function \rust{\kwFn main() $e$}. We
extend the global semantics of \zcref{sec:glob-progr-semant} to our
present case:

\begin{definition}[Asynchronous Program Semantics]\label{def:prog-sem-async}
  Let $r$ be a an asynchronous \Rust program with the main function
  \rust{\kwFn main() $e$}. Assume a maximal sequence of applications
  of
  \zcref{eq:rule-Eval,eq:rule-SchedTask,eq:rule-PushTask,eq:rule-PushCall,eq:rule-Pop,eq:rule-Empty,eq:rule-ReschedTask,eq:rule-Suspend}
  \[
    \asCfg{\sh_0}{\left[\contP{f^\text{init}}{e}\right]}{\emptyset}\to\asCfg{\sh_1}{\mathit{active}_1}{Q_1}\to\cdots
  \]
  with $\sh_0=\mkTr{\sigma_0}$, $\dom(\sigma_0)=\emptyset$ and
  $f^\text{init}$ is the initial future identifier. If the sequence is
  finite, it has the form
  \[
    \asCfg{\mkTr{\sigma}}{\left[\contP{f^\text{init}}{e}\right]}{\emptyset}\to^\ast\asCfg{\shi}{\left[\many{k}\right]}{Q}
  \]
  If the sequence is infinite, we define
  $\shi=\lim_{i\to\infty}\sh_i$. We denote the set of all such
  possibly infinite traces as $\traces{r}$.
\end{definition}

This semantics for asynchronous programs now gives meaning to programs
with \Rust's async features. It provides a highly non-deterministic
scheduler that is maximally permissive: Any possible scheduling of
futures in a program~$r$ is covered by the semantics and included in
the set $\traces{r}$. The well-formedness restrictions ensure that any
trace in $\traces{r}$ is valid, i.e., it follows the existing \Rust
semantics. While typically not desired from an actual runtime, it
gives us a base to which we can compare any real-world runtime, as
they must produce a subset of the traces in $\traces{r}$.

It remains to model the Tokio runtime and show that they produce fair
schedules.


\section{Showing the Fairness of Tokio's Schedulers}
\label{sec:fair-tokio}

A central property of any runtime (or scheduler) is its
\emph{fairness}: If a future is scheduled, a developer wants
guarantees that, under certain conditions, this future is eventually
executed. In the case of \Rust's Tokio runtime, its documentation
gives informal fairness guarantees~\autocite{TokioDoc}.

We cover both the formalization of Tokio's scheduling behavior in
terms of composition rules and a formalization of fairness guarantees,
including those of the documentation. We also provide a formal proof
that these guarantees hold for one of Tokio's runtimes.

\subsection{Background: Tokio}
\label{sec:bg-tokio}

The Tokio library provides two runtimes: the \emph{current thread
  runtime}, which works on a single thread; and the
\emph{multi-threaded runtime}, which provides concurrency and
parallelism. Although the implementations differ slightly, both
runtimes are intended to work similarly, except for the number of
threads.

The following description is derived partly from Tokio's
documentation~\autocite{TokioDoc} and partly from analysis of Tokio's
code base~\autocite{tokio}.

In both runtimes, each thread is associated with a \emph{worker}, each
of which has a \emph{local queue} of tasks and, potentially, an active
task (similar to the notion of active task in
\zcref{sec:async-comp}). A task is a future scheduled on the
runtime. Each worker also has a \emph{tick}, a counter that is
increased whenever a future is polled. Further, each runtime has a
\emph{global queue} of tasks. Additionally, the runtimes have a
parameter for the \emph{global queue interval}, call it
$I_G \in \mathbb{N}$, controlling the interval in which the global
queue is considered. This parameter cannot be changed while the
runtime is active. To which queue a scheduled task is appended
depends on where it is scheduled and the queues' current state. Tasks
can be scheduled either from a task executed by a worker or from
outside the runtime, i.e., from parts of the \Rust program not
controlled by the runtime. If the task is scheduled from outside the
runtime or the local queue is full, the task is pushed to the global
queue.  Otherwise, it is pushed onto the corresponding local queue.

Additionally, the runtimes track I/O timers and events. However, as
I/O is more difficult to model, we ignore this feature for now and
only model futures.

\subsubsection{The Current Thread Runtime}
The current thread runtime must decide which task to execute next if
the worker is \emph{idle}, i.e., if it has no active task: If the tick
is divisible by $I_G$ or the local queue is empty, remove the first
task of the global queue and execute it. Otherwise, remove and execute
the first task of the local queue. Thus, tasks in the local queue are
prioritized while the global queue is only queried in periodic
intervals.

\subsubsection{The Multi-Threaded Runtime}
The multi-threaded runtime not only has more workers but also some
optimizations and more complex scheduling. Each worker has a \ac{LIFO}
slot in addition to its queue. This slot contains the last locally
scheduled task. The assumption is that the most recently scheduled
task is often more pressing. An idle worker considers the following
cases to choose its next task:
\begin{enumerate*}[(a)]
\item\label{item:mt-ig} If the tick is divisible by $I_G$ and the
  global queue is not empty, execute the first global task;
\item\label{item:mt-lifo} otherwise, if there is a task in the
  \ac{LIFO} slot, execute it;
\item\label{item:mt-local} otherwise, if the local queue is not empty,
  execute its first task;
\item\label{item:mt-local-empty} otherwise, if the global queue is not
  empty, take a fixed number of tasks from the global queue, execute
  the first task, and push the rest to the back of the local queue;
\item\label{item:mt-steal} otherwise, try to \enquote{steal,} i.e.,
  take, another worker's local tasks.
\end{enumerate*}

Case~\zcref[noname]{item:mt-ig} ensures that tasks in the global queue
are not ignored and checked in a fixed interval. If a worker is idle,
case~\zcref[noname]{item:mt-local-empty} restores some balance by
moving global tasks to the idle worker. Similarly,
case~\zcref[noname]{item:mt-steal} balances the load between workers.

\zcref[S]{fig:runtime} shows the structure of both Tokio runtimes. We
now model Tokio by extending our \ac{LAGC} semantics. The local
evaluation rules defined in \zcref{sec:async-comp} are applicable to
both runtimes; only the trace composition and program semantics differ
between the two runtimes.

\begin{figure}[h]
  \centering
  \begin{tikzpicture}
    \newcommand{\tikzQ}[4][2]{
      \foreach \i in {1,...,#2} {
        \pgfmathparse{#2-1}
        \ifnum \i=\pgfmathresult
          \draw (#3 + \i/#1,#4) node {$\cdots$};
        \else
          \draw (#3 + \i/#1,#4) node[minimum size=.5cm,draw] {};
        \fi
      }
    }
    
    \draw[rounded corners] (0.1,.85) rectangle ++(2.85,.7);
    \draw (0.4,1.2) node {$q_G$};
    \tikzQ{4}{.5}{1.2}

    \draw[rounded corners,dashed] (-2.5,-1.65) rectangle ++(4.1,2);
    \draw (-1.9,0.35) node[circle,draw,fill=white] {$w_1$};
    \draw (-2.1, -.4) node {$t_1$};
    \draw (-1.3, -.4) node[minimum size=.5cm,draw] {$\mathit{act}_1$};
    \draw (-.1, -.4) node[minimum size=.5cm,draw,dotted] {$\mathit{lifo}_1$};
    \draw[rounded corners] (-2.4,-1.55) rectangle ++(3.85,.7);
    \draw (-2.1,-1.2) node {$q_1$};
    \tikzQ{6}{-2}{-1.2}

    \draw (2.1,-.65) node {$\cdots$};
    
    \draw[rounded corners,dashed] (2.5,-1.65) rectangle ++(3.6,2);
    \draw (3.1,0.35) node[circle,draw,fill=white] {$w_n$};
    \draw (2.9, -.4) node {$t_n$};
    \draw (3.7, -.4) node[minimum size=.5cm,draw,dotted] {$\mathit{act}_n$};
    \draw (4.9, -.4) node[minimum size=.5cm,draw] {$\mathit{lifo}_n$};
    \draw[rounded corners] (2.6,-1.55) rectangle ++(3.35,.7);
    \draw (2.9,-1.2) node {$q_n$};
    \tikzQ{5}{3}{-1.2}
  \end{tikzpicture}
  \Description{A global queue $q_G$ as a list of boxes. $n$ workers
    $w_1,\ldots,w_n$, each with a tick $t_i$, a slot for an active
    task stack $\mathit{act}_i$, a \ac{LIFO} slot $\mathit{lifo}_i$,
    and a queue of task stacks $q_i$.}
  \caption{The general setup of Tokio's runtimes. $q_G$ is the global
    queue of tasks. There are $n$ workers $w_i$, each with a
    tick~$t_i$, a local queue~$q_i$, a (potentially empty) active task
    $\mathit{act}_i$, and a \ac{LIFO} slot $\mathit{lifo}_i$. Note
    that $w_1$ has an active task but no task in its \ac{LIFO} slot
    (denoted by the dotted border), whereas $w_n$ has a task in its
    \ac{LIFO} slot but no active task (dotted border). The current
    thread runtime has only one worker and no \ac{LIFO} slot.}
  \label{fig:runtime}
\end{figure}
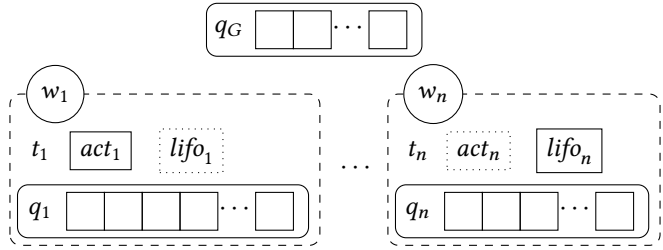

\subsection{Local Evaluation}
\label{sec:tokio-local}

As we already defined spawning of tasks with the \lstinline|spawn|
function in \zcref{sec:asynchr-rust-semant}, we reuse those
definitions to define Tokio's scheduler. No additional local rules are
required.

\subsection{Well-Formedness}
\label{sec:tokio-wf}

We adjust and extend the definition of the $\mathit{wf}$-predicate by
adding a case for the new \emph{schedule event}
$\schedEv{\pv{foo}}{f}$, where \pv{foo} is an asynchronous function
and $f\in\FId$ is a future identifier. It denotes that a task for
function~\pv{foo} with future~$f$ has been scheduled.  This event must
follow corresponding spawn and create events. This ensures a coupling
of creation, spawning, and scheduling. The above requirements also
ensure that evaluation is similarly associated.
\begin{equation}
  \label{eq:wf-sched-tokio}
  \begin{aligned}
    \wf{\consTr{\schedEv{\pv{foo}}{f}}{\sh}} = \,&\wf{\sh} \wedge \exists h\in\HId.\ \spawnEv{f}{h}\in\sh \\
                                                 &\quad\wedge \exists\many{v}.\ \createEv{\pv{foo}}{f}{\many{v}}\in\sh \wedge \nexists\pv{foo}'.\ \schedEv{\pv{foo}'}{f} \in\sh
  \end{aligned}
\end{equation}

All other cases are unchanged from \zcref{sec:async-wf}.

\subsection{Modeling the Current Thread Runtime}
\label{sec:model-curr-thre}

We model Tokio's current thread runtime via the trace composition
rules below. The general approach follows~\autocite{HH2024}.

We modify the configuration to be of the form
$\ctCfg{\sh}{\mathit{active}}{t}{q_G}{q}$, where $\sh$ is the current
trace, $\mathit{active}$ is the currently active continuation stack or
$\Idle$ if the scheduler is idle, $t$ is the tick, $q_G$ is the global
queue, and $q$ is the local queue. See \zcref{sec:axiom-queue} for an
axiomatization of queues. The rules below model the relation $\toCT$
to differentiate it from other possible schedules.

To ensure a deterministic scheduling, which allows us to show
fairness, and to accurately model the current thread runtime, we
require many composition rules. \zcref[S]{fig:ct-rules} shows the
decision tree of the current thread runtime.

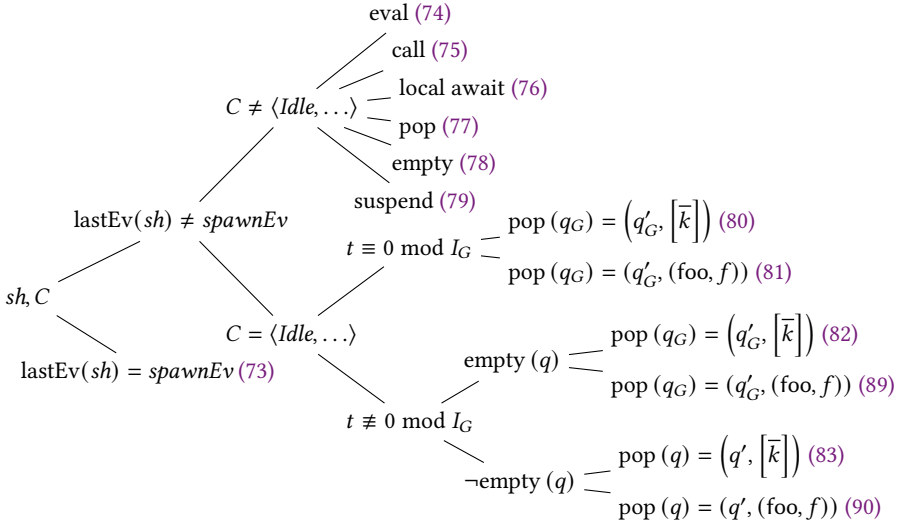
\begin{figure}[h]
  \small
  \centering
  \begin{tikzpicture}[every text node part/.style={align=center},
    level distance=5mm,
    level 1/.style={sibling distance=2cm},
    level 2/.style={sibling distance=7.9cm}
    level 3/.style={sibling distance=2.5cm}]
    \node {$\sh,C$} [grow=0,right]
    child[level distance=-2mm] {node {$\lastEv(\sh)=\textit{spawnEv}$ \zcref[noname]{eq:rule-ctSpawn}}}
    child {node {$\lastEv(\sh)\ne\textit{spawnEv}$}
      [sibling distance=3cm]
      child {node {$C=\langle\Idle,\ldots\rangle$}
        [sibling distance=23mm]
        child[level distance=6mm] {node {$t\not\equiv 0\mod I_G$}
          [sibling distance=16mm,level distance=6mm]
          child {node {$\neg\qEmpty{q}$}
            [sibling distance=7mm, level distance=12mm]
            child {node {$\qPop{q}=(q',(\pv{foo},f))$ \zcref[noname]{eq:rule-ctIdleLocalMarker}}}
            child {node {$\qPop{q}=\left(q',\left[\many{k}\right]\right)$  \zcref[noname]{eq:rule-ctIdleLocal}}}
          }
          child {node {$\qEmpty{q}$}
            [sibling distance=7mm,level distance=12mm]
            child {node {$\qPop{q_G}=(q'_G,(\pv{foo},f))$ \zcref[noname]{eq:rule-ctIdleLocalEmptyMarker}}}
            child {node {$\qPop{q_G}=\left(q'_G,\left[\many{k}\right]\right)$ \zcref[noname]{eq:rule-ctIdleLocalEmpty}}}
          }
        }
        child[level distance=6mm] {
          node {$t\equiv 0\mod I_G$}
          [sibling distance=7mm, level distance=12mm]
          child {node {$\qPop{q_G}=(q'_G,(\pv{foo},f))$ \zcref[noname]{eq:rule-ctIdleIntervalMarker}}}
          child {node {$\qPop{q_G}=\left(q'_G,\left[\many{k}\right]\right)$ \zcref[noname]{eq:rule-ctIdleInterval}}}
        }
      }
      child[level distance=5mm] {node {$C\ne\langle\Idle,\ldots\rangle$}
        [sibling distance=5mm, level distance=9mm]
        child[level distance=7mm] {node {suspend \zcref[noname]{eq:rule-ctActiveSuspend}}}
        child[level distance=12mm] {node {empty \zcref[noname]{eq:rule-ctActiveEmpty}}}
        child[level distance=13mm] {node {pop \zcref[noname]{eq:rule-ctActivePop}}}
        child[level distance=13mm] {node {local await \zcref[noname]{eq:rule-ctActivePush}}}
        child[level distance=12mm] {node {call \zcref[noname]{eq:rule-ctActivePushCall}}}
        child[level distance=9mm] {node {eval \zcref[noname]{eq:rule-ctActiveEval}}}
      }
    };
  \end{tikzpicture}
  \Description{A decision tree for the scheduling rules of the current
    thread runtime.}
  \caption{The current thread scheduler's rule structure. Read left to
    right.}
  \label{fig:ct-rules}
\end{figure}

If the last event was a spawn event, \zcref{eq:rule-ctSpawn} appends a
schedule event to the trace and, if the result is well-formed, adds a
\emph{task marker} $(\pv{foo}, f)$ to the end of the local
queue.\footnote{For brevity, we only model scheduling tasks from an
  active worker.} The well-formedness requirements above ensure that
the schedule event contains the correct function name. We only add
this event to ensure that afterwards no spawn event is the last event,
guaranteeing that \zcref{eq:rule-ctSpawn} is applied at most once per
scheduled future. Observe that the function \pv{foo} is coupled with
the future~$f$ via the create event in \zcref{eq:wf-sched-tokio}. We
write $\qPush{q}{k}$ for the result of enqueuing $k$ to $q$.
\gcondrule{ctSpawn}{
  \lastEv(\sh)=\spawnEv{f}{h} \quad \rust{\kwAsync~\kwFn foo($\many{x}$) $e$} \in \Fn \\
  \sh'={\consTr{\last(\sh)}{\consTr{\schedEv{\pv{foo}}{f}}{\sh}}} \quad \wf{\sh'}
}{
  \ctCfg{\sh}{\mathit{active}}{t}{q_G}{q} \toCT \ctCfg{\sh'}{\mathit{active}}{t}{q_G}{\qPush{q}{(\pv{foo},f)}}
}

\begin{example}\label{ex:tokio-run-spawn}
  We demonstrate the scheduling rules using \zcref{lst:async-prog}
  (\zcref[page]{lst:async-prog}). The scheduler starts in
  configuration
  $\ctCfg{\mkTr{\sigma_0}}{\left[\contP{f_\text{init}}{e_\text{main}}\right]}{0}{\queue}{\queue}$,
  where $\dom(\sigma_0)=\emptyset$ and $e_\text{main}$ is the main
  function's body.

  Assume the body has been executed so that the \lstinline|spawn| call
  in \zcref{line:call-two-power} was just evaluated, emitting a spawn
  event. Then, \zcref{eq:rule-ctSpawn} is the next rule to be applied
  and we arrive at the configuration
  \begin{align*}
    &\consTr{\sigma_1}{\consTr{\schedEv{\pv{two\_power}}{f_0}}{\consTr{\sigma_1}{\consTr{\cdots}{\mkTr{\sigma_0}}}}},\\
    &\ctState{\left[\contP{f_\text{init}}{\rust{\{\kwLet~f1 = H; $\ldots$\}}}\right]}{0}{\queue}{\qPush{\queue}{(\pv{two\_power},f_0)}}
  \end{align*}
\end{example}

We continue by presenting the rules that are analogous to the rules in
\zcref{sec:async-comp}, namely the rules where there is an active task
which is to be executed, possibly causing function calls or awaiting
local futures. Since an active task exists, $\mathit{active}$ is a
stack of continuation markers. Because the rules are so similar, we
\highlightRule{highlight} which elements of the rule have changed
compared to \zcref{sec:async-comp}.

If the top continuation marker in the stack contains a non-value
expression---which is equivalent to its local evaluation yielding a
non-empty set---we evaluate the expression one step, add the resulting
trace to the configuration, and replace the stack's top element with
the new continuation marker. \zcref[S]{eq:rule-ctActiveEval} adapts
\zcref{eq:rule-Eval}.
\gcondrule{ctActiveEval}{
  \highlightRule{\nexists f,h.\ \lastEv(\sh)=\spawnEv{f}{h} \quad \nexists f.\ \lastEv(\sh)=\localAwEv{f}} \\
  \sigma=\last(\sh) \quad \concatTr{\tau}{\cont{e'}} \in \valP{f}{e} \quad \rho\text{ concretizes }\tau \quad \wf{\chopTrSem{\sh}{\rho(\tau)}}
}{
  \ctCfgH{\sh}{\left[\contP{f}{e},\many{k}\right]}{t}{q_G}{q} \toCTH \ctCfgH{\chopTrSem{\sh}{\rho(\tau)}}{\left[\contP{f}{e'},\many{k}\right]}{t}{q_G}{q}
}

If the last event was an invocation event, we evaluate the called
function and push it onto the active task as in
\zcref{sec:trace-comp-rules,sec:async-comp}. We adapt
\zcref{eq:rule-PushCall} to obtain
\zcref{eq:rule-ctActivePushCall}. Because the last event is an
invocation event, it cannot be a spawn event; hence, we need no
explicit condition as in the previous rules.
\gcondrule{ctActivePushCall}{
  \highlightRule{\lastEv(\sh)=\textit{invEv}(\pv{foo},\ldots)} \quad
  \sigma=\last(\sh) \quad \rust{\kwFn foo($\many{x}$ $e$)}\in\mathit{Fn}
  \\ \concatTr{\tau}{\cont{e'}} \in \valP{f}{e} \quad
  \rho\text{ concretizes }\tau \quad \wf{\chopTrSem{\sh}{\rho(\tau)}}
}{
  \ctCfgH{\sh}{\left[\many{k}\right]}{t}{q_G}{q} \toCTH \ctCfgH{\chopTrSem{\sh}{\rho(\tau)}}{\left[\contP{f}{e'},\many{k}\right]}{t}{q_G}{q}
}

If the last event was a local await event, we evaluate the
corresponding function and push it onto the active task. Hence, we
model awaiting on local futures within the same task similar to
function calls in the synchronous
case. \zcref[S]{eq:rule-ctActivePush} is an adaptation of
\zcref{eq:rule-PushTask}.
\gcondrule{ctActivePush}{
  \highlightRule{\lastEv(\sh)=\localAwEv{f}} \quad
  \sigma=\last(\sh) \quad \rust{\kwAsync~\kwFn foo($\many{x}$ $e$)}\in\mathit{Fn}\\
  \concatTr{\tau}{\cont{e'}} \in \valP{f}{e} \quad
  \rho\text{ concretizes }\tau \quad \wf{\chopTrSem{\sh}{\rho(\tau)}}
}{
  \ctCfgH{\sh}{\left[\many{k}\right]}{t}{q_G}{q} \toCTH \ctCfgH{\chopTrSem{\sh}{\rho(\tau)}}{\left[\contP{f}{e'},\many{k}\right]}{t}{q_G}{q}
}

\begin{example}\label{ex:tokio-run-active-push}
  We continue with the running example. The local \kwAwait in
  \zcref{line:local-bar} pushes a new continuation marker for function
  \lstinline|two_power| on the active stack:
  \begin{align*}
    &\sh_1, \ctState{\left[\contP{f_1}{\rust{$\ldots$ $f_2$.\kwAwait~* 2$\ldots$}}\right]}{t_1}{\queue}{\qPush{\queue}{\left[\contP{f^\text{init}}{\ldots}\right]}}\\
    \to\,&\sh_2, \ctState{\left[\contP{f_2}{\ldots},\contP{f_1}{\rust{$\ldots$ ?$f_2$ * 2 $\ldots$}}\right]}{t_1}{\queue}{\qPush{\queue}{\left[\contP{f^\text{init}}{\ldots}\right]}}
  \end{align*}
  Note that the timer is not increased.
\end{example}

\zcref[S]{eq:rule-ctActivePop}, an adaptation of \zcref{eq:rule-Pop},
removes the top continuation marker of the active task if possible, as
this implies that we awaited a local future, which is now fully
evaluated, or we fully evaluated a synchronous call.
\gcondrule{ctActivePop}{
  \highlightRule{\nexists f,h.\ \lastEv(\sh)=\spawnEv{f}{h}} \quad \sigma=\last(\sh) \\
  \concatTr{\tau}{\cont{e''}} \in \valP{f'}{e'} \quad \rho\text{ concretizes }\tau \quad \wf{\chopTrSem{\sh}{\rho(\tau)}}
}{
  \ctCfgH{\sh}{\left[\contP{f}{e},\contP{f'}{e'},\many{k}\right]}{t}{q_G}{q} \toCTH \ctCfgH{\sh}{\left[\contP{f'}{e''},\many{k}\right]}{t}{q_G}{q}
}

\begin{example}\label{ex:tokio-run-active-pop}
  Continuing \zcref{ex:tokio-run-active-push}, we find that
  \zcref{eq:rule-ctActivePop} can be applied in the execution of
  \zcref{lst:async-prog} when the call to \lstinline|two_power| from
  \lstinline|bar| (\zcref{line:local-bar}) is fully computed. Then,
  the top continuation marker is popped off the stack.
\end{example}

If the active continuation stack is empty, we remove it and make the
scheduler idle. \zcref[S]{eq:rule-ctActiveEmpty} adapts
\zcref{eq:rule-Empty}.
\gcondrule{ctActiveEmpty}{
  \highlightRule{\nexists f,h.\ \lastEv(\sh)=\spawnEv{f}{h}}
}{
  \ctCfgH{\sh}{\left[\right]}{t}{q_G}{q} \toCTH \ctCfgH{\sh}{\Idle}{t}{q_G}{q}
}

When the expression of the active task is suspended in $\sh$, the the
worker is made idle and the stack is pushed to the end of the
queue. \zcref[S]{eq:rule-ctActiveSuspend} follows
\zcref{eq:rule-Suspend}.
\gcondrule{ctActiveSuspend}{
  \highlightRule{\nexists f,h.\ \lastEv(\sh)=\spawnEv{f}{h}} \quad \suspended{e}{\sh}
}{
  \ctCfgH{\sh}{\left[\contP{f}{e},\many{k}\right]}{t}{q_G}{q} \toCTH \ctCfgH{\sh}{\Idle}{t}{q_G}{\qPush{q}{\left[\contP{f}{e},\many{k}\right]}}
}
Observe that if \zcref{eq:rule-ctActiveSuspend} is applicable,
\zcref{eq:rule-ctActiveEval} cannot be applicable, as the resulting
trace cannot be well-formed if the computation event is missing.

\begin{example}\label{ex:tokio-run-active-suspend}
  We continue the running example by demonstrating suspension. After
  the initial task has spawned its sub-tasks
  (\zcref{line:call-two-power,line:call-bar}) and reached the first
  \lstinline|await| in \zcref{line:await-spawned}, it cannot continue
  until the first task is resolved. Hence, the initial task is
  suspended and \zcref{eq:rule-ctActiveSuspend} is applicable:
  \begin{align*}
    &\sh_3, \ctState{\left[\contP{f^\text{init}}{\ldots}\right]}{t_2}{\queue}{\qPush{\qPush{\queue}{(\pv{two\_power},f_0)}}{(\pv{bar},f_1)}}\\
    \to&\sh_4, \ctState{\Idle}{t_2}{\queue}{\qPush{\qPush{\qPush{\queue}{(\pv{two\_power},f_0)}}{(\pv{bar},f_1)}}{\left[\contP{f^\text{init}}{\ldots}\right]}}
  \end{align*}
\end{example}

Now we model the case of the idle worker, following the descriptions
in \zcref{sec:bg-tokio}. If \zcref{eq:rule-ctSpawn} is not applicable,
the tick is divisible by the global interval, the global queue is not
empty (pop is defined), and the configuration is idle, then
\zcref{eq:rule-ctIdleInterval} takes the first task of $q_G$, advances
the tick, and changes the active task to the popped one.
\gcondrule{ctIdleInterval}{
  \nexists f,h.\ \lastEv(\sh)=\spawnEv{f}{h} \quad t \equiv 0 \mod I_G \quad \qPop{q_G}=(q'_G,\left[\many{k}\right])
}{
  \ctCfg{\sh}{\Idle}{t}{q_G}{q} \toCT \ctCfg{\sh}{\left[\many{k}\right]}{t + 1}{q'_G}{q}
}

An analogous rule is required for the case that the first task in
$q_G$ is not a continuation marker but a task marker:
\zcref[S]{eq:rule-ctIdleIntervalMarker} evaluates the function of the
task marker one step, concretizes it, and makes it the active task.
\gcondrule{ctIdleIntervalMarker}{
  \nexists f,h.\ \lastEv(\sh)=\spawnEv{f}{h} \quad t \equiv 0 \mod I_G \\
  \qPop{q_G}=(q'_G,(\pv{foo},f)) \quad  \sigma=\last(\sh) \quad \lookup{\pv{foo}}=\rust{\kwAsync~\kwFn foo($\many{x}$) e} \\
  \concatTr{\tau}{\cont{e'}} \in \valP{f}{\rust{\kwAsync~\kwFn foo($\many{x}$) e}} \quad  \rho\text{ concretizes }\tau \quad \wf{\chopTrSem{\sh}{\rho(\tau)}}
}{
  \ctCfg{\sh}{\Idle}{t}{q_G}{q} \toCT \ctCfg{\chopTrSem{\sh}{\rho(\tau)}}{\left[\contP{f}{e'}\right]}{t + 1}{q'_G}{q}
}

\zcref[S]{eq:rule-ctIdleLocalEmpty} proceeds similarly in the case
where the local queue is empty: The first global task is taken and
executed.
\gcondrule{ctIdleLocalEmpty}{
  \nexists f,h.\ \lastEv(\sh)=\spawnEv{f}{h} \quad t \not\equiv 0 \mod I_G \quad
  \qEmpty{q} \quad \qPop{q_G}=(q'_G,[\many{k}])
}{
  \ctCfg{\sh}{\Idle}{t}{q_G}{q} \toCT \ctCfg{\sh}{\left[\many{k}\right]}{t + 1}{q'_G}{q}
}

Otherwise, the first local task is executed per
\zcref{eq:rule-ctIdleLocal}.
\gcondrule{ctIdleLocal}{
  \nexists f,h.\ \lastEv(\sh)=\spawnEv{f}{h} \quad t \not\equiv 0 \mod I_G \quad \qPop{q}=(q',[\many{k}])
}{
  \ctCfg{\sh}{\Idle}{t}{q_G}{q} \toCT \ctCfg{\sh}{\left[\many{k}\right]}{t + 1}{q_G}{q'}
}

Analogous rules to
\zcref{eq:rule-ctIdleLocalEmpty,eq:rule-ctIdleLocal}, where the task
from the queue is not a continuation marker but a task marker, are
defined as
\zcref{eq:rule-ctIdleLocalEmptyMarker,eq:rule-ctIdleLocalMarker} in
\zcref{app:addt-curr-thre}.

\begin{example}\label{ex:tokio-run-idle-task}
  We continue from \zcref{ex:tokio-run-active-suspend}: Assume the two
  tasks in \zcref{lst:async-prog} have been spawned
  (\zcref{line:call-two-power,line:call-bar}) but not yet started and
  the initial task is suspended while it waits for the spawned
  tasks. Then, the program is in an idle configuration and we can
  apply \zcref{eq:rule-ctIdleLocalMarker}, yielding the following
  transition:
  \begin{align*}
    &\sh_5, \ctState{\Idle}{t_3}{\queue}{\qPush{\qPush{\qPush{\queue}{(\pv{two\_power},f_0)}}{(\pv{bar},f_1)}}{\left[\contP{f^\text{init}}{\ldots}\right]}}\\
    \to&\sh_6, \ctState{\left[\contP{f_0}{e_2}\right]}{t_3+1}{\queue}{\qPush{\qPush{\queue}{(\pv{bar},f_1)}}{\left[\contP{f^\text{init}}{\ldots}\right]}}
  \end{align*}
  where $e_2$ is the expression left after evaluating the function
  body of \lstinline|two_power| one step and $\sh_2$ is the extension
  of $\sh_1$ with the resulting trace. Note that the time is
  incremented and the local queue is shortened by one element. The
  cases for starting function \lstinline|bar| and continuing the
  initial task are similar, though the latter uses
  \zcref{eq:rule-ctIdleLocal}.
\end{example}

The definition of the rules with their restrictive conditions ensures
a deterministic scheduler.

\begin{proposition}\label{prop:ct-determ}
  \zcref[S]{eq:rule-ctSpawn,eq:rule-ctIdleInterval,eq:rule-ctIdleIntervalMarker,eq:rule-ctIdleLocalEmpty,eq:rule-ctIdleLocal,eq:rule-ctActivePop,eq:rule-ctActivePush,eq:rule-ctActiveEval,eq:rule-ctActiveSuspend,eq:rule-ctIdleLocalEmptyMarker,eq:rule-ctIdleLocalMarker,eq:rule-ctActiveEmpty,eq:rule-ctActivePushCall}
  are mutually exclusive and model a \emph{deterministic},
  well-defined scheduler.
\end{proposition}

\subsection{Fairness}
\label{sec:fairness}

When we consider the fairness of the runtime defined above, the
question arises, which definition of fairness to examine. Strong
fairness, i.e., every expression which is schedulable infinitely often
is executed as some point, is too strong for our case.  As
demonstrated by \textcite{HH2024}, weak fairness, i.e., every
expression which is \emph{always} schedulable from some point onwards
is eventually executed, is inadequate for cooperative scheduling
(e.g., Tokio) and is even inadequate for the more straightforward case
of ProMeLa~\autocite{H2004}. This limitation is easy to see in the
runtime above: So long as we have an active task, the composition
rules permit no other task to be scheduled, even when it is
\enquote{ready.} However, weak fairness guarantees only hold for a
task that is \emph{continuously} schedulable. Thus, weak fairness
offers no guarantees for the Tokio runtime at all.

To remedy this limitation, \emph{quiescent fairness} is
introduced~\autocite{HH2024}, which is more suitable to cooperative
settings. Intuitively, it states that if a task is continuously
schedulable when there is no active task, it will be scheduled at some
future point. Hence, it is an adaption of weak fairness, as a task
need not be schedulable in every state but only on \emph{quiescent
  states}, i.e., states with no active task.

If we consider programs with no \emph{locally infinite computations,}
i.e., programs where at any point in time, there will be a quiescent
state in the future, we can formally state quiescent fairness of the
current thread scheduler. \zcref[S]{thm:ct-quiescent} assumes a
function \enquote{$\schedTransSym$,} translating a current thread
runtime configuration $\ctCfg{\sh}{\mathit{act}}{t}{q_G}{q}$ into a
configuration $\asCfg{\sh'}{\mathit{act}}{Q}$ as in
\zcref{sec:async-comp}. See \zcref{app:quiesc-fairn-curr} for a formal
treatment.

\begin{theorem}[Quiescent Fairness]\label{thm:ct-quiescent}
  Assume a program with no locally infinite computations. Starting
  from an initial configuration, the current thread scheduler is
  quiescent fair: If \[C_0 \toCT C_1 \toCT C_2 \toCT \cdots\] then
  there is a quiescent fair execution
  \[\mkTr{\sigma_0},\{\contP{f^\text{init}}{e}\} \to^\ast \schedTrans{C_1}
    \to^\ast \schedTrans{C_2} \to \cdots\] producing the same trace
  (excluding \textit{schedEv} events).
\end{theorem}

With the \ac{LAGC} approach presented in~\cite{HH2024},
\zcref{thm:ct-quiescent} is straightforward (albeit tedious) to prove;
see \zcref{app:quiesc-fairn-curr}. However, the Tokio documentation,
informally, claims a different notion of fairness for both
runtimes~\cite{tokio}: If, for a given \Rust
program~$r$,\footnote{This addition is not explicit in the
  documentation but it is crucial to make the fairness usable.} there
is some $T_\text{max}\in\mathbb{N}$ such that at any point the number
of tasks in the runtime is less or equal $T_\text{max}$ and there is
some $P_\text{max}\in\mathbb{N}$ that the evaluation of a
\lstinline|poll| call on a task in the scheduler finishes within
$P_\text{max}$ time units, then there is some
$D_\text{max}\in\mathbb{N}$ such that a task is scheduled in
$D_\text{max}$ time units.

\zcref[S]{def:tokio-fair} makes us of the predicate $\quiescent{C}$
over a configuration~$C$, which holds iff $C$ has no active task. The
function $\newExpr{C}{C'}$ returns a set of exactly those continuation
markers that are newly reachable after a step
$C \to C'$~\autocite{HH2024}. With function $\enabled{C}$ we get all
continuation markers which are reachable from one step of $C$:
$\enabled{C}:=\{\newExpr{C}{C'} \mid C \to C'\}$.

\begin{definition}[Tokio Fairness]\label{def:tokio-fair}
  Let $r$ be an asynchronous \Rust program. Assume there are
  $T_\text{max},P_\text{max}\in\mathbb{N}$ such that for all infinite
  sequences of \ac{LAGC} rule applications $C_0 \to C_1 \to \cdots$
  the following holds
  \begin{align*}
    \forall i.\ \mathrm{numTasks}(Q_i) \le T_\text{max}
  \end{align*}
  where $Q_i$ is the set of task stacks in the configuration $C_i$ and
  $\mathrm{numTasks}(Q)=\left|\left\{\left[\many{k}\right]\in Q\mid
      \left[\many{k}\right]\ne []\right\}\right|$. Assume further that
  all local computations finish in $P_\text{max}$ time units. Then the
  infinite sequence of rule applications $C_0 \to C_1 \to \cdots$ is
  \emph{Tokio fair} if there is a $D_\text{max}\in\mathbb{N}$ such
  that
  \begin{align*}
    &\forall m,e,f.\ \big(\forall n \ge m.\ \quiescent{\schedTrans{C_n}} \Rightarrow \contP{f}{e}\in\enabled{\schedTrans{C_n}}\big) \\
    &\quad\Rightarrow \big(\exists n'.\ m < n' \le m + D_\text{max} \wedge \contP{f}{e}\in\newExpr{\schedTrans{C_{n'}}}{\schedTrans{C_{n'+1}}}\big)
  \end{align*}
\end{definition}

We can show that the current thread runtime is Tokio fair as a result
of it being quiescent fair (\zcref{thm:quiescent-tokio}). In fact, we
show that each quiescent fair scheduler where a scheduling distance
with a well-founded order can be defined (as we do in the fairness
proof in \zcref{app:quiesc-fairn-curr}) is also Tokio fair.

\begin{theorem}\label{thm:quiescent-tokio}
  Assume an infinite sequences of \ac{LAGC} rule applications
  $C_0 \toCT C_1 \toCT \cdots$, then there is a Tokio fair execution
  $\mkTr{\sigma_0},\{\contP{f^\text{init}}{e}\} \to^\ast
  \schedTrans{C_1} \to \cdots$ producing the same trace (excluding
  \textit{schedEv} events).
\end{theorem}
\begin{proof}[Proof Sketch]
  Tokio fairness follows from quiescent fairness and the fact that we
  can construct a $D_\text{max}\in\mathbb{N}$ from the distance
  function we used to proof quiescent fairness in
  \zcref{app:quiesc-fairn-curr} in combination with $T_\text{max}$ and
  $P_\text{max}$. See \zcref{app:relat-quiesc-tokio} for a full proof.
\end{proof}

\subsection{A Comment on the Multi-Thread Runtime}
\label{sec:comment-multi-thread}

The present paper is focussed on the current thread runtime and we
omit a formalization and fairness proof of the multi-tread
runtime. While the multi-thread runtime is more widely used, we leave
it for future work.

Creating a formalization of the multi-thread runtime consists mainly
in creation of new trace composition rules; local evaluation and
well-formedness can be left unaltered. Developing these composition
rules from the existing rules is a relatively straightforward
process. The configuration needs to be extended to support a fixed
number of workers, each containing a local queue; the global queue is
left as is.

We do not present these composition rules for two reasons. First, for
space limitations. Second, because in the multi-threaded runtime, idle
workers may \emph{steal} tasks from other workers to spread the
workload more evenly. This stealing procedure is only described
conceptually, the concrete algorithm for it is left oblique. Study of
the complex code is a time consuming task.

More adaption would be necessary to prove fairness of the
multi-threaded runtime. Because multiple workers execute in parallel,
an additional counter atop the \enquote{local} counter in
\zcref{app:quiesc-fairn-curr} is required to ensure that the
composition rules do not \enquote{starve} a worker. Additionally,
parallelism causes the scheduler to be non-deterministic, which
complicates the fairness reasoning. It is possible to overcome these
hurdles, but space reasons do not permit an attempt in the present
work.


\section{Application: Soundness of a Dynamic Logic Calculus}
\label{sec:sound-dl}

One promising area of application for our semantics is the soundness
of program logic calculi. The development of program logics to provide
semantics for a language and proof correctness of programs go back to
\textcite{H1969}. Since then, \textcite{P1976} introduced \acf{DL}, an
extension of first-order logic, where programs can appear in formulas,
e.g., $\dlboxf{x = x + 1;}{(\pv{x} > 0)}$. \ac{DL} is closed under all
logical operators. Hence, any Hoare triple $\{\phi\} P \{\psi\}$,
where $P$ is a program and $\phi,\psi$ are formulas can be expressed
in \ac{DL} as $\phi \rightarrow \dlbox{P}{\psi}$. \ac{DL} is also
strictly more expressive than Hoare logic~\autocite{Reif65}. \ac{DL}s
have proven especially helpful for program verification close to the
source code level, as there is a reduced mental and semantical gap
between the program under verification and the formulas. The \ac{DL}
developed for Java as part of the \KeY verification
system~\autocite{ABBHSU2016}, JavaDL, has been especially
successful~\autocite{BBDHLPUW2025}.

As such program logic calculi are designed to be highly modular, i.e.,
typically have one rule per language construct, a modular semantics is
needed to show soundness.  So far, the discrepancy between calculus
design and semantics design hindered or outright prevented soundness
proofs for program logic calculi.

As no suitable semantics was defined, soundness of JavaDL's calculus
rules dealing with programs---mainly its \emph{symbolic execution}
rules---could not be shown, even though it is \enquote{\textins{t}he
  most important property of the JavaDL
  calculus}~\autocite{ABBHSU2016}. Since real-world programming
languages rarely have (official) formal semantics and program
semantics are typically difficult to relate to the semantics of a
logic, resolving this issue has been difficult. In this section, we
show how the modular semantics developed in
\zcref{sec:lagc-semant-rust} can be used for soundness proofs of a
\ac{DL} for \Rust.

In previous work~\autocite{DH2026,DH2026TR} we presented a \ac{DL}
called \enquote{\RustyDL{}} and a calculus to verify \Rust
programs. This calculus has been implemented in the \KeY system. While
\RustyDL's calculus rules have been validated by using examples, any
formal verification system needs a formal proof of soundness. As most
\Rust verification tools rely on translations to other
languages---e.g., Prusti~\autocite{ABFGMMPS2022} translating to
Viper~\autocite{MSS2016}---correctness proofs of the tools rely on a
semantics of both \Rust and the target language and the correct
relation between the two. \RustyDL on the other hand, can be proven
correct merely by showing that all its rules are sound.

In this section we only consider the sequential subset of \Rust
discussed in \zcref{sec:lagc-semant-rust}. While \ac{DL}s for
concurrent languages have been
proposed~\autocite{DBH2015,DO2015,LAGC23}, the non-determinicity of
concurrent language features significantly alters the semantics and
calculus. \zcref[S]{sec:synt-semant-rustydl,sec:calculus-rules}
largely recapitulate the results of \autocite{DH2026}, with minor
modifications for brevity and presentation, while
\zcref{sec:showing-soundness} contains our new contributions.

\RustyDL extends first-order logic with \emph{modalities} containing
\Rust program fragments. Hence, formula
$\phi:=\pv{x}\ge 0 \rightarrow \dlf{x += 1;}{(\pv{x} > 0)}$ expresses
that if \pv{x} is greater than or equal to $0$ in a state~$s$ and the
program \enquote{\rust{x += 1;}} is executed in $s$, then it
terminates in state~$s'$, in which $\pv{x}>0$ holds. We prove such a
formula valid, i.e., $\models \phi$, by deriving the sequent
$\seq{}{\phi}$ in our sequent calculus.

\subsection{Syntax and Semantics of \RustyDL}
\label{sec:synt-semant-rustydl}

The set of \emph{dynamic logic terms of sort~$A$}, called $\DLTrm{A}$,
is defined as for sorted FOL with an additional base case:
$\pv{x}\in\DLTrm{A}$ for all $\pv{x}\in\Var$ where \pv{x}'s type maps
to the sort~$A$.\footnote{Note that the set $\Var$ is the set of
  \emph{program variables}---called \enquote{$\PVSym$} in
  \autocite{DH2026}. It is distinct from the \emph{logical variables}
  used in quantifiers; quantifying over program variables is not
  permitted and neither is a logical variable occurring in a program.}
Similarly, the set of \emph{dynamic logic formulas} $\DLFor$ is
similar to FOL extended with $\dia{p}\phi$ (\enquote{diamond}) and
$\dlbox{p}\phi$ (\enquote{box}) where $p$ is a \emph{program fragment}
and $\phi\in\DLFor$. Note that $\DLFor$ is closed under all logical
operatives and nesting of modalities, e.g.,
$\dia{p_1}(\phi\wedge \dia{p_2}\psi)$ is a syntactically correct
formula.

To manage fine-grained state changes, we add the notion of
\emph{updates} (\zcref{def:upd}) to express uniform, total state
changes in our logic. These updates may modify program variables or
modify memory locations borrowed by a mutable reference.

\begin{definition}[Update, Update Application]\label{def:upd}
  We define the set $\Updates$ inductively by:
  \begin{enumerate}
  \item $(\pv{a}\upd t)\in\Updates$ for each program variable
    $a:A$ and each term $t$ of sort $A$,
  \item $(t_1\dupd t_2)\in\Updates$ for each term
    $t_1$ of sort $\RefM{A}$ and $t_2$ of sort $A$,
  \item $\skipUp\in\Updates$,
  \item $(u_1\parupd u_2)\in\Updates$ for all updates
    $u_1,u_2\in\Updates$,
  \item $(\applyUp{u_1}{u_2})\in\Updates$ for all updates
    $u_1,u_2\in\Updates$.
  \end{enumerate}
\end{definition}

We extend the inductive definition of terms with term
$\applyUp{u}{t}\in\DLTrm{A}$ for all updates $u\in\Updates$ and terms
$t\in\DLTrm{A}$. The inductive definition of formulas is extended with
formula $\applyUp{u}{\phi}$ for all updates $u\in\Updates$ and
formulas $\phi\in\DLFor$. Update applications represent state changes
to terms and formulas syntactically and change the state in which a
formula or term is evaluated.

\begin{example}\label{ex:update}
  Let \pv{a}, \pv{b} be program variables of type \lstinline|u64| and
  \pv{c} of type \lstinline|&mut u64|. Consider the following examples
  of updates:
  \begin{itemize}
  \item $\skipUp$ represents the effectless update.
  \item An \emph{elementary update} $\pv{a}\upd 2\cdot\pv{a}$
    expresses a state change where \pv{a}'s value is doubled. Hence,
    $\applyUp{\pv{a}\upd 2\cdot\pv{a}}{\pv{a}}$ is semantically
    equivalent to $2\cdot\pv{a}$.
  \item A \emph{parallel update}
    $(\pv{a}\upd\pv{b}\parupd\pv{b}\upd\pv{a})$ represents a state
    change where the values \pv{a} and \pv{b} are swapped. There is no
    interference between the updates in a parallel update. If there
    are multiple updates to the same program variable in one parallel
    update, only the last one affects the state.
  \item The \emph{mutable update} $\pv{c}\dupd\pv{a}+\pv{b}$ stands
    for the state change where the variable borrowed by \pv{c} is set
    to $\pv{a}+\pv{b}$. The value of \pv{c} itself is unchanged.
  \end{itemize}
\end{example}

\begin{example}
  Let $\pv{x}: \Int,\pv{y}: \Int, \pv{z} : \RefM{\Int}$ be program
  variables. Then the formula
  \[\applyUp{\pv{x}\upd\pv{y} \parupd \pv{y}\upd\pv{x} \parupd
      \pv{z}\dupd\pv{x}+1}{\dlf{x = y + *z;}{(\pv{x} \keyeq
        2\cdot\pv{x} + 1)}}\]
  expresses that, in a state where the values of \pv{x} and \pv{y} are
  swapped and the value referenced by \pv{z} is updated to \pv{x} plus
  one, the program fragment \enquote{\rust{x = y + *z;}} terminates
  and afterwards the formula $\pv{x} \keyeq 2\cdot\pv{x} + 1$
  holds. The formula is valid if the addition in the program does not
  overflow.
\end{example}

Mutating updates allow our logic to represent mutable references. Just
as in the semantics, mutable references need to contain the borrowed
memory location, i.e., the borrowed place.  We introduce the sort
$\Place$ for places. For each program variable \pv{x}, its place is
represented with the unique constant $\place{\pv{x}} : \Place$.  With
the function $\refMSym$, mutable references of a place~$p$ can be
represented as a term $\refM{p}$. If $p$ is a place of a lender of
sort~$A$, then $\refM{p}$ is of sort $\RefM{A}$.

Formally, we define the semantics of \RustyDL using a \emph{Kripke
  structure} consisting of \emph{states}, i.e., first-order models,
and a program transition function $\rho$.

\begin{definition}[\RustyDL Kripke structure]\label{def:kripke}
  Let $r$ be a valid \Rust program. A \emph{\RustyDL Kripke structure}
  for a signature derived from $r$ is a tuple
  \[
    \Kstruct=(\domain,\interp,\States,\rho)
  \]
  where
  \begin{itemize}
  \item $\domain$ is the logic's domain, i.e., a non-empty set. This
    domain can be split into different sub-domains per sort, e.g.,
    $\domain^{\Int}$ for integers.
  \item $\interp$ is the interpretation of functions and predicates.
  \item $\States$ is an infinite set of \emph{states}
    $s:\Var\rightarrow\domain$.
  \item $\rho$ is a function mapping a program fragment $p$ to a
    partial function $\States\rightharpoonup\States$ with
    $\rho(p)(s_1)=s_2$ iff evaluating $p$ in state $s_1$ terminates
    normally in $s_2$. If $p$ does not terminate normally starting in
    $s_1$, then $\rho(p)(s_1)\undefSym$. We assume \Rust programs to
    be deterministic, which is true of the subset considered in
    \zcref{sec:lagc-semant-rust}.
  \end{itemize}
\end{definition}

Whereas \autocite{DH2026} left $\rho$ unspecified due to a lack of a
formal semantics, we can now use the \ac{LAGC} semantics to close this
gap. For this, we must clarify the relation between \ac{LAGC}
states~$\sigma$ and \RustyDL states~$s$. We assume a relation
$\dlLagcSim$ between \RustyDL states and \ac{LAGC} states (see
\zcref{def:corr-state} for a formal definition), pronounced
\enquote{corresponds to,} in the expected manner: integer, boolean,
and compound values map to each other in the expected way; references,
however, have no borrow identifier on the \RustyDL side, hence we
obtain a many-to-one mapping.

\begin{proposition}
  Let $\Kstruct=(\domain,\interp,\States,\rho)$ be \RustyDL Kripke
  structure and $s\in\States$ a \RustyDL state. There exist an
  infinite number of LAGC states $\sigma$ such that $s\sim\sigma$. All
  states in the set $\{\sigma\mid s\sim\sigma\}$ differ only in the
  borrow identifiers in shared and mutable references.
\end{proposition}

With this relation, we proceed to define $\rho$ in terms of the
generated traces. A panic does not count as normal termination and is
therefore excluded from the program transition in \zcref{def:rho}.

\begin{definition}[Program Transition Function]\label{def:rho}
  Let $r$ be a valid \Rust program,
  $\Kstruct=(\domain,\interp,\States,\rho)$ a Kripke structure for
  $r$, $p$ a program fragment for $r$, $s\in\States$ a \RustyDL state,
  and $\sigma$ an \ac{LAGC} state such that $s\sim\sigma$. Then the
  program transition function $\rho$ is defined as
  \[
    \rho(p)(s)=\begin{cases}
      s'&\text{if }\exists\sigma'.\ s'\dlLagcSim\sigma'\wedge\exists\sh.\ (\consTr{\sigma'}{\sh})\in\traces{e,\sigma} \wedge \lastEv(\sh)\ne\panicEv\\
      \undefSym&\text{otherwise}
    \end{cases}
  \]
  where $e:=\{\many{\rust{\kwLet}\;\pv{x};}\;p\}$ and
  $\many{\pv{x}}:=\freeProgVars(p)$ are the free program variables in
  $p$. In the following, we write $\cont{p}$ and $\traces{p,\sigma}$
  for $\cont{e}$ and $\traces{e,\sigma}$, respectively.

  By \zcref{lem:determ}, $\sigma'$ is uniquely identified. Our
  definition for $s\sim\sigma$ ensures that $s$ is uniquely identified
  as well.
\end{definition}

We must also fix the semantics of updates, as they are a crucial part
of the calculus rules.

\begin{definition}[Semantics of \RustyDL updates]\label{def:upd-sem}
  Let $\Crate$ be a crate, $\Kstruct=(\domain,\interp,\States,\rho)$ a
  Kripke structure for $\Crate$, $s,s'\in\States$ a state,
  $\pv{a}:A,\pv{b}:B\in\Var$ program variables,
  $\beta:\VSym\rightarrow \domain$ a variable assignment,
  $u_1,u_2\in\Updates, t,t_2\in\DLTrm{A}$, and
  $t_1\in\DLTrm{\RefM{B}}$. The valuation function
  $\valS:\Updates\rightarrow(\States\rightarrow\States)$ is defined
  as:
  \begin{align*}
    \valS(\pv{a}\upd t)(s')(\pv{b})=&
                                      \begin{cases}
                                        \valS(t) &\text{if}~\pv{b}=\pv{a}\\
                                        s'(\pv{b})&\text{otherwise}
                                      \end{cases}\\
    \valS(t_1\dupd t_2)(s')(\pv{b})=&
                                      \begin{cases}
                                        \valS(t_2)&\text{if}~\valS(t_1)=\valS(\refM{\place{\pv{b}}})\\
                                        s'(\pv{b})&\text{otherwise}
                                      \end{cases}\\
    \valS(\skipUp)(s')=&\,s' \\
    \valS(u_1\parupd u_2)(s')=&\,\valS(u_2)(\valS(u_1)(s')) \\
    \valS(\applyUp{u_1}u_2)=&\,\valSPrime(u_2) \text{ where }s'=\valS(u_1)(s)
  \end{align*}
\end{definition}

With updates and program transition defined, the semantics of \RustyDL
terms and formulas is now definable.

\begin{definition}[Semantics of \RustyDL Terms and
  Formulas]\label{def:dl-sem}
  Let $\Crate$ be a valid \Rust crate, $\VSym$ the set of logic
  variable symbols, $\Kstruct=(\domain,\interp,\States,\rho)$ a Kripke
  structure modeling \Rust programs, $s\in\States$ a state, and
  $\beta:\VSym\rightarrow\domain$ a variable assignment.

  For every \RustyDL term $t$ we define its evaluation inductively by
  the evaluation function $\valS$ exactly as in first-order logic,
  with evaluation of program variables defined as
  $\valS(\pv{a})=s(\pv{a})$ for each program variable~\pv{a} and
  $\valS(\applyUp{u}{t})=\valSPrime(t)$ where $s'=\valS(u)(s)$.
  
  For every \RustyDL formula $\phi$ we define satisfaction, denoted
  $\valid\phi$, exactly as in first-order logic with three new cases:
  \begin{itemize}
  \item $\valid\dlbox{p}\phi$ iff $\rho(p)(s)\undefSym$ or
    $\rho(p)(s)=s'$ and $\validSPrime\phi$,
  \item $\valid\dia{p}\phi$ iff $\rho(p)(s)=s'$ exists and
    $\validSPrime\phi$.
  \item $\valid\applyUp{u}{\phi}$ iff $\validSPrime\phi$ where
    $s'=\valS(u)(s)$.
  \end{itemize}
\end{definition}

\subsection{Calculus Rules}
\label{sec:calculus-rules}

To show the validity of \RustyDL formulas, we rely on a sequent
calculus over our logic, a standard approach for verification and
theorem proving.

Beyond the familiar rules for propositional and first-order logic, as
well as rules for reasoning about, e.g., integer and boolean values,
the calculus also contains a large number of \emph{symbolic execution
  rules}. These allow the unfolding and decomposition of complex
expressions and statements into simpler ones until the entire program
has been executed and only first-order formulas remain.

To represent the symbolic execution calculus in a compact and
understandable manner, a further extension of the \Rust syntax is
required. In program fragments in rules---and only there---we can use
$\pi$ and $\omega$ to represent the \emph{inactive prefix} and the
\emph{rest of the program}, respectively. The inactive prefix contains
any opening braces and similar structures that represent the context;
$\omega$ stands for the rest of the program, including closing
braces. Hence, the schematic program fragment
$\pi~\rust{x = 1}~\omega$ matches any program fragment where \rust{x =
  1} is the first \emph{active} expression, e.g., \rust{\{\{x = 1;
  ()\} x = 2\}}.\footnote{The notion an active expression is related
  to the next sub-expression discussed in \zcref{sec:eval-sub-expr},
  but it is not identical.}

We present these rules in a schematic form with schema variables for
different parts of the sequent. In the rule $\ruleName{ifElseSplit}$
below, $\Gamma,\Delta$ stand for sets of \ac{DL} formulas, $\phi$ for
a \ac{DL} formula, $\upl\updateContext\upr$ stands for any potential
update application (it may also be empty; no update is necessary for
the rule to be applicable), $\mathit{se}$ stands for a \enquote{simple
  expression,} i.e., a program variable or literal,
$\{\mathit{be}_0\},\{\mathit{be}_1\}$ stand for block expressions, and
$\pi$ and $\omega$ stand for the inactive program prefix and the
program rest, respectively.

\[
  \seqRule{ifElseSplit}%
  {\sequent{\applyUp{\updateContext}{(\mathit{se}\keyeq\keybooltrue)}}{\applyUp{\updateContext}{\ptstub{\{$\mathit{be}_0$\}}}} \\
    \sequent{\applyUp{\updateContext}{(\mathit{se}\keyeq\keyboolfalse)}}{\applyUp{\updateContext}{\ptstub{\{$\mathit{be}_1$\}}}}}
  {\sequent{}{\applyUp{\updateContext}{\ptstub{\kwIf
          $\mathit{se}$
          \{$\mathit{be}_0$\} \kwElse \{$\mathit{be}_1$\}}}}}
\]

The rule above splits the sequent when an \lstinline|if| expression
with a side-effect-free condition is encountered into two branches:
one where the condition is true and one where it is false.

Crucial rules of our calculus simplify the program in modalities and
extract singular effects into updates. For instance, the following
rule $\ruleName{copy}$ transforms an assignment, where the right-hand
side is a simple expression, into an update and leaves behind a unit
expression.
\[
  \seqRule{copy}%
  {\sequent{}{\applyUp{\updateContext}{\upl\pv{x}\upd\mathit{se}\upr\ptstub{()}}}}
  {\sequent{}{\applyUp{\updateContext}{\ptstub{x = $\mathit{se}$}}}}\quad
  \mbox{$\mathit{se}$ is a literal or implements \lstinline|Copy|}
\]

For details on \RustyDL and more calculus rules, we refer to
\autocite{DH2026,DH2026TR}.

\subsection{Showing Soundness}
\label{sec:showing-soundness}

Following the standard definitions of the sequent calculus, we define
the soundness of \RustyDL rules as follows:

\begin{definition}[Soundness of Rules]\label{def:soundness}
  A \RustyDL rule with $n$ premisses of form
  $\seq{\Gamma_i}{\Delta_i}$ ($1\le i\le n$) and conclusion
  $\sequent{}{}$ is sound iff $\models\sequent{}{}$ under the
  assumption that $\models\seq{\Gamma_i}{\Delta_i}$ for all
  $1\le i\le n$.
\end{definition}

From \zcref{def:dl-sem}, we conclude that only the final state of a
trace is relevant for the \RustyDL semantics. Hence, when proving
soundness of a rule with a modality, it suffices to show that two sets
of traces generated by two modalities agree on some suffix. For this
purpose, \zcref{def:tr-suffix} defines the suffix relations of trace
sets and \zcref{lem:suffix-impl} relates this to validity in \RustyDL.

\begin{definition}[Trace Suffix]\label{def:tr-suffix}
  Let $T_1, T_2$ be sets of traces. We call $T_1$ the \emph{suffix
    set} of $T_2$, denoted as $\trSuff{T_1}{T_2}$, iff there is a
  finite shining trace~$\sh$ such that
  $T_2=\{\chopTrSem{\sh}{\sh'} \mid \sh'\in T_1\}$. We call $\sh$ the
  \emph{prefix.}
\end{definition}

\begin{lemma}\label{lem:suffix-impl}
  Let $\Kstruct$ be a \Rust Kripke structure, $\beta$ a variable
  assignment, $s_1,s_2$ \RustyDL states, $\sigma_1,\sigma_2$ \ac{LAGC}
  states with $s_1\dlLagcSim\sigma_1$ and $s_2\dlLagcSim\sigma_2$,
  $\phi$ a \RustyDL formula, and
  $p_1=\pi~e_1~\omega, p_2=\pi~e_2~\omega$ two program fragments. If
  $\validSTwo\dia{p_2}{\phi}$ and
  $\trSuff{\traces{p_2,\sigma_2}}{\traces{p_1,\sigma_1}}$, then
  $\validSOne\dia{p_1}{\phi}$.
\end{lemma}
\begin{proof}
  We assume $\validSTwo\dia{p_2}{\phi}$. By
  \zcref{def:rho,def:dl-sem}, we know that there is an \ac{LAGC}
  state~$\sigma$ and a state~$s$ such that at least one trace in
  $\traces{p_2,\sigma_2}$ ends in $\sigma$ and $s\dlLagcSim\sigma$
  such that $\valid\phi$. By \zcref{lem:determ}, all traces in
  $\traces{p_2,\sigma_2}$ end in~$\sigma$.

  Since $\trSuff{\traces{p_2,\sigma_2}}{\traces{p_1,\sigma_1}}$, we
  know that all traces in $\traces{p_1,\sigma_1}$ end in
  $\sigma$. Thus, \zcref{def:dl-sem} gives us
  $\validSOne\dia{p_1}{\phi}$.
\end{proof}

We can now begin proving our calculus rules sound.

\begin{theorem}[Soundness of rule
  $\ruleName{copy}$]\label{thm:sound-copy}
  Let $\Gamma,\Delta\in 2^{\DLFor}$ be sets of formulas,
  $\updateContext\in\Updates$ be an update, \pv{x} be a program
  variable whose type is \lstinline|Copy|, $\mathit{se}$ be a variable
  or a literal expression, $\pi~\rust{x = $\mathit{se}$}~\omega$ be a
  program fragment, and $\phi\in\DLFor$ be a formula. If
  \[
    \models\sequent{}{\applyUp{\updateContext}{\applyUp{\pv{x}\upd\mathit{se}}{\ptstub{()}}}}
  \]
  then the following holds:
  \[
    \models\sequent{}{\applyUp{\updateContext}{\ptstub{x = $\mathit{se}$}}}.
  \]
\end{theorem}
\begin{proof}
  We assume
  $\models\sequent{}{\applyUp{\updateContext}{\applyUp{\pv{x}\upd\mathit{se}}{\ptstub{()}}}}$.
  It is sufficient to show
  \[
    \validSOne\ptstub{x = $\mathit{se}$}
  \]
  for all $\Kstruct, s, \beta$ with $s_1=\valS(\updateContext)(s)$. We
  infer from the assumption that
  \begin{align}
    \validSTwo\ptstub{()}\label{eq:copy-asm-2}
  \end{align}
  with $s_2=\valSOne(\pv{x}\upd\mathit{se})(s_1)$.

  Let $\sigma_1,\sigma_2$ be LAGC states such that $s_1\sim\sigma_1$
  and $s_2\sim\sigma_2$.
  
  We know that \rust{x = $\mathit{se}$} is the active expression. By
  the construction of the local evaluation rules, we know that it is
  also the next (sub-)expression to be evaluated and added to the
  trace $\sh$ using \zcref{eq:rule-sync-base}. Thus we can use the
  local evaluation. We split over the construction of
  $\mathit{se}$. Since it is a simple expression, it is either a value
  or a variable.
  \begin{case}[$\mathit{se}\in\Val$]\label{case:copy-1}
    Then
    $\valB{\sigma_1}{\rust{x =
        $\mathit{se}$}}=\{\concatTr{\consTr{\update{\sigma_1}{\rust{x}}{\mathit{se}}}{\mkTr{\sigma_1}}}{\cont{()}}\}$.
    Hence the only possible composition is,
    $\mkTr{\sigma_1},\cont{\pi~\rust{x =
        $\mathit{se}$}~\omega}\rightarrow\consTr{\update{\sigma_1}{\rust{x}}{\mathit{se}}}{\mkTr{\sigma_1}},\cont{\pi~\rust{()}~\omega}$. As
    $s_2=\valSOne(\pv{x}\upd\mathit{se})(s_1)$, we know
    $\sigma_2=\update{\sigma_1}{\pv{x}}{\mathit{se}}$.
    
    Thus,
    \( \traces{\pi~\rust{x =
        $\mathit{se}$}~\omega,\sigma_1}=\{\concatTr{\mkTr{\sigma_1}}{\shi}\mid\shi\in\traces{\pi~\rust{()}~\omega,\sigma_2}\}
    \).  Clearly, then,
    $\trSuff{\traces{\pi~\rust{()}~\omega,
        \sigma_2}}{\traces{\pi~\rust{x =
          $\mathit{se}$}~\omega,\sigma_1}}$. By
    \zcref{lem:suffix-impl} and assumption~\eqref{eq:copy-asm-2}, we
    have \( \validSOne\ptstub{x = $\mathit{se}$} \).
  \end{case}
  \begin{case}[$\mathit{se}\in\Var$]
    Let $v=s_1(\mathit{se})$. Since $\mathit{se}$ is not symbolic (our
    semantics only generates concrete traces with concrete states),
    $v\in\Val$. Then
    $\valB{\sigma_1}{\mathit{se}}=\{\concatTr{\mkTr{\sigma_1}}{\cont{v}}\}$
    and
    $\valB{\sigma_1}{\rust{x =
        $\mathit{se}$}}=\{\concatTr{\tau}{\cont{\rust{x =
          $e$}}}\mid\concatTr{\tau}{\cont{e}}\in\valB{\sigma_1}{\mathit{se}}\}
    =\{\concatTr{\mkTr{\sigma_1}}{\cont{\rust{x = $v$}}}\}$.
   
      Further, we have
      $\valB{\sigma_1}{\rust{x =
          $v$}}=\{\concatTr{\consTr{\update{\sigma_1}{\rust{x}}{v}}{\mkTr{\sigma_1}}}{\cont{()}}\}$.
   
    Hence, the only possible composition is
    \[
      \mkTr{\sigma_1},\cont{\pi~\rust{x =
          $\mathit{se}$}~\omega} \to \mkTr{\sigma_1},\cont{\pi~\rust{x
          =
          $v$}~\omega} \to
      \consTr{\update{\sigma_1}{\rust{x}}{v}}{\mkTr{\sigma_1}},\cont{\pi~\rust{()}~\omega}
    \]
    Since $s_2=\valSOne(\pv{x}\upd\mathit{se})(s_1)$, we know that
    $\sigma_2=\update{\sigma_1}{\pv{x}}{v}$.
    
    Thus, we have
    \( \traces{\pi~ \rust{x =
        $\mathit{se}$} ~\omega,\sigma_1}
    =\{\concatTr{\consTr{\sigma_2}{\mkTr{\sigma_1}}}{\shi}\mid\shi\in\traces{\pi~\rust{()}~\omega,\sigma_2}\}
    \) which, with \zcref{def:rho}, assumption~\eqref{eq:copy-asm-2},
    and the argument in \zcref{case:copy-1}, gives us
    \( \validSOne\ptstub{x = $\mathit{se}$} \).\qedhere
  \end{case}
\end{proof}

The modularity of \ac{LAGC} thus permits relatively concise and
schematic soundness proofs. This ease carries over even to more
complex rules, as both the semantics and calculus are on a similar
level of modularity and there is, thus, a suitable correlation. We
present proofs for additional rules encountered in \cite{DH2026} in
\zcref{sec:addit-soundn-proofs}. The rules presented and proven here
and in \zcref{sec:addit-soundn-proofs} present some of the core and
most important rules of our calculus. Already, the rules support
references and loops. Given the schematic nature of the proofs,
proving the remaining calculus rules presents time consuming but
straightforward work.  These results lead us to postulate
\zcref{thm:sound-calc} with considerable confidence.

\begin{theorem}[Soundness of \RustyDL calculus]\label{thm:sound-calc}
  The calculus of \RustyDL is sound.
\end{theorem}

Thus, \RustyDL is, to our knowledge, the first \ac{DL} calculus for a
real-world programming language, where
\textcquote{ABBHSU2016}{\textins{t}he most important property of the
  \textelp{} calculus} has been proven.


\section{Related Work}
\label{sec:related-work}

For related work, we first consider other semantics and formalizations
of parts of \Rust and compare them to \ac{LAGC} in approach and
scope. Then, we consider other applications of the \ac{LAGC} approach
and how they differ from our \Rust semantics. Finally, we comment on
different \ac{LAGC}-like applications to deductive verification.

\subsection{\Rust Semantics}
\label{sec:rust-semantics}

Multiple semantics for substantial parts of \Rust have been defined in
recent years, all with different scopes, focal points, and
drawbacks. We present a broad selection here.

In contrast to the approaches discussed below, our semantics targets
the high-level \ac{HIR} representation and abstracts away from memory
representations of the data and other low-level concerns
indistinguishable for a \Rust programmer. Most existing work targets
the lower-level \ac{MIR}; the differences between these two layers are
substantial~\autocite{RustcGuide}, resulting in a significant semantic
gap between many existing formalizations and source code level
\Rust. As a consequence, we also leave other concerns such as type
checking and borrow checking aside. Rather than formalizing the borrow
checker properties in our semantics---which would be possible---we
take \Rust's guarantees for granted and model the operational behavior
only. For similar reasons, we only consider safe \Rust with its
stronger guarantees, whereas other approaches also cover unsafe
\Rust. Lastly, our semantics is the only one to date with support for
async \Rust. Hence, a model of a runtime as in \zcref{sec:fair-tokio}
and fairness proofs are not possible in other
semantics. \zcref[S]{tab:semantics} offers a comprehensive overview of
the existing semantics.

\textcite{JJKD2018,JJKD2017} introduced RustBelt and set out to prove
soundness of \Rust's safety guarantees and to formalize the conditions
under which a \Rust library is a \enquote{safe extension} of the
language. To accomplish this, the authors define a semantic model of a
subset of \Rust called $\lambda_\text{\Rust}$ and implement this model
using Iris~\autocite{JSSSTBD2015}. This approach has been extended to
also consider properties other than safety~\autocite{MDJD2022}. Given
their scope, $\lambda_\text{\Rust}$ is more closely aligned with
\Rust's \ac{MIR}. Furthermore, \citeauthor{JJKD2018} are concerned
with the type system's semantics, i.e., how types are modeled and
inhabited, while we take a more abstract view, focused on the
operational semantics of expressions. RustBelt supports unsafe \Rust
and all relevant data types. As $\lambda_\text{\Rust}$ is a version of
the $\lambda$-calculus, the semantics is based on translation. Enums,
e.g., are modeled via sum types.

The Mini\Rust project~\autocite{MiniRust} provides a specification for
a \ac{MIR}-like language with the goal of modeling the core of
\Rust. Though unofficial, it covers a crucial subset of the \ac{MIR}
and serves as a more formal \ac{FLS}, i.e., a formal specification of
\Rust at a lower level. Its formalization of references and data types
is more explicit than ours in that Mini\Rust, covering also unsafe
\Rust, tracks aliases. MiniRust's \ac{BOS} is not presented in a
mathematical style like an \ac{LAGC} semantics but a programmatic
style. It supports all data types and expressions, though loops are
not explicitly modeled at this level, as they have been translated to
basic blocks.

Similar to Mini\Rust but more comprehensive is the a-mir-formality
project~\autocite{AMF}, which is part of the \Rust foundation's
efforts. It also aims to cover the type and borrow checker and the
trait system. As it is currently under development, we cannot offer a
more detailed comparison at this point.

\textcite{HFP2024,HFP2024a} present a semantics for a \Rust subset
called \ac{LLBC} in the context of the verification tool
Aeneas~\autocite{HP2022}. Their semantics is based on translation to
multiple, successively lower-level languages and can serve as model of
\Rust's borrow checker. It supports all relevant types and
expressions, though as translations to other languages.

\textcite{P2021} offers an \ac{SOS}-style semantics for a \Rust
subset, \ac{FR}. This semantics also covers references, ownership, and
a version of borrow checking. Of the compound types, only tuples are
explicitly covered; others can be translated to suitable
representations. Loops are discussed and an approach outlined, but no
rules are provided.

\def\supports{\bullet}
\def\hsupports{\circ} 
\def\nsupports{-}

\begin{table}[h]
  \centering
  \begin{tabular}{lllcccccccc}
    \toprule
    Semantics & Level & Style & Tuple & Array & Struct & Enum & Loop & Async & Unsafe & BC \\\midrule
    RustBelt~\autocite{JJKD2017} & \ac{MIR} & \ac{SOS} & $\hsupports$ & $\hsupports$ & $\hsupports$ & $\hsupports$ & $\hsupports$ & $\nsupports$ & $\bullet$ & $\bullet$ \\
    MiniRust~\autocite{MiniRust} & \ac{MIR} & \ac{BOS} & $\supports$ & $\supports$ & $\supports$ & $\supports$ & $\hsupports$ & $\nsupports$ & $\bullet$ & $\bullet$ \\
    \ac{LLBC}$^\#$~\autocite{HFP2024} & \ac{MIR} & \ac{BOS} & $\hsupports$ & $\hsupports$ & $\hsupports$ & $\hsupports$ & $\hsupports$ & $\nsupports$ & $\nsupports$ & $\bullet$ \\
    \ac{FR}~\autocite{P2021} & Src & \ac{SOS} & $\supports$ & $\hsupports$ & $\hsupports$ & $\hsupports$ & $\hsupports$ & $\nsupports$ & $\nsupports$ & $\bullet$ \\
    Rust \ac{LAGC} & \ac{HIR} & \ac{LAGC} & $\bullet$ & $\bullet$ & $\bullet$ & $\bullet$ & $\bullet$ & $\bullet$ & $\nsupports$ & $\nsupports$ \\\bottomrule
  \end{tabular}
  \caption{Comparison of \Rust semantics. We denote a feature being
    supported by~$\supports$, unsupported by~$\nsupports$, and
    partially supported (e.g., by translation to sum types or lambda
    functions) by $\hsupports$. \enquote{BC} stands for the borrow
    checker being modeled; \enquote{Src} stands for source code
    level.}
  \label{tab:semantics}
\end{table}

There are formal models of \Rust ownership and reference system
designed for particular perspectives. \textcite{JDKD2019} present
stacked borrows, a memory model of borrowing in terms of stacks, which
has been extended to more permissive tree borrows by
\textcite{VHDJ2025}. Even more encompassing are the formalizations of
\Rust's upcoming borrow checker Polonius
\autocite{M2018,M2023,M2023a,RM2023,Polonius}. As we, at present,
offer no model of borrowing, these works can be seen as orthogonal:
their semantics can be incorporated into our semantics by using events
and the $\mathit{wf}$ predicate. The work of \textcite{CGK2023} is
also relevant, as they provide a conceptual model of ownership and
borrowing along with a concise formalization of Polonius.

\subsection{\ac{LAGC} Semantics}
\label{sec:rel-lagc-semantics}

The \ac{LAGC} approach has been applied to a number of programming
languages. In the original \ac{LAGC} paper, \textcite{LAGC23} defined
semantics for a simple \textsc{While} language, for a ProMeLa-like
language~\autocite{H2004}, and for a subset of the asynchronous
modeling languages ABS~\autocite{JHSSS2011} and similar active object
languages~\autocite{dDFHHJKRSSY2017}. The latter has been expanded by
\citeauthor{T2022a} to a semantics for Real-Time
ABS~\autocite{T2022a}. Other examples include Timed Rebeca by
\textcite{HJT2025} and the smart contract modeling language SmartML by
\textcite{R2024a}. The breadth of these languages demonstrate the
flexibility of \ac{LAGC} and especially its usability for concurrency.

Nevertheless, all prior applications of the \ac{LAGC} approach were
chiefly to small subsets of languages or modeling languages;
real-world programming languages were not approached. One result of
this limited scope is that expressions had so far been assumed to be
side-effect free and pure. We provide the first \ac{LAGC} semantics
that evaluates expressions, not just statements. In fact, this leads
to a slight simplification of the approach: path conditions, a central
element of symbolic traces in prior semantics, can be dropped
entirely.

Another avenue for \ac{LAGC} research is the fairness of schedulers
for programs with cooperative scheduling. \textcite{HH2024} formalized
a deterministic and quiescent fair scheduler for a language inspired by
ABS~\autocite{JHSSS2011} in \ac{LAGC} and proved its fairness. Our
work builds on their approach and adapts it to the Tokio scheduler.

\subsection{Deductive Verification based on \ac{LAGC}}
\label{sec:deduct-verif}

Applying \ac{LAGC} to devise a sound program logic was a core design
goal of \textcite{LAGC23} and part of their work. However, their logic
was limited first by their target languages being modeling languages,
as discussed above, and second by the fact that their logics lack
implementation. \textcite{BGHS2023} extend this trace-based deductive
verification approach by devising a logic and calculus for a more
sophisticated setting. It has since been further refined by
\textcite{HKS2024,S2026}. An implementation of this calculus has since
been developed, though only the sequential parts are supported as
yet~\autocite{S2026a}.

Our work is set apart from the existing \ac{LAGC} verification
approaches, since we did not design a program logic and calculus based
on a semantics or based on traces, but used a trace-based semantics to
validate the soundness of an existing calculus~\autocite{DH2026}. We
demonstrate that even when a calculus is independent of an \ac{LAGC}
semantics, soundness proofs follow quite readily from the semantics
due to its modular nature. We find that \ac{LAGC} with its local
evaluation rules fits naturally onto a \ac{DL} calculus, even for
complex languages such as \Rust.

Somewhat related is the work of \textcite{BB2013}, where they develop
a \ac{DL}, called \emph{dynamic trace logic}, over traces and provide
a sequent formulas. The approach differs in that \citeauthor{BB2013}
are confined to sequential programs and consider only traces over
states without events. Their goal was a connection between \ac{DL} and
\ac{LTL}.


\section{Future Work and Conclusion}
\label{sec:future-work-concl}

\subsection{Future Work}
\label{sec:future-work}

The original goal of applying \ac{LAGC} to \Rust in order to obtain a
semantics with which our dynamic logic calculus can be proven sound
has been achieved in \zcref{sec:sound-dl}. An avenue for future work
is to extend these soundness proofs for the remaining calculus rules,
as we present and discuss only the core rules here---though the
remaining rules can be proven straightforwardly following the
schematic approach demonstrated here. Similarly, we plan to consider
completeness of the calculus and demonstrate where incompleteness is
still encountered.

As future work, we would like to obtain a mechanization of our
complete semantics, the Tokio formalization, and the \ac{DL} proofs in
a proof assistant. \zcref[S]{sec:basics} has largely been mechanized
in Lean~\autocite{dU2021}, see \textcite{RustLAGCLean}. Mechanization
of \zcref{sec:lagc-semant-rust} is in progress. We intend to extend
this to our other contributions.

Another direction for future work is to provide an \ac{LAGC} semantics
for other parts of \Rust not adequately modeled in the present
article. For instance, a semantics of \ac{MIR} might be developed,
which more directly expresses \Rust's ownership and reference
semantics. Having a semantics both for the \ac{HIR} and \ac{MIR} level
would permit defining a notion of correct translation from \ac{HIR} to
\ac{MIR}, against which concrete implementations can be
checked. Though we believe it possible even on our chosen level of
abstraction, references and borrow checking are more easily modeled on
the \ac{MIR} level. We may also explore integrating the Polonius
borrow-checking formalization into \ac{LAGC}. A formalization of
\Rust's reference and ownership properties in terms of well-formedness
of traces is not part of our work due to constraints of time and
space.

A final avenue for research is to extend the model of Tokio's current
thread runtime to its multi-threaded counterpart and prove its
fairness. While modeling the runtime is a manageable task, the
fairness proof requires additional proof methods or a more complex
proof. Further, extending the formalization to include I/O events
presents an interesting challenge. While \textcite{LAGC23} consider
rudimentary user input, modeling I/O in \ac{LAGC} has not been
attempted at a larger scale.

\subsection{Conclusion}
\label{sec:conclusion}

Any formal semantics of a real-world programming language is
necessarily complex and technically challenging, which is exacerbated
by considering concurrency and scheduling choices. \Rust introduces
the additional challenge of multiple language layers (\ac{HIR} and
\ac{MIR}) and sub-languages (unsafe \Rust). Hence, any formalization
and semantics of \Rust contains abstraction and differing focal
points. We choose to present a relatively abstract, high-level
semantics both in terms of language layer and representation of
states. While properties such as memory layout or type safety are not
expressible on this level, we can reason effectively about \ac{HIR}
\Rust and asynchronous programs.

We demonstrate the expressive power of the \ac{LAGC} approach by
formalizing the Tokio runtime and proving two notions of fairness for
our model. Additionally, we show how soundness of a program logic
calculus can be derived from our semantics. Both of these endeavors
would be much more complex, if not unfeasible, in a \ac{SOS}-style
semantics or a semantics for \Rust's lower layers. Thereby, we show
that the seemingly high upfront complexity of the \ac{LAGC} approach
is offset by the gain in its modular approach.

Our semantics also demonstrate the \ac{LAGC} approach's applicability
to real-world programming languages, bridging the gap from the
academic languages to which it had been confined. The approach is
sufficiently flexible and powerful to accurately model a language like
\Rust which, with its more functional syntax, lays outside the purview
of \ac{LAGC}'s original scope. In fact, we were able to simplify the
approach, dropping path conditions entirely, and complete the model
with relative ease. It proved to be a natural match for modeling
\Rust's async features and runtimes.

Thus, we offer a formal semantics for \Rust with many different
applications and use cases, as demonstrated above.


\begin{acks}
  This work is supported by the \grantsponsor{DFG}{German Research
    Foundation}{https://www.dfg.de/en} under Grants
  No.~\grantnum{DFG}{BU~2924/3-1} and \grantnum{DFG}{HA~2617/9-1}.

  We would like to thank the following people for their careful and
  incisive feedback: Reiner Hähnle, Richard Bubel, Marco Scaletta,
  Lukas Grätz, and Niklas Heidler. We also thank Thanh Long Nguyen for
  his work on the Lean mechanization.
\end{acks}

\printbibliography

@book{ABBHSU2016,
  title = {Deductive Software Verification – the {{KeY}} Book: {{From}} Theory to Practice},
  editor = {Ahrendt, Wolfgang and Beckert, Bernhard and Bubel, Richard and Hähnle, Reiner and Schmitt, Peter H. and Ulbrich, Mattias},
  date = {2016},
  series = {{{LNCS}}},
  number = {10001},
  publisher = {Springer},
  doi = {10.1007/978-3-319-49812-6}
}

@inproceedings{ABFGMMPS2022,
  title = {The {{Prusti}} Project: {{Formal}} Verification for {{Rust}}},
  booktitle = {{{NFM}} 2022 {{Proceedings}}},
  author = {Astrauskas, Vytautas and Bílý, Aurel and Fiala, Jonás and Grannan, Zachary and Matheja, Christoph and Müller, Peter and Poli, Federico and Summers, Alexander J.},
  editor = {Deshmukh, Jyotirmoy V. and Havelund, Klaus and Perez, Ivan},
  date = {2022},
  series = {{{LNCS}}},
  number = {13260},
  pages = {88--108},
  publisher = {Springer},
  doi = {10.1007/978-3-031-06773-0_5}
}

@software{AMF,
  title = {A-Mir-Formality},
  author = {The Rust developers},
  date = {2026},
  url = {https://github.com/rust-lang/a-mir-formality}
}

@article{B2007,
  title = {A Semantics for Concurrent Separation Logic},
  author = {Brookes, Stephen},
  date = {2007},
  journaltitle = {Theor. Comput. Sci.},
  volume = {375},
  number = {1--3},
  pages = {227--270},
  doi = {10.1016/j.tcs.2006.12.034}
}

@inproceedings{BB2013,
  title = {Dynamic Logic with Trace Semantics},
  booktitle = {Automated {{Deduction}} – {{CADE-24}}},
  author = {Beckert, Bernhard and Bruns, Daniel},
  editor = {Bonacina, Maria Paola},
  date = {2013},
  series = {{{LNAI}}},
  number = {7898},
  pages = {315--329},
  publisher = {Springer},
  location = {Berlin, Heidelberg},
  doi = {10.1007/978-3-642-38574-2_22},
  langid = {english}
}

@inproceedings{BBDHLPUW2025,
  title = {The {{Java}} Verification Tool {{KeY}}: {{A}} Tutorial},
  booktitle = {Formal {{Methods}}},
  author = {Beckert, Bernhard and Bubel, Richard and Drodt, Daniel and Hähnle, Reiner and Lanzinger, Florian and Pfeifer, Wolfram and Ulbrich, Mattias and Weigl, Alexander},
  editor = {Platzer, Andre and Rozier, Kristin Yvonne and Pradella, Matteo and Rossi, Matteo},
  date = {2025},
  series = {{{LNCS}}},
  number = {14934},
  pages = {597--623},
  publisher = {Springer Nature Switzerland},
  location = {Cham},
  doi = {10.1007/978-3-031-71177-0_32},
  isbn = {978-3-031-71177-0}
}

@inproceedings{BBFHKSS2025,
  title = {Hax: {{Verifying}} Security-Critical {{Rust}} Software Using Multiple Provers},
  shorttitle = {Hax},
  booktitle = {Verified {{Software}}. {{Theories}}, {{Tools}} and {{Experiments}} 2024 {{Proceedings}}},
  author = {Bhargavan, Karthikeyan and Buyse, Maxime and Franceschino, Lucas and Hansen, Lasse Letager and Kiefer, Franziskus and Schneider-Bensch, Jonas and Spitters, Bas},
  editor = {Protzenko, Jonathan and Raad, Azalea},
  date = {2025},
  series = {{{LNCS}}},
  number = {15525},
  pages = {96--119},
  publisher = {Springer},
  doi = {10.1007/978-3-031-86695-1_7},
  langid = {english}
}

@inproceedings{BGHS2023,
  title = {Trace-Based Deductive Verification},
  booktitle = {{{LPAR}} 2023: {{Proceedings}} of 24th International Conference on Logic for Programming, Artificial Intelligence and Reasoning, Manizales, Colombia, 4--9th June 2023},
  author = {Bubel, Richard and Gurov, Dilian and Hähnle, Reiner and Scaletta, Marco},
  editor = {Piskac, Ruzica and Voronkov, Andrei},
  date = {2023},
  series = {{{EPiC}} Series in Computing},
  volume = {94},
  pages = {73--95},
  publisher = {EasyChair},
  doi = {10.29007/VDFD},
  bibsource = {dblp computer science bibliography, https://dblp.org}
}

@article{CGK2023,
  title = {A Grounded Conceptual Model for Ownership Types in {{Rust}}},
  author = {Crichton, Will and Gray, Gavin and Krishnamurthi, Shriram},
  editor = {Hicks, Michael},
  date = {2023-10},
  journaltitle = {Proc. ACM Program. Lang.},
  volume = {7},
  pages = {265:1--265:29},
  publisher = {ACM},
  location = {New York, NY, USA},
  doi = {10.1145/3622841},
  issue = {OOPSLA2},
  langid = {english}
}

@inproceedings{DBH2015,
  title = {{{KeY-ABS}}: {{A}} Deductive Verification Tool for the Concurrent Modelling Language {{ABS}}},
  booktitle = {Proc. 25th Intl. {{Conf}}. on Automated Deduction ({{CADE}}), Berlin, Germany},
  author = {Din, Crystal Chang and Bubel, Richard and Hähnle, Reiner},
  editor = {Felty, Amy and Middeldorp, Aart},
  date = {2015},
  series = {{{LNCS}}},
  number = {9195},
  pages = {517--526},
  publisher = {Springer},
  doi = {10.1007/978-3-319-21401-6_35}
}

@article{dDFHHJKRSSY2017,
  title = {A Survey of Active Object Languages},
  author = {family=Boer, given=Frank, prefix=de, useprefix=true and Din, Crystal Chang and Fernandez-Reyes, Kiko and Hähnle, Reiner and Henrio, Ludovic and Johnsen, Einar Broch and Khamespanah, Ehsan and Rochas, Justine and Serbanescu, Vlad and Sirjani, Marjan and Yang, Albert Mingkun},
  date = {2017-10},
  journaltitle = {ACM Computing Surveys},
  volume = {50},
  number = {5},
  pages = {76:1--76:39},
  doi = {10.1145/3122848}
}

@article{DDJO2012,
  title = {Observable Behavior of Distributed Systems: {{Component}} Reasoning for Concurrent Objects},
  author = {Din, Crystal Chang and Dovland, Johan and Johnsen, Einar Broch and Owe, Olaf},
  date = {2012},
  journaltitle = {Journal of Logic and Algebraic Programming},
  volume = {81},
  number = {3},
  pages = {227--256},
  doi = {10.1016/j.jlap.2012.01.003}
}

@inproceedings{DH2026,
  title = {{{RustyDL}}: {{A}} Program Logic for {{Rust}}},
  booktitle = {Proc. 27th {{Intl}}. {{Symp}}. on {{Formal Methods}}, {{Tokyo}}, {{Japan}}, {{Part I}}},
  author = {Drodt, Daniel and Hähnle, Reiner},
  editor = {Sampaio, Augusto and Stoelinga, Marielle and Brucker, Achim D. and Jakobs, Marie-Christine},
  date = {2026-05},
  series = {{{LNCS}}},
  number = {16556},
  pages = {357--377},
  publisher = {Springer},
  location = {Cham},
  doi = {10.1007/978-3-032-26204-2_19}
}

@report{DH2026TR,
  type = {Technical report},
  title = {{{RustyDL}}: {{A}} Program Logic for {{Rust}}},
  author = {Drodt, Daniel and Hähnle, Reiner},
  date = {2026},
  eprint = {2602.22075},
  eprinttype = {arXiv},
  eprintclass = {cs.PL},
  doi = {10.48550/arXiv.2602.22075}
}

@inproceedings{DHJPT2017,
  title = {Locally Abstract, Globally Concrete Semantics of Concurrent Programming Languages},
  booktitle = {{{TABLEAUX}} 2017},
  author = {Din, Crystal Chang and Hähnle, Reiner and Johnsen, Einar Broch and Pun, Violet Ka I and Tapia Tarifa, Silvia Lizeth},
  editor = {Nalon, Cláudia and Schmidt, Renate},
  date = {2017-09},
  series = {{{LNCS}}},
  volume = {10501},
  pages = {22--43},
  publisher = {Springer},
  doi = {10.1007/978-3-319-66902-1_2}
}

@inproceedings{DJM2022,
  title = {Creusot: {{A}} Foundry for the Deductive Verification of {{Rust}} Programs},
  booktitle = {{{ICFEM}} 2022, {{Proceedings}}},
  author = {Denis, Xavier and Jourdan, Jacques-Henri and Marché, Claude},
  editor = {Riesco, Adrián and Zhang, Min},
  date = {2022},
  series = {{{LNCS}}},
  number = {13478},
  pages = {90--105},
  publisher = {Springer},
  doi = {10.1007/978-3-031-17244-1_6}
}

@article{DO2015,
  title = {Compositional Reasoning about Active Objects with Shared Futures},
  author = {Din, Crystal Chang and Owe, Olaf},
  date = {2015},
  journaltitle = {Formal Asp. Comput.},
  volume = {27},
  number = {3},
  pages = {551--572},
  doi = {10.1007/s00165-014-0322-y},
  bibsource = {dblp computer science bibliography, http://dblp.org}
}

@inproceedings{dU2021,
  title = {The {{Lean}} 4 Theorem Prover and Programming Language},
  booktitle = {{{CADE}}},
  author = {family=Moura, given=Leonardo, prefix=de, useprefix=true and Ullrich, Sebastian},
  date = {2021},
  series = {{{LNCS}}},
  number = {12699},
  pages = {625--635},
  publisher = {Springer},
  doi = {10.1007/978-3-030-79876-5_37}
}

@report{FLS,
  type = {Reference},
  title = {Ferrocene Language Specification},
  shorttitle = {{{FLS}}},
  author = {{The Ferrocene developers} and {The Rust project contributors}},
  date = {2026},
  institution = {Rust Foundation},
  url = {https://rust-lang.github.io/fls},
  urldate = {2026-07-23}
}

@software{G2024,
  title = {Glommio},
  author = {{Glommio Developers}},
  date = {2024-03-25},
  url = {https://docs.rs/glommio/0.9.0/glommio/},
  urldate = {2026-08-22},
  version = {0.9.0}
}

@article{H1969,
  title = {An Axiomatic Basis for Computer Programming},
  author = {Hoare, C. A. R.},
  date = {1969-10-01},
  journaltitle = {Commun. ACM},
  volume = {12},
  number = {10},
  pages = {576--580},
  issn = {0001-0782},
  doi = {10.1145/363235.363259}
}

@book{H2004,
  title = {The {{SPIN}} Model Checker -- Primer and Reference Manual},
  author = {Holzmann, Gerard J.},
  date = {2004},
  publisher = {Addison-Wesley}
}

@article{HEJ2009,
  title = {Semantics-Aware Trace Analysis},
  author = {Hoffman, Kevin J. and Eugster, Patrick and Jagannathan, Suresh},
  date = {2009-06},
  journaltitle = {SIGPLAN Not.},
  volume = {44},
  number = {6},
  pages = {453--464},
  publisher = {Association for Computing Machinery},
  location = {New York, NY, USA},
  issn = {0362-1340},
  doi = {10.1145/1543135.1542527},
  issue_date = {June 2009},
  pagetotal = {12}
}

@article{HFP2024,
  title = {Sound Borrow-Checking for {{Rust}} via Symbolic Semantics},
  author = {Ho, Son and Fromherz, Aymeric and Protzenko, Jonathan},
  date = {2024},
  journaltitle = {Proc. ACM Program. Lang.},
  volume = {8},
  pages = {426--454},
  doi = {10.1145/3674640},
  bibsource = {dblp computer science bibliography, https://dblp.org},
  issue = {ICFP}
}

@online{HFP2024a,
  title = {Sound Borrow-Checking for {{Rust}} via Symbolic Semantics (Long Version)},
  author = {Ho, Son and Fromherz, Aymeric and Protzenko, Jonathan},
  date = {2024-07-01},
  eprint = {2404.02680},
  eprinttype = {arXiv},
  eprintclass = {cs},
  doi = {10.48550/arXiv.2404.02680},
  pubstate = {prepublished}
}

@article{HH2024,
  title = {Provably Fair Cooperative Scheduling},
  author = {Hähnle, Reiner and Henrio, Ludovic},
  date = {2024},
  journaltitle = {Art Sci. Eng. Program.},
  volume = {8},
  number = {2},
  doi = {10.22152/PROGRAMMING-JOURNAL.ORG/2024/8/6},
  bibsource = {dblp computer science bibliography, https://dblp.org}
}

@inproceedings{HJT2025,
  title = {An {{LAGC}} Semantics for {{Timed Rebeca}}},
  booktitle = {Rebeca for Actor Analysis in Action -- {{Essays}} Dedicated to {{Marjan Sirjani}} on the Occasion of Her 60th Birthday},
  author = {Hähnle, Reiner and Johnsen, Einar Broch and Tarifa, Silvia Lizeth Tapia},
  editor = {Lee, Edward A. and Mousavi, Mohammad Reza and Talcott, Carolyn L.},
  date = {2025},
  series = {{{LNCS}}},
  number = {15560},
  pages = {326--348},
  publisher = {Springer},
  doi = {10.1007/978-3-031-85134-6_14},
  bibsource = {dblp computer science bibliography, https://dblp.org}
}

@online{HKMB2023,
  title = {Our Vision for the {{Rust}} Specification},
  author = {Huss, Eric and Klock, Felix and Marcey, Joel and Bos, Mara},
  date = {2023-11-15},
  url = {https://blog.rust-lang.org/inside-rust/2023/11/15/spec-vision/},
  urldate = {2026-07-23},
  organization = {Inside Rust Blog}
}

@inproceedings{HKS2024,
  title = {Context-Aware Trace Contracts},
  booktitle = {Active Object Languages: {{Current}} Research Trends},
  author = {Hähnle, Reiner and Kamburjan, Eduard and Scaletta, Marco},
  editor = {De Boer, Frank and Damiani, Ferruccio and Hähnle, Reiner and Johnsen, Einar Broch and Kamburjan, Eduard},
  date = {2024},
  series = {{{LNCS}}},
  number = {14360},
  pages = {292--325},
  publisher = {Springer},
  location = {Cham},
  doi = {10.1007/978-3-031-51060-1_11}
}

@inproceedings{HMM1983,
  title = {A Hardware Semantics Based on Temporal Intervals},
  booktitle = {Automata, Languages and Programming},
  author = {Halpern, Joseph and Manna, Zohar and Moszkowski, Ben},
  editor = {Diaz, Josep},
  date = {1983},
  pages = {278--291},
  publisher = {Springer},
  location = {Berlin, Heidelberg},
  doi = {10.1007/BFb0036915},
  langid = {english}
}

@article{HP2022,
  title = {Aeneas: {{Rust}} Verification by Functional Translation},
  author = {Ho, Son and Protzenko, Jonathan},
  date = {2022},
  journaltitle = {Proc. ACM Program. Lang.},
  volume = {6},
  pages = {711--741},
  doi = {10.1145/3547647},
  issue = {ICFP}
}

@article{JDKD2019,
  title = {Stacked Borrows: An Aliasing Model for {{Rust}}},
  shorttitle = {Stacked Borrows},
  author = {Jung, Ralf and Dang, Hoang-Hai and Kang, Jeehoon and Dreyer, Derek},
  date = {2019-12-20},
  journaltitle = {Proc. ACM Program. Lang.},
  volume = {4},
  pages = {41:1--41:32},
  doi = {10.1145/3371109},
  issue = {POPL}
}

@inproceedings{JHSSS2011,
  title = {{{ABS}}: {{A}} Core Language for Abstract Behavioral Specification},
  booktitle = {{{FMCO}} 2010 {{Proceedings}}},
  author = {Johnsen, Einar Broch and Hähnle, Reiner and Schäfer, Jan and Schlatte, Rudolf and Steffen, Martin},
  editor = {Aichernig, Bernhard K. and family=Boer, given=Frank, prefix=de, useprefix=true and Bonsangue, Marcello M.},
  date = {2011},
  series = {{{LNCS}}},
  number = {6957},
  pages = {142--164},
  publisher = {Springer},
  doi = {10.1007/978-3-642-25271-6_8}
}

@report{JJKD2017,
  type = {Technical appendix},
  title = {{{RustBelt}}: {{Securing}} the Foundations of the {{Rust}} Programming Language},
  author = {Jung, Ralf and Jourdan, Jacques-Henri and Krebbers, Robbert and Dreyer, Derek},
  date = {2017},
  url = {https://plv.mpi-sws.org/rustbelt/popl18/},
  urldate = {2026-07-23},
  langid = {english}
}

@article{JJKD2018,
  title = {{{RustBelt}}: {{Securing}} the Foundations of the {{Rust}} Programming Language},
  author = {Jung, Ralf and Jourdan, Jacques-Henri and Krebbers, Robbert and Dreyer, Derek},
  date = {2018},
  journaltitle = {Proc. ACM Program. Lang.},
  volume = {2},
  pages = {66:1--66:34},
  doi = {10.1145/3158154},
  issue = {POPL}
}

@inproceedings{JSSSTBD2015,
  title = {Iris: {{Monoids}} and Invariants as an Orthogonal Basis for Concurrent Reasoning},
  booktitle = {Proceedings of the 42nd Annual {{ACM SIGPLAN-SIGACT}} Symposium on Principles of Programming Languages, {{POPL}} 2015, {{Mumbai}}, {{India}}, {{January}} 15--17, 2015},
  author = {Jung, Ralf and Swasey, David and Sieczkowski, Filip and Svendsen, Kasper and Turon, Aaron and Birkedal, Lars and Dreyer, Derek},
  editor = {Rajamani, Sriram K. and Walker, David},
  date = {2015},
  pages = {637--650},
  publisher = {ACM},
  doi = {10.1145/2676726.2676980},
  bibsource = {dblp computer science bibliography, https://dblp.org}
}

@inproceedings{K2019,
  title = {Behavioral Program Logic},
  booktitle = {{{TABLEAUX}} 2019},
  author = {Kamburjan, Eduard},
  editor = {Cerrito, Serenella and Popescu, Andrei},
  date = {2019},
  series = {{{LNCS}}},
  volume = {11714},
  pages = {391--408},
  publisher = {Springer},
  doi = {10.1007/978-3-030-29026-9_22},
  bibsource = {dblp computer science bibliography, https://dblp.org}
}

@book{KN2022,
  title = {The {{Rust}} Programming Language},
  author = {Klabnik, Steve and Nichols, Carol and Krycho, Chris},
  date = {2026-03},
  edition = {3},
  publisher = {No Starch Press},
  location = {San Francisco, CA},
  url = {https://doc.rust-lang.org/stable/book/},
  urldate = {2026-07-23}
}

@article{LAGC23,
  title = {Locally Abstract, Globally Concrete Semantics of Concurrent Programming Languages},
  author = {Din, Crystal Chang and Hähnle, Reiner and Henrio, Ludovic and Johnsen, Einar Broch and Pun, Violet Ka I and Tapia Tarifa, Silvia Lizeth},
  date = {2024-03},
  journaltitle = {Transactions on Programming Languages and Systems},
  volume = {46},
  number = {1},
  pages = {3:1--3:58},
  publisher = {ACM},
  doi = {10.1145/3648439}
}

@inproceedings{LGYWX2024,
  title = {An Empirical Study of {{Rust-for-Linux}}: {{The}} Success, Dissatisfaction, and Compromise},
  booktitle = {{{USENIX ATC}} 2024 {{Proceedings}}},
  author = {Li, Hongyu and Guo, Liwei and Yang, Yexuan and Wang, Shangguang and Xu, Mengwei},
  editor = {Bagchi, Saurabh and Zhang, Yiying},
  date = {2024},
  pages = {425--443},
  publisher = {USENIX Association},
  url = {https://www.usenix.org/conference/atc24/presentation/li-hongyu},
  bibsource = {dblp computer science bibliography, https://dblp.org}
}

@article{LHCBSZHPH2023,
  title = {Verus: {{Verifying Rust}} Programs Using Linear Ghost Types},
  author = {Lattuada, Andrea and Hance, Travis and Cho, Chanhee and Brun, Matthias and Subasinghe, Isitha and Zhou, Yi and Howell, Jon and Parno, Bryan and Hawblitzel, Chris},
  date = {2023},
  journaltitle = {Proc. ACM Program. Lang.},
  volume = {7},
  pages = {286--315},
  doi = {10.1145/3586037},
  issue = {OOPSLA1}
}

@inproceedings{M1998,
  title = {A Functional Representation of Data Structures with a Hole},
  booktitle = {Proc. of the 25th {{ACM SIGPLAN-SIGACT}} Symposium on Principles of Programming Languages},
  author = {Minamide, Yasuhiko},
  date = {1998},
  series = {Popl '98},
  pages = {75--84},
  publisher = {ACM},
  address = {New York, NY, USA},
  location = {San Diego, California, USA},
  doi = {10.1145/268946.268953},
  pagetotal = {10}
}

@article{M2004,
  title = {Modular Structural Operational Semantics},
  author = {Mosses, Peter D},
  date = {2004-07-01},
  journaltitle = {The Journal of Logic and Algebraic Programming},
  series = {Structural {{Operational Semantics}}},
  volume = {60--61},
  pages = {195--228},
  issn = {1567-8326},
  doi = {10.1016/j.jlap.2004.03.008}
}

@online{M2018,
  title = {An Alias-Based Formulation of the Borrow Checker},
  author = {Matsakis, Niko},
  date = {2018-04-27},
  url = {https://smallcultfollowing.com/babysteps/blog/2018/04/27/an-alias-based-formulation-of-the-borrow-checker/},
  urldate = {2026-07-23},
  organization = {baby steps}
}

@online{M2023,
  title = {Polonius Revisted, Part 1},
  author = {Matsakis, Niko},
  date = {2023-09-22},
  url = {https://smallcultfollowing.com/babysteps/blog/2023/09/22/polonius-part-1/},
  urldate = {2026-07-23},
  organization = {baby steps}
}

@online{M2023a,
  title = {Polonius Revisited, Part 2},
  author = {Matsakis, Niko},
  date = {2023-09-29},
  url = {https://smallcultfollowing.com/babysteps/blog/2023/09/29/polonius-part-2/},
  urldate = {2026-07-23},
  organization = {baby steps}
}

@inproceedings{MDJD2022,
  title = {{{RustHornBelt}}: {{A}} Semantic Foundation for Functional Verification of {{Rust}} Programs with Unsafe Code},
  booktitle = {{{PLDI}} '22: 43rd {{ACM SIGPLAN International Conference}} on {{Programming Language Design}} and {{Implementation}}, {{San Diego}}, {{CA}}, {{USA}}, {{June}} 13 - 17, 2022},
  author = {Matsushita, Yusuke and Denis, Xavier and Jourdan, Jacques-Henri and Dreyer, Derek},
  editor = {Jhala, Ranjit and Dillig, Isil},
  date = {2022},
  pages = {841--856},
  publisher = {ACM},
  doi = {10.1145/3519939.3523704}
}

@software{MiniRust,
  title = {{{MiniRust}}},
  author = {Jung, Ralf},
  date = {2026-03-25T17:30:16Z},
  origdate = {2022-05-21T15:30:40Z},
  url = {https://github.com/minirust/minirust},
  organization = {minirust}
}

@inproceedings{MSS2016,
  title = {Viper: {{A}} Verification Infrastructure for Permission-Based Reasoning},
  booktitle = {Verification, Model Checking, and Abstract Interpretation},
  author = {Müller, Peter and Schwerhoff, Malte and Summers, Alexander J.},
  editor = {Jobstmann, Barbara and Leino, K. Rustan M.},
  date = {2016},
  pages = {41--62},
  publisher = {Springer Berlin Heidelberg},
  location = {Berlin, Heidelberg},
  doi = {10.1007/978-3-662-49122-5_2},
  isbn = {978-3-662-49122-5}
}

@inproceedings{P1976,
  title = {Semantical Considerations on {{Floyd-Hoare}} Logic},
  booktitle = {17th Annual Symposium on Foundations of Computer Science, {{Houston}}, {{TX}}, {{USA}}},
  author = {Pratt, Vaughan R},
  date = {1976},
  pages = {109--121},
  publisher = {IEEE},
  location = {Los Alamitos, CA},
  doi = {10.1109/SFCS.1976.27}
}

@article{P2021,
  title = {A Lightweight Formalism for Reference Lifetimes and Borrowing in {{Rust}}},
  author = {Pearce, David J.},
  date = {2021-04-17},
  journaltitle = {ACM Trans. Program. Lang. Syst.},
  volume = {43},
  number = {1},
  pages = {3:1--3:73},
  issn = {0164-0925},
  doi = {10.1145/3443420}
}

@software{Polonius,
  title = {Polonius},
  author = {The Rust developers},
  date = {2025},
  url = {https://github.com/rust-lang/polonius}
}

@thesis{R2024a,
  type = {bathesis},
  title = {Formal Semantics for Smart Contracts},
  author = {Reinhold, Tobias},
  date = {2024-12-11},
  institution = {Technical University of Darmstadt, Department of Computer Science},
  location = {Darmstadt, Germany},
  pagetotal = {68}
}

@inproceedings{Reif65,
  title = {The Many Uses of Dynamic Logic},
  booktitle = {Go Where the Bugs Are -- Essays Dedicated to {{Wolfgang Reif}} on the Occasion of His 65th Birthday},
  author = {Ahrendt, Wolfgang and Beckert, Bernhard and Bubel, Richard and Hähnle, Reiner and Ulbrich, Mattias},
  editor = {Ernst, Gidon and Güdemann, Matthias and Knapp, Alexander and Nafz, Florian and Ortmeier, Frank and Ponsar, Hella and Schellhorn, Gerhard and Schiendorfer, Alexander},
  date = {2025-06},
  series = {{{LNCS}}},
  number = {15765},
  pages = {56--82},
  publisher = {Springer},
  location = {Cham},
  doi = {10.1007/978-3-031-92196-4_4},
  bibsource = {dblp computer science bibliography, https://dblp.org}
}

@online{RM2023,
  title = {Polonius Update},
  author = {Rakic, Rémy and Matsakis, Niko},
  date = {2023-10-06},
  url = {https://blog.rust-lang.org/inside-rust/2023/10/06/polonius-update/},
  urldate = {2026-07-23},
  organization = {Inside Rust Blog}
}

@book{RustcGuide,
  title = {Rust Compiler Development Guide},
  author = {{Rust Community}},
  date = {2026},
  url = {https://rustc-dev-guide.rust-lang.org/},
  urldate = {2026-07-23}
}

@software{RustLAGCLean,
  title = {Rust-{{LAGC}}},
  author = {Drodt, Daniel},
  date = {2026-09-21T11:55:03Z},
  origdate = {2026-09-21T11:21:22Z},
  url = {https://github.com/Drodt/Rust-LAGC}
}

@report{RustRef,
  type = {Reference},
  title = {The {{Rust}} Reference},
  author = {{Rust Community}},
  date = {2026},
  institution = {Rust Foundation},
  url = {https://doc.rust-lang.org/stable/reference/introduction.html},
  urldate = {2026-07-23}
}

@thesis{S2026,
  type = {phdthesis},
  title = {Modular Specification and Verification of Concurrent Programs with Context-Aware Trace Contracts},
  author = {Scaletta, Marco},
  date = {2026},
  institution = {Technical University of Darmstadt, Department of Computer Science},
  doi = {10.26083/tuda-8276}
}

@software{S2026a,
  title = {{{KeY4CATs}}: {{A}} Deductive Verification Tool for Context-Aware Trace Contracts},
  author = {Scaletta, Marco},
  date = {2026-06},
  doi = {10.5281/zenodo.20762522},
  organization = {Zenodo}
}

@software{smol,
  title = {{{smol}}},
  author = {{smol Developers}},
  date = {2024-09-07},
  url = {https://docs.rs/smol/2.0.2/smol/},
  urldate = {2026-08-22},
  version = {2.0.2}
}

@inproceedings{SW2017,
  title = {A New Invariant Rule for the Analysis of Loops with Non-Standard Control Flows},
  booktitle = {{{iFM}} 2017 Proceedings},
  author = {Steinhöfel, Dominic and Wasser, Nathan},
  editor = {Polikarpova, Nadia and Schneider, Steve A.},
  date = {2017},
  series = {{{LNCS}}},
  number = {10510},
  pages = {279--294},
  publisher = {Springer},
  doi = {10.1007/978-3-319-66845-1_18},
  bibsource = {dblp computer science bibliography, https://dblp.org}
}

@inproceedings{T2022a,
  title = {Locally Abstract Globally Concrete Semantics of Time and Resource Aware Active Objects},
  booktitle = {The Logic of Software. {{A}} Tasting Menu of Formal Methods -- {{Essays}} Dedicated to {{Reiner Hähnle}} on the Occasion of His 60th Birthday},
  author = {Tarifa, Silvia Lizeth Tapia},
  editor = {Ahrendt, Wolfgang and Beckert, Bernhard and Bubel, Richard and Johnsen, Einar Broch},
  date = {2022},
  series = {{{LNCS}}},
  number = {13360},
  pages = {481--499},
  publisher = {Springer},
  doi = {10.1007/978-3-031-08166-8_23},
  bibsource = {dblp computer science bibliography, https://dblp.org}
}

@software{tokio,
  title = {Tokio},
  author = {{Tokio Developers}},
  date = {2026-07-20},
  url = {https://crates.io/crates/tokio},
  version = {1.53.1}
}

@online{TokioDoc,
  title = {{{tokio::runtime -- Documentation}}},
  author = {{Tokio Developers}},
  date = {2026-07-20},
  url = {https://docs.rs/tokio/1.53.1/tokio/runtime/index.html#detailed-runtime-behavior},
  urldate = {2026-07-23}
}

@article{VHDJ2025,
  title = {Tree Borrows},
  author = {Villani, Neven and Hostert, Johannes and Dreyer, Derek and Jung, Ralf},
  date = {2025-06-13},
  journaltitle = {Proc. ACM Program. Lang.},
  volume = {9},
  pages = {188:1019--188:1042},
  doi = {10.1145/3735592},
  issue = {PLDI}
}

@inproceedings{W2018,
  title = {Programming Language Foundations in {{Agda}}},
  booktitle = {Formal {{Methods}}: {{Foundations}} and {{Applications}}},
  author = {Wadler, Philip},
  editor = {Massoni, Tiago and Mousavi, Mohammad Reza},
  date = {2018},
  series = {{{LNCS}}},
  number = {11254},
  pages = {56--73},
  publisher = {Springer},
  location = {Cham},
  doi = {10.1007/978-3-030-03044-5_5},
  langid = {english}
}

\appendix

\section{Additional Definitions of the $\nextExprSym$ Function}
\label{sec:add-def-next}

First, we show the expressions which have no sub-expression under any
circumstances. Those are values, program variables, continue, loop,
and panic expressions. For them, the entire expression is considered
the sub-expression and the hole expression is the hole~$\exprHole$.
\begin{align*}
  \nextExpr{v}{M} = (v, \exprHole) \\
  \nextExpr{\pv{x}}{M} = (\pv{x}, \exprHole) \\
  \nextExpr{\rust{\kwContinue~'$l$}}{M} = (\rust{\kwContinue~'$l$}, \exprHole) \\
  \nextExpr{\rust{'$l$ \kwLoop~$e$}}{M} = (\rust{'$l$ \kwLoop~$e$}, \exprHole) \\
  \nextExpr{\rust{\panic}}{M} = (\rust{\panic}, \exprHole)
\end{align*}
Dereferencing expressions have, potentially, one sub-expression, which
is always evaluated to a value. Similar definitions holds for match,
break, return, and array repeat expressions.
\begin{align*}
  \nextExpr{\rust{*$e$}}{M} =
  \begin{cases}
    (\rust{*$e$}, \exprHole) &\text{if}~e\in\Val\\
    (e',\rust{*$\mathit{he}$}) &\text{if}~e\not\in\Val\wedge(e',\mathit{he})=\nextExpr{e}{\Val}
  \end{cases}\\
  \nextExpr{\rust{\kwMatch~$e$ \{$\many{\mathit{arm}}$\}}}{M} =
  \begin{cases}
    (\rust{\kwMatch~$e$ \{$\many{\mathit{arm}}$\}},\exprHole)&\text{if}~e\in\Val\\
    (e',\rust{\kwMatch~$\mathit{he}$ \{$\many{\mathit{arm}}$\}}) &\text{if}~e\not\in\Val\wedge(e',\mathit{he})=\nextExpr{e}{\Val}
  \end{cases}\\
  \nextExpr{\rust{\kwBreak~'$l$~$e$}}{M} =
  \begin{cases}
    (\rust{\kwBreak~'$l$~$e$},\exprHole)&\text{if}~e\in\Val\\
    (e',\rust{\kwBreak~'$l$~$\mathit{he}$}) &\text{if}~e\not\in\Val\wedge(e',\mathit{he})=\nextExpr{e}{\Val}
  \end{cases}\\
  \nextExpr{\rust{\kwReturn~$e$}}{M} =
  \begin{cases}
    (\rust{\kwReturn~$e$},\exprHole)&\text{if}~e\in\Val\\
    (e',\rust{\kwReturn~$\mathit{he}$}) &\text{if}~e\not\in\Val\wedge(e',\mathit{he})=\nextExpr{e}{\Val}
  \end{cases}\\
  \nextExpr{\rust{[$e$; $n$]}}{M} =
  \begin{cases}
    (\rust{[$e$; $n$]}, \exprHole)&\text{if}~e\in\Val\\
    (e',\rust{[$\mathit{he}$; $n$]}) &\text{if}~e\not\in\Val\wedge(e',\mathit{he})=\nextExpr{e}{\Val}
  \end{cases}
\end{align*}

Block expressions must be considered depending on their first
statement. If it is a let statement, there is no next
sub-expression. If it is an expression, but not a break expression
with a value or continue expression, it behaves as above. We exclude
break and continues here to fit the expected rule shape expected in
\zcref{sec:loops}. If there are no statements, it is exactly as in the
cases above, as there is only one possible sub-expression.
\begin{align*}
  \nextExpr{\rust{\{\kwLet~$\mathit{pat}$ = $e_1$; $\many{s}$ $e_2$\}}}=(\rust{\{\kwLet~$\mathit{pat}$ = $e_1$; $\many{s}$ $e_2$\}},\exprHole)\\
  \nextExpr{\rust{\{$e_1$; $\many{s}$ $e_2$\}}} =
  \begin{cases}
    (\rust{\{$e_1$; $\many{s}$ $e_2$\}},\exprHole) &\text{if}~e_1\in\Val \vee e_1=\rust{\kwBreak~'$l$~$v$} \vee e_1=\rust{\kwContinue~'$l$}\\
    (e_1',\rust{\{$\mathit{he}$; $\many{s}$ $e_2$\}}) &\text{otherwise, where}~(e_1',\mathit{he})=\nextExpr{e_1}{\Val}
  \end{cases}
  \\
  \nextExpr{\rust{\{$e$\}}}{M} =
  \begin{cases}
    (\rust{\{$e$\}},\exprHole)&\text{if}~e\in\Val\\
    (e',\rust{\{$\mathit{he}$\}}) &\text{if}~e\not\in\Val\wedge(e',\mathit{he})=\nextExpr{e}{\Val}
  \end{cases}
\end{align*}

Operator expressions have two potential sub-expressions, which are all
evaluated to values. We only show the case for addition; other
operators are analogous.
\[
  \nextExpr{\rust{$e_1$ + $e_2$}}{M} =
  \begin{cases}
    (\rust{$e_1$ + $e_2$},\exprHole)&\text{if}~e_1,e_2\in\Val\\
    (e_1',\rust{$\mathit{he}$ + $e_2$}) &\text{if}~e_1\not\in\Val\wedge(e_1',\mathit{he})=\nextExpr{e_1}{\Val}\\
    (e_2',\rust{$e_1$ + $\mathit{he}$}) &\text{if}~e_1\in\Val\wedge e_2\not\in\Val\wedge (e_2',\mathit{he})=\nextExpr{e_2}{\Val}
  \end{cases}
\]
Assignment operators behave differently. When $e_1,e_2$ are of integer
types, the sub-expressions are evaluated right-to-left; otherwise
left-to-right. We always evaluate $e_2$ to a value, but $e_1$ to a
place. Here, we show the integer version.
\[
  \nextExpr{\rust{$e_1$ += $e_2$}}{M} =
  \begin{cases}
    (\rust{$e_1$ += $e_2$},\exprHole)&\text{if}~e_1\in\Place\wedge e_2\in\Val\\
    (e_2',\rust{$e_1$ += $\mathit{he}$}) &\text{if}~e_2\not\in\Val\wedge(e_2',\mathit{he})=\nextExpr{e_2}{\Val}\\
    (e_1',\rust{$\mathit{he}$ += $e_2$}) &\text{if}~e_2\in\Val\wedge e_1\not\in\Place\wedge (e_1',\mathit{he})=\nextExpr{e_1}{\Place}
  \end{cases}
\]
Arrays and tuples behave identically: Evaluation is left-to-right and
we might be in an assignee or value context. Struct expressions
\rust{x\{$\mathit{ident}_1:e_1,\ldots,\mathit{ident}_n:e_n$\}} are
analogous to arrays and tuples.
\[
  \nextExpr{\rust{[$e_1,\ldots,e_n$]}}{M} =
  \begin{cases}
    (\rust{[$e_1,\ldots,e_n$]}, \exprHole)&\text{if}~e_1,\ldots,e_n\in M\\
    (e'_k,\rust{[$e_1,\ldots,e_{k-1},\mathit{he},e_{k+1},\ldots,e_n$]}) &\text{if}~e_1,\ldots,e_{k-1}\in M \wedge e_k\not\in M\\
                                                                        &\quad\wedge (e'_k,\mathit{he})=\nextExpr{e_k}{M}
  \end{cases}
\]
\[
  \nextExpr{\rust{($e_1,\ldots,e_n$)}}{M} =
  \begin{cases}
    (\rust{($e_1,\ldots,e_n$)}, \exprHole)&\text{if}~e_1,\ldots,e_n\in M\\
    (e'_k,\rust{($e_1,\ldots,e_{k-1},\mathit{he},e_{k+1},\ldots,e_n$)}) &\text{if}~e_1,\ldots,e_{k-1}\in M \wedge e_k\not\in M\\
                                                                        &\quad\wedge (e'_k,\mathit{he})=\nextExpr{e_k}{M}
  \end{cases}
\]
Function calls behave like tuples and array but with an additional
expression, the callee, which is evaluated first. Method calls
\rust{$e$.$\mathit{ident}$($e_1,\ldots,e_n$)} analogous to function
calls and are omitted.
\[
  \nextExpr{\rust{$e$($e_1,\ldots,e_n$)}}{M} =
  \begin{cases}
    (\rust{$e$($e_1,\ldots,e_n$)}, \exprHole) &\text{if}~e,e_1,\ldots,e_n\in\Val\\
    (e',\rust{$\mathit{he}$($e_1,\ldots,e_n$)}) &\text{if}~e\not\in\Val\wedge(e',\mathit{he})=\nextExpr{e}{\Val}\\
    (e'_k,\rust{$e$($e_1,\ldots,e_{k-1},\mathit{he},e_{k+1},\ldots,e_n$)}) &\text{if}~e,e_1,\ldots,e_{k-1}\in\Val \wedge e_k\not\in\Val\\
                                                   &\quad\wedge (e'_k,\mathit{he})=\nextExpr{e_k}{\Val}
  \end{cases}
\]

Field access expressions depend on the context, as they may evaluate
their single sub-expression to an assignee value or a value.
\[
  \nextExpr{\rust{$e$.$\mathit{ident}$}}{M} =
  \begin{cases}
    (\rust{$e$.$\mathit{ident}$}, \exprHole)&\text{if}~e\in M\\
    (e',\rust{$\mathit{he}$.$\mathit{ident}$}) &\text{if}~e\not\in M\wedge (e',\mathit{he})=\nextExpr{e}{M}
  \end{cases}
\]

Reference expressions always evaluate their sub-expression to a place.
\begin{align*}
  \nextExpr{\rust{\&$e$}}{M} =
  \begin{cases}
    (\rust{\&$e$}, \exprHole) &\text{if}~e\in\Place\\
    (e',\rust{\&$\mathit{he}$}) &\text{if}~e\not\in\Place\wedge(e',\mathit{he})=\nextExpr{e}{\Place}
  \end{cases}\\
  \nextExpr{\rust{\&\kwMut~$e$}}{M} =
  \begin{cases}
    (\rust{\&\kwMut~$e$}, \exprHole)&\text{if}~e\in\Place\\
    (e',\rust{\&\kwMut~$\mathit{he}$}) &\text{if}~e\not\in\Place\wedge(e',\mathit{he})=\nextExpr{e}{\Place}
  \end{cases}
\end{align*}


\section{Axiomatization of Queue}
\label{sec:axiom-queue}

We repeat and adapt the axiomatization for queues found
in~\cite{HH2024}.

\begin{definition}[Queue]\label{def:queue}
  We define a \emph{queue of type $T$} $q$ as queues over elements of
  type~$T$. The operations are:
  \begin{equation*}
    \begin{array}[t]{rcl}
      \queue & \text{---} & \text{creates an empty queue}\\
      \qEmpty{q} & \text{---} & \text{is true iff $q$ is the empty queue}\\
      \qPush{q}{k} & \text{---} & \text{returns a queue $q'$ that is $q$ extended with $k$ as its last element}\\
      \qPop{q} & \text{---} & \text{returns a pair $(q',k)$, where $k$ is the first element in $q$ and $q'$}\\
             &&\text{the remaining queue after $k$ is removed; the operation is}\\
             &&\text{undefined iff $q$ is empty}
    \end{array}
  \end{equation*}
  The operations are defined by the following axioms:
  \begin{align}
    \qEmpty{\queue}\label{eq:queue-ax-empty}\\
    \qPop{\queue}\undefSym\label{eq:queue-ax-pop-create}\\
    \forall q, k.\ \neg\qEmpty{\qPush{q}{k}}\label{eq:queue-ax-empty-push}\\
    \forall q, k, q', k'.\ \qPop{\qPush{q}{k}}=
    \begin{cases}
      (\queue, k) &\text{if}~\qEmpty{q}\\
      (\qPush{q'}{k}) &\text{otherwise, where}~(q',k')=\qPop{q}
    \end{cases}\label{eq:queue-ax-pop-push}
  \end{align}
  If $q$ is a queue of stacks of continuation markers, we also define
  the set of continuation markers in $q$, $\getConts{q}$, as:
  \[
    \getConts{q}=
    \begin{cases}
      \emptyset &\text{if}~\qEmpty{q}\\
      \many{k} \uplus \getConts{q'} &\text{where}~\qPop{q}=(q',\many{k})
    \end{cases}
  \]
\end{definition}


\section{Additional Current Thread Composition Rules}
\label{app:addt-curr-thre}

\gcondrule{ctIdleLocalEmptyMarker}{
  \nexists f,h.\ \lastEv(\sh)=\spawnEv{f}{h} \quad t \not\equiv 0 \mod I_G \\
  \qEmpty{q} \quad \qPop{q_G}=(q'_G,(\pv{foo},f)) \quad \sigma=\last(\sh) \quad \rust{\kwFn foo($\many{x}$) $e$}\in\mathit{Fn}\\
  \concatTr{\tau}{\cont{e'}} \in \valP{f}{\rust{\kwFn foo($\many{x}$) $e$}} \quad \rho\text{ concretizes }\tau \quad \wf{\chopTrSem{\sh}{\rho(\tau)}}
}{
  \ctCfg{\sh}{\Idle}{t}{q_G}{q} \toCT \ctCfg{\sh}{\left[\contP{f}{e'}\right]}{t + 1}{q'_G}{q}
}

\gcondrule{ctIdleLocalMarker}{
  \nexists f,h.\ \lastEv(\sh)=\spawnEv{f}{h} \quad t \not\equiv 0 \mod I_G \\
  \qPop{q}=(q',(\pv{foo},f))\quad
  \sigma=\last(\sh) \quad \rust{\kwFn foo($\many{x}$) $e$}\in\mathit{Fn} \\
  \concatTr{\tau}{\cont{e'}} \in \valP{f}{\rust{\kwFn foo($\many{x}$) $e$}} \quad
  \rho\text{ concretizes }\tau \quad \wf{\chopTrSem{\sh}{\rho(\tau)}}
}{
  \ctCfg{\sh}{\Idle}{t}{q_G}{q} \toCT \ctCfg{\sh}{\left[\contP{f}{e'}\right]}{t + 1}{q_G}{q'}
}


\section{Quiescent Fairness of the Current Thread Scheduler}
\label{app:quiesc-fairn-curr}

We first define a translation function \enquote{$\schedTransSym$}
mapping current thread configurations to configurations in the
semantics of \zcref{sec:asynchr-rust-semant}. Traces and the
scheduler's state can be mapped independently to traces and sets of
tasks, respectively. For this, we lift $\schedTransSym$ to both:
\[
  \schedTrans{\ctCfg{\sh}{\mathit{active}}{t}{q_G}{q}} = \schedTrans{\sh},\schedTrans{\ctState{\mathit{active}}{t}{q_G}{q}}
\]

Traces are \enquote{translated} via a projection, by removing the
schedule events and the preceding state (as \zcref{eq:rule-ctSpawn}
repeats this state).
\begin{equation*}
  \begin{array}{r@{\;}c@{\;}l}
    \schedTrans{\varepsilon} & = & \epsilon \\
    \schedTrans{\consTr{\schedEv{\pv{foo}}{f}}{\consTr{\sigma}{\sh}}} & = & \schedTrans{\sh}\\
    \schedTrans{\consTr{\stateorevent}{\sh}} & = & \consTr{\stateorevent}{\schedTrans{\sh}} \quad \text{when}~\stateorevent\ne\schedEv{\pv{foo}}{f}
  \end{array}
\end{equation*}

Next, we define how queues of continuation stacks and task markers are
mapped to sets of these stacks. The inductive definition is obvious;
note that the third case drops task markers.
\[
  \schedAw{q}=
  \begin{cases}
    \emptyset &\text{if}~\qEmpty{q}\\
    \left\{\left[\many{k}\right]\right\} \uplus \schedAw{q'} &\text{if}~\qPop{q}=\left(q',\left[\many{k}\right]\right)\\
    \schedAw{q'} &\text{otherwise, with}~\qPop{q}=(q',k)
  \end{cases}
\]


Finally, we translate the scheduler state by combining the global and
local queue with the active task:
\[
  \schedTrans{\ctState{\mathit{act}}{t}{q_G}{q}} =
  \asState{\mathit{act}}{\schedAw{q_G} \uplus \schedAw{q}}
\]

Next, we define a measure of the scheduling distance. We use this in
our fairness proof, showing that at every step, an enabled expression
is either executed or its distance decreases. Similar to a termination
witness in formal verification~\autocite{ABBHSU2016}, this distance
does not need to be an accurate representation of how many
\enquote{steps} it takes until an expression is executed. It suffices
to be a finitely descending chain of values.
\begin{definition}[Scheduling Distance: Current Thread]\label{def:ct-dist}
  Let the configuration $C=\ctCfg{\sh}{\mathit{act}}{t}{q_G}{q}$ and
  $c\in\mathbb{Z}$. We define the \emph{scheduling distance}~$d$ of a
  continuation marker $\contP{f}{e}$ in $C$ as follows:
  \[
    \ctDist{\contP{f}{e}}{C,c} = (\pos(\contP{f}{e},C),\mathrm{idle}(C),\mathrm{unsched}(\sh),c)
  \]
  where $\pos$ defines the position of the continuation in the
  scheduler as:
  \begin{align*}
    &\quad\pos(\contP{f}{e},\ctCfg{\sh}{\mathit{act}}{t}{q_G}{q})=\\
    &\begin{cases}
      1 &\text{if}~\mathit{act}=[k_1,\ldots,k_n] \wedge \contP{f}{e}=k_1 \\
        &\quad\wedge\, \neg\suspended{e}{\sh}\\
      \left|\getConts{q}\right| + \left|\getConts{q_G}\right| + 2 &\text{if}~\mathit{act}=[k_1,\ldots,k_n] \wedge \contP{f}{e}=k_1 \\
        &\quad\wedge \,\suspended{e}{\sh}\\
      i + \left|\getConts{q}\right| + \left|\getConts{q_G}\right| + 2 &\text{if}~\mathit{act}=[k_1,\ldots,k_n] \wedge \contP{f}{e}=k_i \wedge 1 < i \le n \\
      1 + \qPos(\contP{f}{e}, q) + \left|\getConts{q_G}\right| &\text{if}~\mathit{act}\ne\Idle \wedge \contP{f}{e}\in\getConts{q}\\
      1 + \left|\getConts{q}\right| + \qPos(\contP{f}{e}, q_G) &\text{if}~\mathit{act}\ne\Idle \wedge \contP{f}{e}\in\getConts{q_G}\\
      \qPos(\contP{f}{e}, q) + \left|\getConts{q_G}\right| &\text{if}~\mathit{act}=\Idle \wedge \contP{f}{e}\in\getConts{q}\\
      \left|\getConts{q}\right| + \qPos(\contP{f}{e}, q_G) &\text{if}~\mathit{act}=\Idle \wedge \contP{f}{e}\in\getConts{q_G}\\
      \mathit{undef} &\text{if}~\contP{f}{e}\not\in\getConts{\schedTrans{\ctState{[k_1,\ldots,k_n]}{t}{q_G}{q}}}
    \end{cases}
  \end{align*}
  Note that this position definition is \emph{pessimistic}: If a
  continuation marker is in one queue, it assumes that the entire
  other queue may be scheduled before it. While unrealistic, our
  scheduling distance need only be an over approximation of the
  \enquote{real distance.} Further, the distance is only used in our
  proof for \emph{enabled} continuation markers, i.e., continuation
  markers executable in the next step.

  A continuation marker's position in a queue is defined as follows:
  \[
    \qPos(\contP{f}{e},q)=
    \begin{cases}
      i &\text{if}~\qPop{q}=(q',[k_1,\ldots,k_n]) \wedge \contP{f}{e}=k_i\\
      n + \qPos(\contP{f}{e},q') &\text{if}~\qPop{q}=(q',[k_1,\ldots,k_n]) \wedge \contP{f}{e}\not\in k_1,\ldots,k_n\\
      1 + \qPos(\contP{f}{e},q') &\text{if}~\qPop{q}=(q',(\pv{foo},f))\\
      \mathit{undef} &\text{if}~\contP{f}{e}\not\in\getConts{q}
    \end{cases}
  \]
  The function \enquote{idle} is $1$ if the scheduler is idle and $0$
  otherwise.
  \[
    \mathrm{idle}(\ctCfg{\sh}{\mathit{act}}{t}{q_G}{q})=
    \begin{cases}
      1 &\text{if}~\mathit{act}=\Idle\\
      0 &\text{otherwise}
    \end{cases}
  \]
  With \enquote{unsched,} we check whether there is an unscheduled
  spawn events in the trace. It is always either $0$ or $1$, as there
  may be at most one unscheduled spawn event in the trace.
  \[
    \mathrm{unsched}(\sh)=
    \begin{cases}
      1 &\text{if}~\exists f,h.\ \lastEv(\sh)=\spawnEv{f}{h}\\
      0 &\text{otherwise}
    \end{cases}
  \]
  We extend the definition to task markers as follows:
  \[
    \qPos((\pv{foo},f),q)=
    \begin{cases}
      1 &\text{if}~\qPop{q}=(q',(\pv{foo},f))\\
      \left|\many{k}\right| + \qPos((\pv{foo},f),q') &\text{if}~\qPop{q}=(q',[\many{k}])\\
      1 + \qPos((\pv{foo},f),q') &\text{if}~\qPop{q}=(q',(\pv{foo}',f')) \wedge (\pv{foo}\ne\pv{foo}' \vee f\ne f')\\
      \mathit{undef} &\text{if}~(\pv{foo},f)\not\in q
    \end{cases}
  \]
  \begin{align*}
    &\pos((\pv{foo},f), \ctCfg{\sh}{\mathit{act}}{t}{q_G}{q})\\
    =&\begin{cases}
      1 + \qPos((\pv{foo},f), q) + \left|\getConts{q_G}\right| &\text{if}~\mathit{act}\ne\Idle \wedge (\pv{foo},f)\in q\\
      1 + \left|\getConts{q}\right| + \qPos((\pv{foo},f), q_G) &\text{if}~\mathit{act}\ne\Idle \wedge (\pv{foo},f)\in q_G\\
      \qPos((\pv{foo},f), q) + \left|\getConts{q_G}\right| &\text{if}~\mathit{act}=\Idle \wedge (\pv{foo},f)\in q\\
      \left|\getConts{q}\right| + \qPos((\pv{foo},f), q_G) &\text{if}~\mathit{act}=\Idle \wedge (\pv{foo},f)\in q_G
    \end{cases}
  \end{align*}
  We define the order relation ${>}$ over scheduling distances as the
  canonical ordering of tuples.
\end{definition}

As can be seen by the distance definition,
$\ctDist{\contP{f}{e}}{C,c} \ge (1,0,0,c)$ if
$\contP{f}{e}\in\getConts{\schedTrans{C}}$; otherwise, the distance is
undefined. If $\ctDist{\contP{f}{e}}{C,c}=(1,0,0,c)$ and
$e\not\in\Val$, then \zcref{eq:rule-ctActivePush} or
\zcref{eq:rule-ctActiveEval} are the next applicable rules. Likewise,
if $\ctDist{\contP{f}{e}}{C,c}=(1,1,0,c)$, then the next rule is one
of \zcref{eq:rule-ctIdleInterval}, \zcref{eq:rule-ctIdleLocalEmpty},
or \zcref{eq:rule-ctIdleLocal}.

The distance of a continuation marker being undefined is not a problem
in the following discussion. As we will consider the production of new
expressions, we need only consider continuation or task markers that
may lead to new expressions via composition rules. These are precisely
those markers that have a defined scheduling distance. We never
consider the ordering where either component is undefined.

The component~$c$ in the distance tracks the local computation. Its
value is changed depending on the composition rules.  We extend
$\toCT$ with a new local task counter $c$: $C, c \toCT C',
c'$. Initially, $c_\text{init}=0$. The scheduler rules are updated to
modify $c$ in the following ways: \zcref[S]{eq:rule-ctSpawn} does not
change~$c$, hence, $c'=c$.
\zcref[S]{eq:rule-ctIdleInterval,eq:rule-ctIdleIntervalMarker,eq:rule-ctIdleLocalEmpty,eq:rule-ctIdleLocalEmptyMarker,eq:rule-ctIdleLocal,eq:rule-ctIdleLocalMarker}
set $c=0$, as a task is (re-)started. In
\zcref{eq:rule-ctActiveSuspend,eq:rule-ctActiveEmpty}, we leave $c$ as
is: $c'=c$. Finally,
\zcref{eq:rule-ctActivePop,eq:rule-ctActivePush,eq:rule-ctActiveEval,eq:rule-ctActivePushCall}
decrease $c$ by one: $c'=c-1$.

\begin{lemma}\label{lem:called-fn}
  Assume an execution
  $\ctCfg{\sh_0}{\mathit{act}_0}{t_0}{q_{G,0}}{q_0} \toCT^{+}
  \ctCfg{\sh_1}{\mathit{act}_1}{t_1}{q_{G,1}}{q_1}$. Further assume
  that $\spawnEv{f}{h}\in\sh_1$, $\lastEv(\sh_1)\ne\spawnEv{f}{h}$ and
  $\createEv{\pv{foo}}{f}{\many{v}}\in\sh_1$. Then either
  $(\pv{foo},h)\in q_{G,1}$, $(\pv{foo},h)\in q_{1}$, or
  $\evalEv{\pv{foo}}{h}{\many{v}}\in\sh_1$.
\end{lemma}
\begin{proof}[Proof Sketch]
  This lemma is proven by a straightforward induction over the
  scheduling rules. Once a future is spawned its corresponding task
  marker is in a queue until its execution begins, in which case an
  evaluation event is emitted.
\end{proof}

Naturally, the absence of locally infinite computations preclude
infinitely decreasing chains of scheduling distances:

\begin{lemma}\label{lem:finite-chain}
  If there are no infinite local computations, then there is no
  infinitely decreasing chain of scheduling distances.
\end{lemma}

\begin{theorem}[Quiescent Fairness With Counter]
  Starting from an initial judgement, the current thread scheduler
  (adjusted with a counter) is quiescent fair if there are no infinite
  local computations: If $C_0 \toCT C_1 \toCT C_2 \cdots$, then there
  is a quiescent fair execution
  $\mkTr{\sigma_0}, \asState{[\contP{f^\text{init}}{e}]}{\emptyset}
  \to^\ast \schedTrans{C_1} \to^\ast \schedTrans{C_2} \to^\ast \cdots$
  producing the same trace (modulo \textit{schedEv} events).
\end{theorem}
\begin{proof}
  We must prove two parts: First, for any execution derived by
  $\toCT$, there exists an execution using $\to$. Then, we prove the
  quiescent fairness.

  \textbf{Existence:} We proceed by induction over $n$ derivation
  steps with $\toCT$. As both the semantics using~$\to$ and $\toCT$
  have the same starting point, the base case is free. Assume a
  derivation in $n$ steps for $C,c$. We apply $\toCT$ once more to get
  $C',c'$. By the induction hypothesis, we have a $\to$-execution for
  $\schedTrans{C}$. We show via a case distinction over the rule
  applied to get $C',c'$ that a derivation for $\schedTrans{C'}$
  exists. Let in the following
  $C=\ctCfg{\sh}{\mathit{act}}{t}{q_G}{q}$.

  \newcommand{\Zcref}[1]{{\zcsetup{cap=true}\zcref{#1}}}

  \begin{description}
  \item[\Zcref{eq:rule-ctSpawn}] The task marker $(\pv{foo},f)$ is
    pushed to the end of the local queue~$q$ and appends a schedule
    event and a repetition of $\sh$'s last state to $\sh$. By the
    definition of $\schedTransSym$, this has no effect on
    $\schedTrans{C}$. Thus, $\schedTrans{C}$ is already the required
    $\to$-configuration.
  \item[\Zcref{eq:rule-ctActiveEval}:] Then \zcref{eq:rule-Eval} is
    applicable. The resulting traces are equivalent according to
    \enquote{$\schedTransSym$.} Additionally, per definition,
    $\schedTrans{\ctState{[\contP{f}{e},\many{k}]}{t}{q_G}{q}}$
    differs from
    $\schedTrans{\ctState{[\contP{f}{e'},\many{k}]}{t}{q_G}{q}}$ only
    by $\contP{f}{e}$ being replaced by $\contP{f}{e'}$. This is the
    same change done by \zcref{eq:rule-Eval}. Thus, we have a
    $\to$-derivation.
  \item[\Zcref{eq:rule-ctActivePushCall}:] Then
    \zcref{eq:rule-PushCall} is applicable, since its preconditions
    follow from those of \zcref{eq:rule-ctActivePushCall}. The old
    configuration must be of the form
    $\schedTrans{\ctState{[\many{k}]}{t}{q_G}{q}}=\asState{[\many{k}]}{Q}$. The
    resulting configuration is:
    \[
      \asState{[\contP{f}{e'},\many{k}]}{Q}
    \]
    This is the same as
    $\schedTrans{\ctState{[\contP{f}{e'},\many{k}]}{t}{q_G}{q}}$, as
    only the active task's stack is removed and then the same stack
    with new top element $\contP{f}{e'}$ is added; this is equivalent
    to the active task $[\contP{f}{e'},\many{k}]$. Similarly, the
    traces in the resulting configuration match according to
    $\schedTransSym$.
  \item[\Zcref{eq:rule-ctActivePush}:] Then \zcref{eq:rule-PushTask}
    is applicable. For the same reasons as in
    \zcref{eq:rule-ctActivePushCall}, the resulting configurations
    match.
  \item[\Zcref{eq:rule-ctActivePop}:] Then \zcref{eq:rule-Pop} is
    applicable. Clearly, $\schedTrans{C'}$ matches the result of
    \zcref{eq:rule-Pop}.
  \item[\Zcref{eq:rule-ctActiveEmpty}:] Then \zcref{eq:rule-Empty} is
    applicable. The resulting configurations match per the definition
    of \enquote{$\schedTransSym$.}
  \item[\Zcref{eq:rule-ctActiveSuspend}:] Then \zcref{eq:rule-Suspend}
    is applicable. The resulting configurations clearly agree.
  \item[\Zcref{eq:rule-ctIdleInterval}] moves the first stack from
    $q_G$ to the active slot. Then, \zcref{eq:rule-ReschedTask} is
    applicable.  Since $\sh$ in unchanged, the result
    $\schedTrans{C}=\asCfg{\sh}{[\contP{f}{e},\many{k}]}{Q \setminus
      \{[\contP{f}{e},\many{k}]\}}$ is the desired $\to$-derivation.
  \item[\Zcref{eq:rule-ctIdleIntervalMarker}:] Since the preconditions
    of \zcref{eq:rule-ctIdleIntervalMarker} imply those of
    \zcref{eq:rule-SchedTask}, the latter is applicable on
    configuration
    $\schedTrans{\sh},\schedTrans{\ctState{\Idle}{t}{q_G}{q}}=\schedTrans{\sh},\asState{\Idle}{Q}$,
    resulting in configuration
    $\schedTrans{\sh},\asState{[\contP{f}{e'}]}{Q}$. Since
    $\schedAw{q_G}=\schedAw{q'_G}$, this is exactly the translation
    $\schedTrans{C'}$.
  \item[\Zcref{eq:rule-ctIdleLocalEmpty}:] Same argument as for
    \zcref{eq:rule-ctIdleInterval}.
  \item[\Zcref{eq:rule-ctIdleLocalEmptyMarker}:] Same argument as for
    \zcref{eq:rule-ctIdleIntervalMarker}.
  \item[\Zcref{eq:rule-ctIdleLocal}:] Similar to the cases for
    \zcref{eq:rule-ctIdleInterval,eq:rule-ctIdleLocalEmpty} but with
    $q$ and $q'$ substituted for $q_G$ and $q'_G$, respectively.
  \item[\Zcref{eq:rule-ctIdleLocalMarker}:] Analogous to the cases for
    \zcref{eq:rule-ctIdleIntervalMarker,eq:rule-ctIdleLocalEmptyMarker},
    except with the observation that $\schedAw{q}=\schedAw{q'}$.
  \end{description}

  \noindent\textbf{Mapping:} As in \cite{HH2024}, we define a mapping between
  configurations in $\to$ and $\toCT$ for a fixed
  $\toCT$-execution. Assume such a $\toCT$-execution
  \begin{equation}\label{eq:ct-run}
    C^\text{ct}_0 \toCT C^\text{ct}_1 \toCT^+ C^\text{ct}_n \toCT
    \cdots
  \end{equation}
  For most $\toCT$-applications, there is a direct one-to-one
  correspondence to a $\to$-step. However, \zcref{eq:rule-ctSpawn},
  which schedules a future, has no equivalent. For this rule, we have
  no step (i.e., a no-op) on the $\to$-side. Hence the change from at
  least one step ($\toCT^+$) to the reflexive-transitive closure
  ($\to^\ast$).  Then, for the execution~\zcref[noname]{eq:ct-run}
  there is a $\to$-execution
  \begin{equation}
    \label{eq:mod-run}
    C'_0 \to C'_1 \to^\ast C'_m \to \cdots
  \end{equation}

  Then we can define a mapping $\pi:\mathbb{N}\rightarrow\mathbb{N}$
  such that the execution~\zcref[noname]{eq:mod-run} can be written as
  \begin{equation}\label{eq:mod-trans-run}
    \schedTrans{C^\text{ct}_{\pi(0)}} \to \schedTrans{C^\text{ct}_{\pi(1)}} \to^+ \schedTrans{C^\text{ct}_{\pi(m)}} \to \cdots
  \end{equation}
  namely:
  $\pi(i) = i + \#(\text{applications of \zcref{eq:rule-ctSpawn}})$.
  
  \textbf{Fairness:} Assume a $\toCT$-execution
  like~\zcref[noname]{eq:ct-run} and the corresponding
  $\to$-execution~\zcref[noname]{eq:mod-trans-run}. We need to show
  that the letter is quiescent fair, i.e.:
  \begin{align*}
    \forall m, e, f.\ \Big( &\forall n > m.\\
                            &\quad\big(\exists Q.\ \schedTrans{\ctState{\mathit{act}_{\pi(n)}}{t_{\pi(n)}}{q_{G,\pi(n)}}{q_{\pi(n)}}}=\asState{\Idle}{Q} \Rightarrow \contP{f}{e} \in \enabled{C_{\pi(n)}} \big)\\
                            &\qquad\Rightarrow \exists n' > m.\ \big(\contP{f}{e} \in \getConts{\schedTrans{\ctState{\mathit{act}_{\pi(n')}}{t_{\pi(n')}}{q_{G,\pi(n')}}{q_{\pi(n')}}}} \big) \Big)
  \end{align*}

  Let $m,e,f$ be arbitrary but fixed and assume that $\contP{f}{e}$ is
  enabled in $C_{\pi(n)}$ whenever there is no active task in
  $C_{\pi(n)}$'s translation. We need to show that eventually, a
  scheduler rule occurs that produces~$e$. We consider each possible
  rule to yield $\contP{f}{e}\in\enabled{C_{\pi(n)}}$. There are three
  possibilities for add $\contP{f}{e}$:
  \begin{enumerate*}[(i)]
  \item it is the consequence of evaluating an expression~$e'$ in
    \zcref{eq:rule-Eval} or \zcref{eq:rule-Pop};
  \item it is the result of evaluating the first step of a scheduled task
    (\zcref{eq:rule-SchedTask}); or
  \item it is the result of evaluating a synchronous call
    (\zcref{eq:rule-PushCall}) or a local future
    (\zcref{eq:rule-PushTask}).
  \end{enumerate*}
  We show fairness by proving that either $e$ is produced, the
  quiescent fairness precondition does not hold, or the scheduling
  distance for expression~$e'$ or $(\pv{foo},f)$ decreases. By
  \zcref{lem:finite-chain}, there is no infinitely decreasing chain of
  distances, thus $e'$ or $(\pv{foo}, f)$ are eventually scheduled,
  leading to $e$ being produced.

  We assume the quiescent fairness precondition; hence, $\contP{f}{e}$
  is reachable. Now there are two possible cases: (case
  \textsc{Produce}) expression~$e$ is produced or (case
  \textsc{Decrease}) the scheduling distance of $e'$ or $(\pv{foo},f)$
  decreases.
  
  \begin{description}
  \item[\Zcref{eq:rule-Eval}:] We restate \zcref{eq:rule-Eval} with
    the translated configuration
    \[\schedTrans{C_{\pi(n)}}=\schedTrans{\sh_{\pi(n)}},\schedTrans{\ctState{\mathit{act}_{\pi(n)}}{t_{\pi(n)}}{q_{G,\pi(n)}}{q_{\pi(n)}}}:\]
    \gcondrule{Eval-CT-Trans}{
      \schedTrans{\ctState{\mathit{act}_{\pi(n)}}{t_{\pi(n)}}{q_{G,\pi(n)}}{q_{\pi(n)}}}=\asState{\left[\contP{f}{e'},\many{k}\right]}{Q} \\
      \sigma=\last(\schedTrans{\sh_{\pi(n)}}) \quad
      \concatTr{\tau}{\contP{f}{e}} \in \valP{f}{e'} \\
      \rho\text{ concretizes }\tau \quad
      \wf{\chopTrSem{\schedTrans{\sh_{\pi(n)}}}{\rho(\tau)}} }{
      \schedTrans{C_{\pi(n)}}\to\chopTrSem{\schedTrans{\sh_{\pi(n)}}}{\rho(\tau)},\asState{\left[\contP{f}{e},\many{k}\right]}{Q}
    }
    In this case, because the translated configuration has an active
    task, so does the configuration $C_{\pi(n)}$ and the top of this
    task, $\contP{f}{e'}$ can be evaluated yielding a well-formed
    trace. Hence, only the
    \zcref{eq:rule-ctSpawn,eq:rule-ctActiveEval} can be applied; all
    others have either an idle worker or they cannot evaluate the top
    of the active task.

    If we can apply \zcref{eq:rule-ctSpawn}, then there must be a
    spawn event that has not yet been responded too and the distance
    is of the form $(k, 0, 1, c)$, which is reduced by
    \zcref{eq:rule-ctSpawn} to $(k, 0, 0, c)$. So we have the case
    \textsc{Decrease}.

    Otherwise, we have
    $\ctDist{\contP{f}{e'}}{C_{\pi(n)},c_{\pi(n)}}=(k,0,0,c_{\pi(n)})$
    for some $k\ge 1$. Either $k=1$, which implies that $e'$ is the
    next expression to be executed and is not suspended. Then,
    clearly, \zcref{eq:rule-Eval} is applicable, yielding~$e$
    (\textsc{Produce}).

    If $k>1$, $e'$ is not evaluated by this application of
    \zcref{eq:rule-ctActiveEval}. But we obtain case \textsc{Decrease}
    by the observation that
    $\ctDist{\contP{f}{e'}}{C_{\pi(n)+1},c_{\pi(n)+1}}=(k,0,0,c_{\pi(n)}-1)$.

  \item[\Zcref{eq:rule-PushCall}] executes a synchronous call. We
    restate it in the translated setting:
    \gcondrule{PushCall-CT-Trans}{
      \schedTrans{\ctState{\mathit{act}_{\pi(n)}}{t_{\pi(n)}}{q_{G,\pi(n)}}{q_{\pi(n)}}}=\asState{[\many{k}]}{Q} \quad \sigma=\last(\sh_{\pi(n)}) \quad f\in\FId \\
      \lookup{\pv{foo}}=\rust{\kwFn foo($\many{x}$) $e$}\quad
      \concatTr{\tau}{\contP{f}{e}} \in \valP{f}{\rust{\kwFn foo($\many{x}$) $e'$}} \\
      \rho\text{ concretizes }\tau \quad
      \wf{\chopTrSem{\schedTrans{\sh_{\pi(n)}}}{\rho(\tau)}}
    }{
      \schedTrans{C_{\pi(n)}} \to
      \chopTrSem{\schedTrans{\sh_{\pi(n)}}}{\rho(\tau)},\asState{[\contP{f}{e},\many{k}]}{Q}
    }
    By the same argument as for \zcref{eq:rule-Eval}, only the rules
    \zcref{eq:rule-ctSpawn,eq:rule-ctActivePushCall} are applicable. The
    exact same argument further gives us that we either
    \textsc{Produce} $e$ or \textsc{Decrease} the distance.

  \item[\Zcref{eq:rule-PushTask}] executes a local future. We restate
    \zcref{eq:rule-PushTask} in our translated setting:
    \gcondrule{PushTask-CT-Trans}{
      \schedTrans{\ctState{\mathit{act}_{\pi(n)}}{t_{\pi(n)}}{q_{G,\pi(n)}}{q_{\pi(n)}}}=\asState{[\many{k}]}{Q} \quad \sigma=\last(\sh_{\pi(n)}) \quad f\in\FId \\
      \lookup{\pv{foo}}=\rust{\kwAsync~\kwFn foo($\many{x}$) $e$}\quad
      \concatTr{\tau}{\contP{f}{e}} \in \valP{f}{\rust{\kwAsync~\kwFn foo($\many{x}$) $e'$}} \\
      \rho\text{ concretizes }\tau \quad
      \wf{\chopTrSem{\schedTrans{\sh_{\pi(n)}}}{\rho(\tau)}}
    }{
      \schedTrans{C_{\pi(n)}} \to
      \chopTrSem{\schedTrans{\sh_{\pi(n)}}}{\rho(\tau)},\asState{[\contP{f}{e},\many{k}]}{Q}
    }
    We repeat the argument for \zcref{eq:rule-PushCall} but with
    \zcref{eq:rule-ctSpawn,eq:rule-ctActivePush} applicable.

  \item[\Zcref{eq:rule-Pop}] removes the top continuation of the
    active task's stack if its computation is completed. Restated in
    our terms, the rule goes:
    \gcondrule{Val-CT-Trans}{
      \schedTrans{\ctState{\mathit{act}_{\pi(n)}}{t_{\pi(n)}}{q_{G,\pi(n)}}{q_{\pi(n)}}}=\asState{[\contP{f}{e''},\contP{f}{e'},\many{k}]}{Q} \quad \sigma=\last(\sh_{\pi(n)}) \\
      \concatTr{\tau}{\cont{e}} \in \valP{f}{e'} \quad \rho\text{ concretizes }\tau \quad \wf{\chopTrSem{\schedTrans{\sh_{\pi(n)}}}{\rho(\tau)}}
    }{
      \schedTrans{C_{\pi(n)}} \to \schedTrans{\sh_{\pi(n)}},
      \asState{[\contP{f}{e},\many{k}]}{Q} }
    Again, the same argument as for
    \zcref{eq:rule-Eval,eq:rule-PushCall,eq:rule-PushTask} but now
    \zcref{eq:rule-ctSpawn,eq:rule-ctActivePop} can be applied.

  \item[\Zcref{eq:rule-Empty}:] Has no effect on the set of enabled
    continuations. Thus, it cannot have produced $e'$.
  \item[\Zcref{eq:rule-Suspend}:] Has no effect on the set of enabled
    continuations. Thus, it cannot have produced $e'$.
  \item[\Zcref{eq:rule-SchedTask}] executes a spawned future. We
    restate \zcref{eq:rule-SchedTask} in our translated setting:
    \gcondrule{SchedTask-CT-Trans}{
      \schedTrans{\ctState{\mathit{act}_{\pi(n)}}{t_{\pi(n)}}{q_{G,\pi(n)}}{q_{\pi(n)}}} = \asState{\Idle}{Q} \quad \sigma=\last(\sh_{\pi(n)}) \quad f\in\FId \\
      \lookup{\pv{foo}}=\rust{\kwAsync~\kwFn foo($\many{x}$) $e$} \quad
      \concatTr{\tau}{\contP{f}{e}} \in \valP{f}{\rust{\kwAsync~\kwFn foo($\many{x}$) $e'$}} \\
      \rho\text{ concretizes }\tau \quad
      \wf{\chopTrSem{\schedTrans{\sh_{\pi(n)}}}{\rho(\tau)}}
    }{
      \schedTrans{C_{\pi(n)}} \to
      \chopTrSem{\schedTrans{\sh_{\pi(n)}}}{\rho(\tau)},
      \asState{\left[\contP{f}{e}\right]}{Q}
    }
    Since the evaluation of \rust{\kwAsync~\kwFn foo($\many{x}$) $e$}
    starts with an evaluation event, well-formedness gives us that a
    spawn event for $f$ and \pv{foo} exists in
    $\schedTrans{\sh_{\pi(n)}}$ but no evaluation event for $f$ and
    \pv{foo}.  If $\lastEv\left(\schedTrans{\sh_{\pi(n)}}\right)$ is a
    spawn event, we apply \zcref{eq:rule-ctSpawn} and start our
    argument on its result.

    \zcref[S]{lem:called-fn} gives us that $(\pv{foo},f)$ is in one of
    the scheduler's queues and its distance is greater or equal to
    $(1,1,0,c_{\pi(n)})$. We consider each case of the application
    $C_{\pi(n)}\toCT C_{\pi(n)+1}$.
    \begin{description}
    \item[\Zcref{eq:rule-ctSpawn}] can only be applied to schedule a
      task with a different future identifier. The distance is, hence,
      $\ctDist{(\pv{foo},
        f)}{C_{\pi(n)},c_{\pi(n)}}=(k,i,1,c_{\pi(n)})$ for some
      $k\ge 1$ and $i\in\{0,1\}$ before the step and
      $\ctDist{(\pv{foo},
        f)}{C_{\pi(n)+1},c_{\pi(n)}}=(k,i,0,c_{\pi(n)})$
      afterwards. Thus, the scheduling distance decreases
      (\textsc{Decrease}).
    \item[\Zcref{eq:rule-ctActiveEval}:] Again we have case
      \textsc{Decrease}, as the distance goes from
      $\ctDist{(\pv{foo},
        f)}{C_{\pi(n)},c_{\pi(n)}}=(k,0,0,c_{\pi(n)})$, with $k > 1$,
      to $\ctDist{(\pv{foo}, f)}{C_{\pi(n)+1}}=(k,0,0,c_{\pi(n)}-1)$.
    \item[\Zcref{eq:rule-ctActivePushCall}] evaluates a synchronous
      call. Then, case \textsc{Decrease} holds, as the distance goes
      from
      $\ctDist{(\pv{foo},
        f)}{C_{\pi(n)},c_{\pi(n)}}=(k,0,0,c_{\pi(n)})$, for some
      $k > 1$, to
      $\ctDist{(\pv{foo},
        f)}{C_{\pi(n)+1},c_{\pi(n)}}=(k,0,0,c_{\pi(n)}-1)$.
    \item[\Zcref{eq:rule-ctActivePush}:] Same argument as for
      \zcref{eq:rule-ctActivePushCall}.
    \item[\Zcref{eq:rule-ctActivePop}:] Before the application, we
      have
      $\ctDist{(\pv{foo},
        f)}{C_{\pi(n)},c_{\pi(n)}}=(k,0,0,c_{\pi(n)})$ for some
      $k > 1$. Afterwards,
      $\ctDist{(\pv{foo},
        f)}{C_{\pi(n)+1},c_{\pi(n)+1}}=(k,0,0,c_{\pi(n)}-1)$, which is
      a decrease of the scheduling distance (\textsc{Decrease}).
    \item[\Zcref{eq:rule-ctActiveEmpty}:] Then the scheduling
      distance decreases from
      $\ctDist{(\pv{foo}, f)}{C_{\pi(n)}}=(k,0,0,c_{\pi(n)})$, for
      some $k > 1$, to
      $\ctDist{(\pv{foo},
        f)}{C_{\pi(n)+1}}=(k-1,1,0,c_{\pi(n)}-1)$. We have
      \textsc{Decrease}.
    \item[\Zcref{eq:rule-ctActiveSuspend}:] Same argument as for
      \zcref{eq:rule-ctActiveEmpty}.
    \item[\Zcref{eq:rule-ctIdleInterval}:] Before the application, we
      have
      $\ctDist{(\pv{foo},
        f)}{C_{\pi(n)},c_{\pi(n)}}=(k,1,0,c_{\pi(n)})$ for some
      $k > 1$. Afterwards, we have
      $\ctDist{(\pv{foo}, f)}{C_{\pi(n)+1},c_{\pi(n)}}=(k,0,0,0)$,
      i.e., a decrease of the scheduling distance (\textsc{Decrease}).
    \item[\Zcref{eq:rule-ctIdleIntervalMarker}:] Either
      $\ctDist{(\pv{foo},
        f)}{C_{\pi(n)},c_{\pi(n)}}=(1,1,0,c_{\pi(n)})$, in which case
      the correct continuation is produced (\textsc{Produce}), or the
      distance is $(k,1,0,c_{\pi(n)})$ with $k>1$ and the application
      reduces the scheduling distance (\textsc{Decrease}) to
      $\ctDist{(\pv{foo}, f)}{C_{\pi(n)+1},c_{\pi(n)+1}}=(k,0,0,0)$.
    \item[\Zcref{eq:rule-ctIdleLocalEmpty}:] Same argument as for
      \zcref{eq:rule-ctIdleInterval}.
    \item[\Zcref{eq:rule-ctIdleLocalEmptyMarker}:] Same argument as
      for \zcref{eq:rule-ctIdleIntervalMarker}.
    \item[\Zcref{eq:rule-ctIdleLocal}:] Same as for
      \zcref{eq:rule-ctIdleInterval,eq:rule-ctIdleLocalEmpty}.
    \item[\Zcref{eq:rule-ctIdleLocalMarker}:] Same argument as for
      \zcref{eq:rule-ctIdleIntervalMarker,eq:rule-ctIdleLocalEmptyMarker}.
    \end{description}
  \item[\Zcref{eq:rule-ReschedTask}:] We follow the same argument as
    for \zcref{eq:rule-SchedTask} but now with a task $[\many{k}]$
    rather than a task marker. \qedhere
  \end{description}
\end{proof}


\section{Relation of Quiescent and Tokio Fairness}
\label{app:relat-quiesc-tokio}

\begin{proof}[Proof of \zcref{thm:quiescent-tokio}]
  As established in \zcref{app:quiesc-fairn-curr}, we know that the
  current thread runtime is quiescent fair. We show that the premisses
  of Tokio fairness are stronger than those of quiescent fairness and
  the conclusion follows from that of quiescent fairness. Since the
  first assumption of \zcref{def:tokio-fair} is absent from the
  definition of quiescent fairness, we can add it without altering the
  conclusion.

  The second premiss is stronger than the assumption of locally finite
  computations: For quiescent fairness we only assume that each task
  completes or suspends in a final amount of steps; we do not require
  a fixed upper bound for all tasks.

  Finally, note that the conclusions of both definitions are similar
  except that $\exists n' > m.\ e\in\newExpr{C_{n'}}{C_{n'+1}}$ in
  quiescent fairness is strengthened to
  $\exists n'.\ m < n' \le m + D_\text{max} \wedge
  e\in\newExpr{C_{n'}}{C_{n'+1}}$ in Tokio fairness. Hence, we need to
  show that the $n'$ we find for quiescent fairness is always less
  than $m+D_\text{max}$. The distance function defined in
  \zcref{app:quiesc-fairn-curr} can be used to find $D_\text{max}$.

  We note that, because the local computations must complete within
  $P_\text{max}$ steps, we can make all components of the scheduling
  distance non-negative by adding $P_\text{max}$ to the final
  component. Then, for a scheduling distance $(a,b,c,d)$ we construct
  $D\in\mathbb{N}$ by
  \[
    D:=a\circ b \circ c \circ d'
  \]
  where $\circ$ denotes concatenation of digits and $d'$ is
  $d+P_\text{max}$ extended with leading zeroes to the number of
  digits of $P_\text{max}$. For example, if $P_\text{max}=100$, the
  scheduling distance $(32,1,0,-64)$ is mapped to $D_0 = 3210036$. As
  can be easily seen, this mapping preserves the order over distances.

  By definition of the scheduling distance (\zcref{def:ct-dist}), the
  first component is bounded by the sizes of both queues plus one:
  $0\le a\le |\getConts{q}|+|\getConts{q_G}|+1$; because the number of
  tasks is less than $T_\text{max}$, we have $a\le
  T_\text{max}+1$. The second and third components are either $0$ or
  $1$; the last component is $0\le d\le P_\text{max}$. Thus, all
  distance numbers $D$ for all expressions at all points are less than
  or equal to
  \[
    D_\text{max}:=(T_\text{max}+1) \circ 1 \circ 1 \circ
    P_\text{max}\ . \qedhere
  \]
\end{proof}

As the proof above uses no properties of the current thread runtime
besides its quiescent fairness and the existence of a scheduling
distance, the following corollary follows immediately:

\begin{corollary}\label{col:quiescent-impl-tokio}
  A quiescent fair scheduler, where there is a scheduling distance
  $\ctDist{\contP{f}{e}}{C,c}$ such that, for each continuation marker
  $\contP{f}{e}$ that is in $\enabled{C_i}$ for some configuration in
  an execution, the distance for $\contP{f}{e}$ is either decreased by
  each step or $e$ is produced and the distance is bounded by a
  maximum, is also Tokio fair.
\end{corollary}

It is unknown whether there exist any quiescent fair schedulers that
are not Tokio fair.


\section{Auxiliary Definitions and Theorems for \ac{DL}}
\label{sec:dl-defs}

\begin{definition}[Domain]\label{def:dom}
  We fix the domain $\domain^A$ of each sort $A$ as follows:
  \begin{itemize}
  \item $\domain^\bool=\{\trueSem, \falseSem\}$,
  \item $\domain^\Int=\mathbb{Z}$,
  \item $\domain^\Field$ an infinite set of field identifiers,
  \item
    $\domain^{\Place_A}=\lim_{n\rightarrow\infty}\domain_n^{\Place_A}$ where
    \begin{itemize}
    \item $\domain_0^{\Place_A}=\{\semPlace{\pv{x}}\mid\pv{x}:A\in\PVSym\}$ and
    \item $\domain_{i+1}^{\Place_A}=\domain_i^{\Place_A}\cup\{\semArrPlace{p}{n} \mid p\in\domain_i^{\Place_A}, n\in\mathbb{N}\}\cup\{\semFieldPlace{p}{f} \mid p\in\domain_i^{\Place_A}, f\in\domain^\Field\}$
    \end{itemize}
  \item $\domain^{\RefS{A}}=\{\semRefS{a}\mid a\in\domain^A\}$ for all
    sorts $A$,
  \item
    $\domain^{\RefM{A}}=\{\semRefM{p}\mid p\in\domain^{\Place_A}\}$
    for all sorts $A$,
  \item
    $\domain^{(A_1,\ldots,A_n)}=\domain^{A_1}\times \cdots
    \domain^{A_n}$,
  \item
    $\domain^{[A;n]}=\{\semArray{d_n,\ldots,d_n} \mid
    d_1,\ldots,d_n\in\domain^A\}$,
  \item and domains for each datatype declared in the program.
  \end{itemize}
\end{definition}

\begin{definition}[Correspondence of \RustyDL Domain and LAGC
  Values]\label{def:corr-dom}
  Let $\domain$ be a \RustyDL Domain and $\sigma$ an \ac{LAGC}
  state. We define the binary relation
  $\dlLagcSim_{\sigma}\subseteq\domain\times\Val$ (pronounced
  \enquote{$d$ corresponds to $v$ w.r.t.~$\sigma$}) inductively as
  follows:
  \begin{enumerate}
  \item $n\dlLagcSim_{\sigma} n$ for each $n\in\mathbb{Z}$,
  \item $\trueSem\dlLagcSim_{\sigma}\trueSem$ and
    $\falseSem\dlLagcSim_{\sigma}\falseSem$,
  \item $\semRefS{a}\dlLagcSim_{\sigma}\sRefVal{\pv{x}}{b}$ for each
    $\pv{x}\in\Var$ and for all $b\in\BId$, where
    $a\dlLagcSim_{\sigma} \sigma(\pv{x})$,
  \item $\semRefM{p}\dlLagcSim_{\sigma}\mRefVal{p'}{b}$ for each
    $p\in\domain^{\Place_A},b\in\BId,p'\in\Place$ iff $p$ and $p'$
    correspond,
  \item $(d_1,\ldots,d_n)\dlLagcSim_{\sigma} (v_1,\ldots,v_n)$ iff
    $d_1\dlLagcSim_{\sigma}v_1, \ldots, d_n\dlLagcSim_{\sigma}v_n$,
  \item $\semArray{d_1,\ldots,d_n}\dlLagcSim_{\sigma}[v_1,\ldots,v_n]$
    iff
    $d_1\dlLagcSim_{\sigma}v_1, \ldots, d_n\dlLagcSim_{\sigma}v_n$.
  \end{enumerate}
\end{definition}

\begin{definition}[Correspondence of \RustyDL States and LAGC
  States]\label{def:corr-state}
  Let $\Kstruct=(\domain,\interp,\States,\rho)$ be \RustyDL Kripke
  structure, $s\in\States$ a \RustyDL state, and $\sigma$ an LAGC
  state. We define the binary relation $\dlLagcSim$ (pronounced
  \enquote{corresponds}) as follows: $s\dlLagcSim\sigma$ iff for all
  $\pv{x}: A\in\PVSym$ such that $\pv{x}\in\dom(\sigma)$, we have
  $s(\pv{x})\dlLagcSim_{\sigma}\sigma(\pv{x})$.
\end{definition}

We fix some interpretations of functions and relations:
\begin{itemize}
\item $\interp(+)(n,m)=n + m$ with $n,m\in\domain^\Int$ (similar for
  all integer functions),
\item $\interp(\refSSym)(d)=\semRefS{d}$ with $d\in\domain$,
\item $\interp(\refMSym)(p)=\semRefM{p}$,
\item $\interp(\derefS)(\semRefS{d})=d$,
\item $\interp(\dlfunc{inU32})=\{n\in\mathbb{N} \mid 0 \le n < 2^{32}\}$.
\end{itemize}


\section{Additional Soundness Proofs of Dynamic Logic Calculus Rules}
\label{sec:addit-soundn-proofs}

\begin{theorem}[Soundness of rule
  $\ruleName{ifElseSplit}$]\label{thm:sound-ifElseSplit}
  Let $\Gamma,\Delta\in 2^{\DLFor}$ be sets of formulas,
  $\updateContext\in\Updates$ be an update, $\mathit{se}$ a literal or
  a program variable of type \lstinline|bool|,
  \rust{\{$\mathit{be}_0$\}} and \rust{\{$\mathit{be}_1$\}} be block
  expressions, $\pi~\rust{\kwIf $\mathit{se}$
    \{$\mathit{be}_0$\} \kwElse \{$\mathit{be}_1$\}}~\omega$ be a
  program fragment, and $\phi\in\DLFor$ be a formula. If
  \begin{align*}
    &\models\sequent{\applyUp{\updateContext}{(\mathit{se}\keyeq\keybooltrue)}}{\applyUp{\updateContext}{{\ptstub{\{$\mathit{be}_0$\}}}}}\\
    &\models\sequent{\applyUp{\updateContext}{(\mathit{se}\keyeq\keyboolfalse)}}{\applyUp{\updateContext}{{\ptstub{\{$\mathit{be}_1$\}}}}}
  \end{align*}
  then the following holds:
  \[
    \models\sequent{}{\applyUp{\updateContext}{\ptstub{\kwIf $\mathit{se}$ \{$\mathit{be}_0$\} \kwElse \{$\mathit{be}_1$\}}}}.
  \]
\end{theorem}
\begin{proof}
  We assume
  \begin{align}
    &\models\sequent{\applyUp{\updateContext}{(\mathit{se}\keyeq\keybooltrue)}}{\applyUp{\updateContext}{{\ptstub{\{$\mathit{be}_0$\}}}}}\label{eq:ifElseSplit-asm-1}\\
    &\models\sequent{\applyUp{\updateContext}{(\mathit{se}\keyeq\keyboolfalse)}}{\applyUp{\updateContext}{{\ptstub{\{$\mathit{be}_1$\}}}}}\label{eq:ifElseSplit-asm-2}
  \end{align}
  It is sufficient to show
  \begin{align}
    \validSOne\ptstub{\kwIf $\mathit{se}$ \{$\mathit{be}_0$\} \kwElse \{$\mathit{be}_1$\}}.\label{eq:ifElseSplit-goal-1}
  \end{align}
  We infer from the assumptions that
  \begin{align}
    &\validSOne\mathit{se}\keyeq\keybooltrue\rightarrow\ptstub{\{$\mathit{be}_0$\}}\label{eq:ifElseSplit-asm-3}\\
    &\validSOne\mathit{se}\keyeq\keyboolfalse\rightarrow\ptstub{\{$\mathit{be}_1$\}}\label{eq:ifElseSplit-asm-4}
  \end{align}
  
  We know that \rust{\kwIf $\mathit{se}$ \{$\mathit{be}_0$\} \kwElse
    \{$\mathit{be}_1$\}} is the active expression. By the construction
  of the local evaluation rules, we know that it is also the next
  (sub-)expression to be evaluated and added to the trace $\sh$. Thus,
  we can use the local evaluation. We split over the structure of
  $\mathit{se}$.
  \begin{case}[$\mathit{se}\in\Val$]\label{case:ifElseSplit-1}
    Then $\mathit{se}$ is either \lstinline|true| or
    \lstinline|false|. We cover the former case as both are analogous.
    $\mathit{se}=\mbox{\lstinline|true|}$: Then,
    $\valB{\sigma_1}{\rust{\rsTrue}}=\{\concatTr{\mkTr{\sigma_1}}{\cont{\trueSem}}\}$. From
    that we infer
    \begin{align*}
      &\valB{\sigma_1}{\rust{\kwIf~\rsTrue \{$\mathit{be}_0$\} \kwElse \{$\mathit{be}_1$\}}}\\
      &\quad=\{\concatTr{\tau}{\cont{\rust{\kwIf $e$ \{$\mathit{be}_0$\} \kwElse \{$\mathit{be}_1$\}}}}\mid\concatTr{\tau}{\cont{e}}\in\valB{\sigma_1}{\rust{\rsTrue}}\}\\
      &\quad=\{\concatTr{\mkTr{\sigma_1}}{\cont{\rust{\kwIf $\trueSem$ \{$\mathit{be}_0$\} \kwElse \{$\mathit{be}_1$\}}}}\}
    \end{align*}
    Further,
    \(\valB{\sigma_1}{\rust{\kwIf $\trueSem$
        \{$\mathit{be}_0$\} \kwElse
        \{$\mathit{be}_1$\}}}\nonumber
    =\{\concatTr{\mkTr{\sigma_1}}{\cont{\rust{\{$\mathit{be}_0$\}}}}\}
    \). Hence, the only possible composition is
    \begin{align*}
      \mkTr{\sigma_1},\cont{\rust{\kwIf~\rsTrue \{$\mathit{be}_0$\} \kwElse \{$\mathit{be}_1$\}}}\rightarrow\mkTr{\sigma_1},\cont{\rust{\kwIf $\trueSem$ \{$\mathit{be}_0$\} \kwElse \{$\mathit{be}_1$\}}}     \rightarrow\mkTr{\sigma_1},\cont{\rust{\{$\mathit{be}_0$\}}}
    \end{align*}
    which ends in the start configuration of
    $\traces{\pi~\rust{\{$\mathit{be}_0$\}}~\omega, \sigma_1}$, i.e.,
    we have the relation
    $\trSuff{\traces{\rust{\{$\mathit{be}_0$\}},\sigma_1}}{\traces{\rust{\kwIf
          $\mathit{se}$
          \{$\mathit{be}_0$\} \kwElse
          \{$\mathit{be}_1$\}},\sigma_1}}$. Thus, we
    simplify~\eqref{eq:ifElseSplit-goal-1} to
    \begin{align}
      \validSOne\ptstub{\{$\mathit{be}_0$\}}.\label{eq:ifElseSplit-goal-2}
    \end{align}

    Assumption~\eqref{eq:ifElseSplit-asm-3} can be simplified, by the
    definition of the implication, to
    \begin{align}
      \nvalidSOne\mathit{se}\keyeq\keybooltrue\quad\text{or}\quad\validSOne\ptstub{\{$\mathit{be}_0$\}}\label{eq:ifElseSplit-asm-5}
    \end{align}
    Since $\mathit{se}=\rust{\rsTrue}$, we know that
    $\validSOne\mathit{se}\keyeq\keybooltrue$, which means with
    \eqref{eq:ifElseSplit-asm-5} and \zcref{lem:suffix-impl} that the
    goal \eqref{eq:ifElseSplit-goal-2} holds.
  \end{case}
  \begin{case}[$\mathit{se}\in\Var$]\label{case:ifElseSplit-2}
    Let $v=\sigma_1(\mathit{se})$. Since $\mathit{se}$ is not
    symbolic, $v\in\Val$. Then
    $\valB{\sigma_1}{\mathit{se}}=\{\concatTr{\mkTr{\sigma_1}}{\cont{v}}\}$,
    because $\mathit{se}$ is of type \lstinline|bool|, which
    implements \lstinline|Copy|. We present the case for $v=\trueSem$,
    as the case $v=\falseSem$ is similar. Then the local evaluation
    yields
    \begin{align*}
      &\valB{\sigma_1}{\rust{\kwIf $\mathit{se}$ \{$\mathit{be}_0$\} \kwElse \{$\mathit{be}_1$\}}}\\
      &\quad=\{\concatTr{\tau}{\cont{\rust{\kwIf $e$ \{$\mathit{be}_0$\} \kwElse \{$\mathit{be}_1$\}}}}\mid\concatTr{\tau}{\cont{e}}\in\valB{\sigma_1}{\mathit{se}}\}\\
      &\quad=\{\concatTr{\mkTr{\sigma_1}}{\cont{\rust{\kwIf $\trueSem$ \{$\mathit{be}_0$\} \kwElse \{$\mathit{be}_1$\}}}}\}
    \end{align*}
    As in \zcref{case:ifElseSplit-1}, we know that the only possible
    composition is
    \begin{align*}
      \mkTr{\sigma_1},\cont{\rust{\kwIf $\mathit{se}$ \{$\mathit{be}_0$\} \kwElse \{$\mathit{be}_1$\}}}
      \rightarrow\mkTr{\sigma_1},\cont{\rust{\kwIf $\trueSem$ \{$\mathit{be}_0$\} \kwElse \{$\mathit{be}_1$\}}}
      \rightarrow\mkTr{\sigma_1},\cont{\rust{\{$\mathit{be}_0$\}}}\ .
    \end{align*}
    By the same argument as in \zcref{case:ifElseSplit-1},
    goal~\eqref{eq:ifElseSplit-goal-1} holds. \qedhere
  \end{case}
\end{proof}

\[
  \seqRuleC{move}{
    \sequent{}{\applyUp{\updateContext}{\applyUp{\pv{x}\upd\pv{y}\parupd\pv{y}\upd c}{\ptstub{()}}}}
  }{
    \sequent{}{\applyUp{\updateContext}{\ptstub{x = y}}}
  }{$\pv{y}:A$, $A$ is not \lstinline|Copy|, and $c:A$ a fresh constant}
\]

\begin{theorem}[Soundness of Rule
  $\ruleName{move}$]\label{thm:sound-move}
  Let $\Gamma,\Delta\in 2^{\DLFor}$ be sets of formulas,
  $\updateContext\in\Updates$ be an update, \pv{x} be a program
  variable whose type is not \lstinline|Copy| and not a mutable
  reference, \pv{y} be a program variable of the same type as~\pv{x},
  $c$ a constant of \pv{y}'s type that does not occur in
  $\Gamma,\Delta,\phi$, or $\updateContext$, $\pi~\rust{x = y}~\omega$
  be a program fragment, and $\phi\in\DLFor$ be a formula. If
  \[
    \models\sequent{}{\applyUp{\updateContext}{\applyUp{\pv{x}\upd\pv{y}\parupd\pv{y}\upd c}{\ptstub{()}}}}
  \]
  then the following holds:
  \[
    \models\sequent{}{\applyUp{\updateContext}{\ptstub{x = y}}}.
  \]
\end{theorem}
\begin{proof}
  We assume
  \[
    \models\sequent{}{\applyUp{\updateContext}{\applyUp{\pv{x}\upd\pv{y}\parupd\pv{y}\upd
          c}{\ptstub{()}}}}
  \]

  It is sufficient to show
  \begin{align}
    \validSOne\ptstub{x = y}.\label{eq:move-goal-1}
  \end{align}
  We infer from the assumptions that
  \begin{align}
    \label{eq:move-asm-1}
    \validSOne\applyUp{\pv{x}\upd\pv{y}\parupd\pv{y}\upd
          c}{\ptstub{()}}
  \end{align}
  Let $s_2=\valS(\pv{x}\upd\pv{y}\parupd\pv{y}\upd c)(s_1)$. Then we
  simplify \zcref{eq:move-asm-1} to
  \begin{align}
    \label{eq:move-asm-2}
    \validSTwo\ptstub{()}
  \end{align}
  Let $v=\sigma_1(\pv{y})$. Following \zcref{sec:eval-sub-expr} and
  then \zcref{eq:local-pv} we have the local evaluation
  \begin{align}
    \label{eq:move-local}
    \valB{\sigma_1}{\rust{x = y}}=\{\concatTr{\tau}{\cont{\rust{x = $e$}}} \mid \concatTr{\tau}{\cont{e}}\in\valB{\sigma_1}{\pv{y}}\}
    = \{\concatTr{\mkTr{\sigma_1}}{\cont{\rust{x = $v$}}}\}
  \end{align}
  The local evaluation \zcref{eq:local-assign-pv} further yields
  $\valB{\sigma_1}{\rust{x =
      $v$}}=\{\concatTr{\consTr{\update{\sigma_1}{\pv{x}}{v}}{\mkTr{\sigma_1}}}{\cont{\rust{()}}}\}$.
  Therefore, the only possible composition is
  \begin{align*}
    \mkTr{\sigma_1},\cont{\pi~\rust{x = y}~\omega}
    \to\mkTr{\sigma_1},\cont{\pi~\rust{x = $v$}~\omega}
    \to\consTr{\update{\sigma_1}{\pv{x}}{v}}{\mkTr{\sigma_1}},\cont{\pi~\rust{()}~\omega}
  \end{align*}
  Which implies that
  $\trSuff{\traces{\pi~\rust{()}~\omega,
      \update{\sigma_1}{\pv{x}}{v}}}{\traces{\pi~\rust{x =
        y}~\omega,\sigma_1}}$. By \zcref{lem:suffix-impl}, we know
  that goal \zcref[noname]{eq:move-goal-1} holds if
  $\validSThree\ptstub{()}$ holds, with
  $s_3=\valSTwo(\pv{x}\upd
  v)(s_1)$. Assumption~\zcref[noname]{eq:move-asm-2} gives us a
  stronger formula than this: It gives us validity of $\ptstub{()}$
  with no information about $\pv{y}'s$ value. Thus, by
  \zcref[noname]{eq:move-asm-2}, we are done.
\end{proof}

\[
  \seqRule{assignAddU32}{
    \sequent{\applyUp{\updateContext}{\dlfunc{inU32}(\mathit{se}_1+\mathit{se}_2)}}{\applyUp{\updateContext}{\applyUp{\pv{x}\upd\mathit{se}_1+\mathit{se}_2}{\ptstub{()}}}}\\
    \sequent{\applyUp{\updateContext}{\neg\dlfunc{inU32}(\mathit{se}_1+\mathit{se}_2)}}{\applyUp{\updateContext\ptstub{\panic()}}}
  }{
    \sequent{}{\applyUp{\updateContext}{\ptstub{x = $\mathit{se}_1$ + $\mathit{se}_2$}}}
  }
\]

\begin{theorem}[Soundness of Rule
  $\ruleName{assignAddU32}$]\label{thm:sound-assignAddU32}
  Let $\Gamma,\Delta\in 2^{\DLFor}$ be sets of formulas,
  $\updateContext\in\Updates$ be an update, \pv{x} be a program
  variable of type \lstinline|u32|, $\mathit{se}_1,\mathit{se}_2$
  either literals or program variables of type \lstinline|u32|,
  $\pi~\rust{x = $\mathit{se}_1$ + $\mathit{se}_2$}~\omega$ be a
  program fragment, and $\phi\in\DLFor$ be a formula. If
  \begin{align*}
    &\models\sequent{\applyUp{\updateContext}{\dlfunc{inU32}(\mathit{se}_1+\mathit{se}_2)}}{\applyUp{\updateContext}{\applyUp{\pv{x}\upd\mathit{se}_1+\mathit{se}_2}{\ptstub{()}}}}\\
    &\models\sequent{\applyUp{\updateContext}{\neg\dlfunc{inU32}(\mathit{se}_1+\mathit{se}_2)}}{\applyUp{\updateContext\ptstub{\panic()}}}
  \end{align*}
  then the following holds:
  \[
    \models\sequent{}{\applyUp{\updateContext}{\ptstub{x = $\mathit{se}_1$ + $\mathit{se}_2$}}}.
  \]
\end{theorem}
\begin{proof}
  We only present the case where
  $\mathit{se}_1,\mathit{se}_2\in\Val$. The other two cases are
  handled as in the above proofs.

  We need to show
  \begin{align}
    \validSOne\ptstub{x =
        $\mathit{se}_1$ + $\mathit{se}_2$}.\label{eq:assignAddU32-goal-1}
  \end{align}
  We infer from the assumptions that
  \begin{align}
    &\validSOne\dlfunc{inU32}(\mathit{se}_1+\mathit{se}_2) \rightarrow \applyUp{\pv{x}\upd\mathit{se}_1+\mathit{se}_2}{\ptstub{()}}\label{eq:assignAddU32-asm-2}\\
    &\validSOne\neg\dlfunc{inU32}(\mathit{se}_1+\mathit{se}_2)\rightarrow \ptstub{\panic()}\label{eq:assignAddU32-asm-3}
  \end{align}
  Let $\sigma_1$ be an \ac{LAGC} state such that
  $s_1\dlLagcSim\sigma_1$. Further, let
  $n=\valSOne(\mathit{se}_1)+\valSOne(\mathit{se}_2)$. By the
  semantics if $\dlfunc{inU32}$, we know that
  $\validSOne\dlfunc{inU32}(\mathit{se}_1+\mathit{se}_2)$ iff
  $0 \le n < 2^{32}$.
  \begin{case}[$\dlfunc{inU32}(n)$ holds]
    Then the first case of \zcref{eq:local-add} holds and the
    evaluation yields
    \[\valB{\sigma_1}{\rust{$\mathit{se}_1$ + $\mathit{se}_2$}}=\{\concatTr{\mkTr{\sigma_1}}{\cont{n}}\}\]
    which makes the local evaluation of the assignment
    \begin{align}
      \valB{\sigma_1}{\rust{x = $\mathit{se}_1$ + $\mathit{se}_2$}}&=\{\concatTr{\mkTr{\sigma_1}}{\cont{\rust{x = $n$}}}\}\\
      \valB{\sigma_1}{\rust{x = $n$}}&=\{\concatTr{\consTr{\update{\sigma_1}{\pv{x}}{n}}{\mkTr{\sigma_1}}}{\cont{\rust{()}}}\}
    \end{align}
    Hence, the only possible composition is
    \begin{align*}
      \mkTr{\sigma_1},\cont{\rust{x = $\mathit{se}_1$ + $\mathit{se}_2$}}
      \to\mkTr{\sigma_1},\cont{\rust{x = $n$}}
      \to\consTr{\update{\sigma_1}{\pv{x}}{n}}{\mkTr{\sigma_1}},\cont{\rust{()}}
    \end{align*}
    Since $s_2\dlLagcSim\sigma_2$,
    $\trSuff{\traces{\pi~\rust{()}~\omega,\sigma_2}}{\traces{\pi~\rust{x
          =
          $\mathit{se}_1$ + $\mathit{se}_2$}~\omega,\sigma_1}}$. Hence,
    with \zcref{lem:suffix-impl}, we simplify the goal to
    \begin{equation}
      \label{eq:assignAddU32-goal-2}
      \validSTwo\ptstub{()}
    \end{equation}
    As $\validSOne\dlfunc{inU32}(\mathit{se}_1+\mathit{se}_2)$, we can
    simplify assumption~\eqref{eq:assignAddU32-asm-2} to
    \eqref{eq:assignAddU32-goal-2}
  \end{case}
  \begin{case}[$\dlfunc{inU32}(n)$ does not hold]
    Then the second case of \zcref{eq:local-add} holds and the
    evaluation yields
    \[\valB{\sigma_1}{\rust{$\mathit{se}_1$ + $\mathit{se}_2$}}=\{\concatTr{\mkTr{\sigma_1}}{\cont{\rust{\panic()}}}\}\]
    And
    $\valB{\sigma_1}{\rust{\panic()}}=\{\concatTr{\consTr{\sigma_1}{\consTr{\panicEv}{\mkTr{\sigma_1}}}}{\cont{()}}\}$,
    which means that we have the suffix relation
    $\trSuff{\traces{\pi~\rust{\panic()}~\omega,\sigma_1}}{\traces{\pi~\rust{$\mathit{se}_1$
          + $\mathit{se}_2$}~\omega}}$. With
    $\validSOne\neg\dlfunc{inU32}(n)$, we simplify
    assumption~\eqref{eq:assignAddU32-asm-3} to
    \begin{equation}
      \label{eq:assignAddU32-asm-4}
      \validSOne\ptstub{\panic()}
    \end{equation}
    which, by \zcref{lem:suffix-impl}, gives us
    \eqref{eq:assignAddU32-goal-1}.
  \end{case}
\end{proof}

\[
  \seqRule{borrowShared}{
    \sequent{}{\applyUp{\updateContext}{\applyUp{\pv{x}\upd\refS{\pv{y}}}{\ptstub{()}}}}
  }{
    \sequent{}{\applyUp{\updateContext}{\ptstub{x = \&y}}}
  }
\]

\begin{theorem}[Soundness of Rule
  $\ruleName{borrowShared}$]\label{thm:sound-borrowShared}
  Let $\Gamma,\Delta\in 2^{\DLFor}$ be sets of formulas,
  $\updateContext\in\Updates$ be an update, \pv{x} be a program
  variable of some immutable reference type, \pv{y} a program
  variable, $\pi~\rust{x = \&y}~\omega$ be a program fragment, and
  $\phi\in\DLFor$ be a formula. If
  \begin{align*}
    \models\sequent{}{\applyUp{\updateContext}{\applyUp{\pv{x}\upd\refS{\pv{y}}}{\ptstub{()}}}}
  \end{align*}
  then the following holds:
  \[
    \models\sequent{}{\applyUp{\updateContext}{\ptstub{x = \&y}}}.
  \]
\end{theorem}
\begin{proof}
  We assume
  \begin{align}
    \models\sequent{}{\applyUp{\updateContext}{\applyUp{\pv{x}\upd\refS{\pv{y}}}{\ptstub{()}}}}\label{eq:borrowShared-asm-1}
  \end{align}
  We only need to show
  \begin{align}
    \validSOne\ptstub{x = \&y}.\label{eq:borrowShared-goal-1}
  \end{align}
  We infer from the assumptions that
  \begin{align}
    \validSOne\applyUp{\pv{x}\upd\refS{\pv{y}}}{\ptstub{()}}\label{eq:borrowShared-asm-2}
  \end{align}
  Let $s_2=\valSOne(\pv{x}\upd\refS{\pv{y}})(s_1)$. Then we can transform \eqref{eq:borrowShared-asm-2} to
  \begin{align}
    \validSTwo\ptstub{()}\label{eq:borrowShared-asm-3}
  \end{align}

  We know that \rust{x = \&y} is the active expression. By the
  construction of the local evaluation rules, we know that it is also
  a (sub-)expression to be evaluated and added to the
  trace~$\sh$. Thus we can use the local evaluation.

  Let $v=\sigma_1(\pv{y})$. We know
  $\valB{\sigma_1}{\pv{y}}=\{\concatTr{\mkTr{\sigma_1}}{\cont{v}}\}$. Further,
  \begin{align*}
    \valB{\sigma_1}{\rust{\&y}}&=\{\concatTr{\mkTr{\sigma_1}}{\cont{\sRefVal{\pv{y}}{b}}}\mid b\in\BId\}\\
    \valB{\sigma_1}{\rust{x = \&y}}&=\{\concatTr{\tau}{\cont{\rust{x = $e$}}}\mid\concatTr{\tau}{\cont{e}}\in\valB{\sigma_1}{\rust{\&y}}\}\\
                          &=\{\concatTr{\mkTr{\sigma_1}}{\cont{\rust{x = $\sRefVal{\pv{y}}{b}$}}} \mid b\in\BId\}\\
    \valB{\sigma_1}{\rust{x = $\sRefVal{\pv{y}}{b}$}}&=\{\concatTr{\consTr{\update{\sigma_1}{\pv{x}}{\sRefVal{\pv{y}}{b}}}{\mkTr{\sigma_1}}}{\cont{\rust{()}}} \mid b\in\BId\}
  \end{align*}
  Hence, we know that the only possible compositions starting in state
  $\sigma_1$ have the shape
  \begin{align*}
    \mkTr{\sigma_1},\cont{\pi~\rust{x = \&y}~\omega}
    \to\mkTr{\sigma_1},\cont{\rust{x = $\sRefVal{\pv{y}}{b_1}$}}
    \to\consTr{\update{\sigma_1}{\pv{x}}{\sRefVal{\pv{y}}{b_1}}}{\mkTr{\sigma_1}},\cont{()}
  \end{align*}
  for $b_1\in\BId$ concrete, fresh borrow identifiers.

  We know that
  $\update{\sigma_1}{\pv{x}}{\sRefVal{\pv{y}}{b}}=\sigma_2$, so we
  know that the suffix relation
  $\trSuff{\traces{\pi~\rust{()}~\omega,
      \sigma_2}}{\traces{\pi~\rust{x =
        \&y}~\omega,\sigma_1}}$. Therefore, by \zcref{lem:suffix-impl}
  and assumption~\eqref{eq:borrowShared-asm-3}, we have
  \eqref{eq:borrowShared-goal-1}.
\end{proof}

\[
  \seqRuleC{derefShared}{
    \sequent{}{\applyUp{\updateContext}{\applyUp{\pv{x}\upd\derefS(\pv{y})}{\ptstub{()}}}}
  }{ \sequent{}{\applyUp{\updateContext}{\ptstub{x = *y}}}
  }{$\pv{x}:A,\pv{y}:\RefS{A}$, $A$ is \lstinline|Copy|}
\]

\begin{theorem}[Soundness of Rule
  $\ruleName{derefShared}$]\label{thm:sound-derefShared}
  Let $\Gamma,\Delta\in 2^{\DLFor}$ be sets of formulas,
  $\updateContext\in\Updates$ be an update, \pv{x} be a program
  variable of type \lstinline|T|, \pv{y} a program variable of type
  \lstinline|&T|, $\pi~\rust{x = *y}~\omega$ be a program fragment,
  and $\phi\in\DLFor$ be a formula. If
  \begin{align*}
    \models\sequent{}{\applyUp{\updateContext}{\applyUp{\pv{x}\upd\derefS(\pv{y})}{\ptstub{()}}}}
  \end{align*}
  then the following holds:
  \[
    \models\sequent{}{\applyUp{\updateContext}{\ptstub{x = *y}}}.
  \]
\end{theorem}
\begin{proof}
  We assume
  \begin{align}
    \models\sequent{}{\applyUp{\updateContext}{\applyUp{\pv{x}\upd\derefS(\pv{y})}{\ptstub{()}}}}\label{eq:derefShared-asm-1}
  \end{align}
  It is sufficient to show
  \begin{align}
    \validSOne\ptstub{x = \&y}.\label{eq:derefShared-goal-1}
  \end{align}
  We infer from the assumptions that
  \begin{align}
    \validSOne\applyUp{\pv{x}\upd\derefS(\pv{y})}{\ptstub{()}}\label{eq:derefShared-asm-2}
  \end{align}
  Let $s_2=\valSOne(\pv{x}\upd\derefS(\pv{y}))(s_1)$. Then we can transform \eqref{eq:derefShared-asm-2} to
  \begin{align}
    \validSTwo\ptstub{()}\label{eq:derefShared-asm-3}
  \end{align}

  We know that \rust{x = *y} is the active expression. By the
  construction of the local evaluation rules, we know that it is also
  a (sub-)expression to be evaluated and added to the trace
  $\sh$. Thus we can use the local evaluation.

  Let $\sigma_1$ be an LAGC state such that $s_1\sim\sigma_1$, let
  $\sRefVal{\pv{z}}{b}=\sigma_1(\pv{y})$ with some program
  variable~\pv{z} and some borrow Id~$b$, and let
  $v:=\sigma_1(\pv{z})$. We know that
  $\valB{\sigma_1}{\pv{y}}=\{\concatTr{\mkTr{\sigma_1}}{\cont{\sRefVal{\pv{z}}{b}}}\}$.
  Further,
  \begin{align*}
    \valB{\sigma_1}{\rust{*y}}&=\{\concatTr{\tau}{\cont{\rust{*$e$}}}\mid\concatTr{\tau}{\cont{e}}\in\valB{\sigma_1}{\pv{y}}\}\\
                              &=\{\concatTr{\mkTr{\sigma_1}}{\cont{\rust{*$\sRefVal{\pv{z}}{b}$}}}\}\\
    \valB{\sigma_1}{\rust{*$\sRefVal{\pv{z}}{b}$}}&=\{\concatTr{\mkTr{\sigma_1}}{\cont{v}}\}\\
    \valB{\sigma_1}{\rust{x = *y}}&=\{\concatTr{\tau}{\cont{\rust{x = $e$}}}\mid\concatTr{\tau}{\cont{e}}\in\valB{\sigma_1}{\rust{*y}}\}\\
                              &=\{\concatTr{\mkTr{\sigma_1}}{\cont{\rust{x = *$\sRefVal{\pv{z}}{b}$}}}\}\\
    \valB{\sigma_1}{\rust{x = *$\sRefVal{\pv{z}}{b}$}}&=\{\concatTr{\tau}{\cont{\rust{x = $e$}}}\mid\concatTr{\tau}{\cont{e}}\in\valB{\sigma_1}{\rust{*$\sRefVal{\pv{z}}{b}$}}\}\\
                              &=\{\concatTr{\mkTr{\sigma_1}}{\cont{\rust{x = $v$}}}\}\\
    \valB{\sigma_1}{\rust{x = $v$}}&=\{\concatTr{\consTr{\update{\sigma_1}{\pv{x}}{v}}{\mkTr{\sigma_1}}}{\cont{\rust{()}}}\}
  \end{align*}

  Hence, we know that the only possible composition starting in state
  $\sigma_1$ is
  \begin{align*}
    &\mkTr{\sigma_1},\cont{\pi~\rust{x = *y}~\omega}
    \to\mkTr{\sigma_1},\cont{\pi~\rust{x = *$\sRefVal{\pv{z}}{b}$}~\omega}\\
    \to&\mkTr{\sigma_1},\cont{\pi~\rust{x = $v$}~\omega}
    \to\consTr{\update{\sigma_1}{\pv{x}}{v}}{\mkTr{\sigma_1}},\cont{\pi~\rust{()}~\omega}\ .
  \end{align*}
  Let $\semRefS{w}=s_1(\pv{y})$. We know that
  $\valSOne(\derefS(\pv{y}))=w$, which implies $s_2(\pv{x})=w$. Since
  $s_1\sim\sigma_1$, $\semRefS{w}$ must correspond to
  $\sRefVal{\pv{z}}{b}$, which implies $w\sim_{\sigma_1} v$. Let
  $\sigma_2:=\update{\sigma_1}{\pv{x}}{v}$. Then, $s_2\sim\sigma_2$.

  We know that
  $\trSuff{\traces{\pi~\rust{()}~\omega,\sigma_2}}{\traces{\pi~\rust{x
        = *y}~\omega,\sigma_1}}$. Therefore, by
  \zcref{lem:suffix-impl} and assumption~\eqref{eq:derefShared-asm-3},
  we have \eqref{eq:derefShared-goal-1}.
\end{proof}

\[
  \rwRule{derefOfSharedRef}{\derefS(\refS{t})}{t}
\]

\begin{theorem}[Soundness and Completeness of Rule
  $\ruleName{derefOfSharedRef}$]\label{thm:sound-derefOfSharedRef}
  Let $\Kstruct=(\domain,\interp,\States,\rho)$ be a Kripke structure,
  $s\in\States$ a state, and $\beta$ a variable
  assignment. Then \[\valS(\derefS(\refS{t}))=\valS(t)\]
\end{theorem}
\begin{proof}
  \begin{align*}
    &\valS(\derefS(\refS{t}))\\
    =&\interp(\derefS)(\valS(\refS{t}))\\
    =&\interp(\derefS)(\interp(\refSSym)(\valS(t)))\\
    =&\interp(\derefS)(\semRefS{\valS(t)})\\
    =&\valS(t)
  \end{align*}
\end{proof}

\[
  \seqRule{borrowMut}%
  {\sequent{}{\applyUp{\updateContext}{\applyUp{\pv{x}\upd\refM{\place{\pv{y}}}}{\ptstub{()}}}}}%
  {\sequent{}{\applyUp{\updateContext}{\ptstub{x = \&\kwMut y}}}}
\]

\begin{theorem}[Soundness of Rule
  $\ruleName{borrowMut}$]\label{thm:sound-borrowMut}
  Let $\Gamma,\Delta\in 2^{\DLFor}$ be sets of formulas,
  $\updateContext\in\Updates$ be an update, \pv{x} be a program
  variable of some mutable reference type, \pv{y} a program variable,
  $\pi~\rust{x = \&\kwMut y}~\omega$ be a program fragment, and
  $\phi\in\DLFor$ be a formula. If
  \begin{align*}
    \models\sequent{}{\applyUp{\updateContext}{\applyUp{\pv{x}\upd\refM{\place{\pv{y}}}}{\ptstub{()}}}}
  \end{align*}
  then the following holds:
  \[
    \models\sequent{}{\applyUp{\updateContext}{\ptstub{x = \&\kwMut y}}}.
  \]
\end{theorem}
\begin{proof}
  It suffices to show
  \begin{align}
    \validSOne\ptstub{x = \&\kwMut y}.\label{eq:borrowMut-goal-1}
  \end{align}
  We infer from the assumptions that
  \begin{align}
    \validSOne\applyUp{\pv{x}\upd\refM{\place{\pv{y}}}}{\ptstub{()}}\label{eq:borrowMut-asm-2}
  \end{align}
  Let $s_2=\valSOne(\pv{x}\upd\refM{\place{\pv{y}}})(s_1)$. Then we
  can transform \eqref{eq:borrowMut-asm-2} to
  \begin{align}
    \validSTwo\ptstub{()}\label{eq:borrowMut-asm-3}
  \end{align}

  We know that \rust{x = \&\kwMut y} is the active expression. By the
  construction of the local evaluation rules, we know that it is also
  a (sub-)expression to be evaluated and added to the trace
  $\sh$. Thus we can use the local evaluation.

  Let $\sigma_1$ be an \ac{LAGC} state such that
  $s_1\sim\sigma_1$. Let $v=\sigma_1(\pv{y})$. We know
  $\valB{\sigma_1}{\pv{y}}=\{\concatTr{\mkTr{\sigma_1}}{\cont{v}}\}$. Further,
  \begin{align*}
    \valB{\sigma_1}{\rust{\&\kwMut y}}&=\{\concatTr{\mkTr{\sigma_1}}{\cont{\mRefVal{\pv{y}}{b}}}\mid b\in\BId\}\\
    \valB{\sigma_1}{\rust{x = \&\kwMut y}}&=\{\concatTr{\tau}{\cont{\rust{x = $e$}}}\mid\concatTr{\tau}{\cont{e}}\in\valB{\sigma_1}{\rust{\&\kwMut y}}\}\\
                                   &=\{\concatTr{\mkTr{\sigma_1}}{\cont{\rust{x = $\mRefVal{\pv{y}}{b}$}}} \mid b\in\BId\}\\
    \valB{\sigma_1}{\rust{x = $\mRefVal{\pv{y}}{b}$}}&=\{\concatTr{\consTr{\update{\sigma_1}{\pv{x}}{\mRefVal{\pv{y}}{b}}}{\mkTr{\sigma_1}}}{\cont{\rust{()}}}\}
  \end{align*}
  Hence, we know that the only possible composition starting in state
  $\sigma_1$ is
  \begin{align*}
    \mkTr{\sigma_1},\cont{\pi~\rust{x = \&\kwMut y}~\omega}
    \to\mkTr{\sigma_1},\cont{\rust{x = $\mRefVal{\pv{y}}{b}$}}
    \to\consTr{\update{\sigma_1}{\pv{x}}{\mRefVal{\pv{y}}{b}}}{\mkTr{\sigma_1}},\cont{()}\ .
  \end{align*}
  Let
  $\sigma_2:=\update{\sigma_1}{\pv{x}}{\sRefVal{\pv{y}}{b}}$. Then,
  $\trSuff{\traces{\pi~\rust{()}~\omega,\sigma_2}}{\traces{\pi~\rust{x
        = \&y}~\omega},\sigma_1}$.  We know
  $s_2\sim\sigma_2$. Therefore, by \zcref{lem:suffix-impl} and
  assumption~\eqref{eq:borrowMut-asm-3}, we have
  \eqref{eq:borrowMut-goal-1}.
\end{proof}

\[
  \seqRule{mutate}{
    \sequent{}{\applyUp{\updateContext}{\applyUp{\pv{x}\dupd\mathit{se}}{\ptstub{()}}}}
  }{
    \sequent{}{\applyUp{\updateContext}{\ptstub{*x = $\mathit{se}$}}}
  }
\]

\begin{theorem}[Soundness of Rule
  $\ruleName{mutate}$]\label{thm:sound-mutate}
  Let $\Gamma,\Delta\in 2^{\DLFor}$ be sets of formulas,
  $\updateContext\in\Updates$ be an update, \pv{x} be a program
  variable of type \lstinline|&mut T|, $\mathit{se}$ a literal or a
  program variable of type \lstinline|T|,
  $\pi~\rust{*x = $\mathit{se}$}~\omega$ be a program fragment, and
  $\phi\in\DLFor$ be a formula. If
  \begin{align*}
    \models\sequent{}{\applyUp{\updateContext}{\applyUp{\pv{x}\dupd\mathit{se}}{\ptstub{()}}}}
  \end{align*}
  then the following holds:
  \[
    \models\sequent{}{\applyUp{\updateContext}{\ptstub{*x = $\mathit{se}$}}}.
  \]
\end{theorem}
\begin{proof}
  It suffices to show
  \begin{align}
    \validSOne\ptstub{*x = $\mathit{se}$}.\label{eq:mutate-goal-1}
  \end{align}
  We infer from the assumptions that
  \begin{align}
    \validSOne\applyUp{\pv{x}\dupd\mathit{se}}{\ptstub{()}}\label{eq:mutate-asm-2}
  \end{align}
  Let $s_2=\valSOne(\pv{x}\dupd\mathit{se})(s_1)$. Then we can
  transform \eqref{eq:mutate-asm-2} to
  \begin{align}
    \validSTwo\ptstub{()}\label{eq:mutate-asm-3}
  \end{align}

  We know that \rust{*x = $\mathit{se}$} is the active expression. By
  the construction of the local evaluation rules, we know that it is
  also a (sub-)expression to be evaluated and added to the trace
  $\sh$. Thus we can use the local evaluation.

  We only show the case where $\mathit{se}$ is a value. The case where
  it is a variable is analogous.

  Let $\sigma_1$ be an LAGC state such that $s_1\sim\sigma_1$. Let $b$
  and $\pv{y}$ be such that $\sigma_1(\pv{x})=\mRefVal{\pv{y}}{b}$. We
  know
  \begin{align*}
    \valB{\sigma_1}{\rust{*x}} &= \{\concatTr{\mkTr{\sigma_1}}{\cont{\pv{y}}}\}\\
    \valB{\sigma_1}{\rust{*x = $\mathit{se}$}}&=\{\concatTr{\tau}{\cont{\rust{$\mathit{pe}$ = $\mathit{se}$}}}\mid\concatTr{\tau}{\cont{\mathit{pe}}}\in\valB{\sigma_1}{\rust{*x}}\}=\{\concatTr{\mkTr{\sigma_1}}{\cont{\rust{y = $\mathit{se}$}}}\}\\
    \valB{\sigma_1}{\rust{y = $\mathit{se}$}}&=\{\concatTr{\consTr{\update{\sigma_1}{\pv{y}}{\mathit{se}}}{\mkTr{\sigma_1}}}{\cont{\rust{()}}}\}
  \end{align*}
  Therefore, the only possible composition is
  \begin{align*}
    \mkTr{\sigma_1},\cont{\rust{*x = $\mathit{se}$}}
    \to\mkTr{\sigma_1},\cont{\rust{y = $\mathit{se}$}}
    \to\consTr{\update{\sigma_1}{\pv{y}}{\mathit{se}}}{\mkTr{\sigma_1}},\cont{\rust{()}}
  \end{align*}
  Let $\sigma_2:=\update{\sigma_1}{\pv{y}}{\mathit{se}}$. Since
  $\valSOne(\pv{x}\dupd\mathit{se})(s_1)(\pv{y})=\mathit{se}$ and is
  $s_1(\pv{z})$ for all program variables other than $\pv{y}$, we have
  $s_2\sim\sigma_2$.

  Then,
  $\trSuff{\traces{\pi~\rust{()}~\omega,
      \sigma_2}}{\traces{\pi~\rust{*x =
        $\mathit{se}$}~\omega, \sigma_1}}$, which with
  \zcref{lem:suffix-impl} and
  assumption~\zcref[noname]{eq:mutate-asm-3}, gives us
  goal~\zcref[noname]{eq:mutate-goal-1}.
\end{proof}

\subsection{Loop Rules}
\label{sec:loop-rules}

The remaining rules deal with loops using \emph{loop
  scopes}.~\autocite{SW2017}. The intuitive semantics of a loop scope
\rust{\loopScope{x}{$e$}} with index variable~\pv{x} and body~$e$, is
that $e$ is evaluated. If the evaluation encounters an expression that
exits the loop, such as a \lstinline{break} or \lstinline|return|,
\pv{x} is set to \lstinline|true|; the \lstinline|break|'s value
becomes the loop scope's value. Otherwise, \pv{x} is set to
\lstinline|false|. Thereby, it models exactly one iteration of the
loop and \pv{x} models whether we exit the loop.

In terms of \ac{LAGC} semantics, we define local rules that evaluate
the body~$e$ as for normal blocks, except when the next sub-expression
is a relevant expression, i.e., \lstinline|break|,
\lstinline|continue|, or \lstinline|return|.

When encountering an expression like \rust{\kwBreak $e$} in a loop
scope, we have:
\evalrule{ls-break}{
  \valD{\rust{\loopScope{x}{\{\kwBreak $v$; $\many{s}$ $e$\}}}}&=&\{\concatTr{\mkTr{\sigma}}{\cont{\rust{\{x = \rsTrue; $e$\}}}}\}
}

The local evaluation for \lstinline{continue} sets the value of \pv{x}
to \lstinline|false| and emits a \emph{continue event.}
\evalrule{ls-continue}{
  \valD{\rust{\loopScope{x}{\{\kwContinue; $\many{s}$ $e_2$\}}}}&=&\{\concatTr{\consTr{\update{\sigma}{\pv{x}}{\falseSem}}{\consTr{\continueEv}{\mkTr{\sigma}}}}{\cont{\rust{()}}}\}
}

We must also adapt the global composition rules to stop execution
after a \emph{continue event,} as the evaluation should not proceed
with the rest of the program after a \lstinline|continue| in a loop
scope. \zcref[S]{eq:rule-compLoopScope} is the changed version of
\zcref{eq:rule-sync-base}. All others are analogous.

\gcondrule{compLoopScope}{
  \sigma=\last(\sh) \quad \lastEv(\sh)\ne\panicEv \quad \lastEv(\sh)\ne\continueEv\\
  \concatTr{\tau}{\cont{e'}}\in\valD{e} \quad  \rho\text{ concretizes }\tau
  \quad \wf{\chopTrSem{\sh}{\rho(\tau)}} }{ \sh,
  \cont{e}\rightarrow\chopTrSem{\sh}{\rho(\tau)},\cont{e'} }

With these adapted rules, we can proceed with proving the relevant
\RustyDL calculus rules sound.

\[
  \seqRuleW{loopScopeInvBox}{
    \sequent{}{\applyUp{\updateContext}{\inv}}\\
    \sequent{}{\upl\updateContext\upr\upl\mathcal{V}\upr\Bigl(\inv \rightarrow \dlboxf{$\pi~$\loopScope{x}{$e$; \kwContinue}$~\omega$}{\bigl((\pv{x}\keyeq\keybooltrue \rightarrow \phi) \wedge (\pv{x}\keyeq\keyboolfalse \rightarrow \inv)\bigr)}\Bigr)}
  }{
    \sequent{}{\applyUp{\updateContext}{\dlboxf{$\pi~$\kwLoop $e~\omega$}{\phi}}}
  }
\]
where $\mathcal{V}$ is an \emph{anonymizing} update that overwrites
every variable written to or mutated in $e$,
see~\autocite{ABBHSU2016,DH2026TR} for details.

\begin{theorem}[Soundness of Rule
  $\ruleName{loopScopeInvBox}$]\label{thm:sound-loopScopeInvBox}
  Let $\Gamma,\Delta\in 2^{\DLFor}$ be sets of formulas,
  $\updateContext\in\Updates$ be an update,
  $\pv{x}: \mbox{\lstinline|bool|}$ be a fresh program variable,
  $\pi~\rust{\kwLoop $e$}~\omega$ be a program fragment, and
  $\phi,\inv\in\DLFor$ be formulas, and
  \[ \mathcal{V}=\pv{b}_1\upd b_1\parupd \cdots \parupd\pv{b}_n\upd
    b_n\parupd \pv{d}_1\dupd d_1\parupd\cdots\parupd\pv{d}_m\dupd d_m
  \]
  the \emph{anonymizing update,} where $\pv{b}_1,\ldots,\pv{b}_n$ are
  all variables written to in $e$, $\pv{d}_1,\ldots,\pv{d}_m$ are all
  mutable references mutated in $e$, and
  $b_1,\ldots,b_n,d_1,\ldots,d_m$ are fresh constants of appropriate
  sorts. If
  \begin{align*}
    &\models\sequent{}{\applyUp{\updateContext}{\inv}}\\
    &\models\begin{aligned}
      &\Gamma \turnstyle \upl\updateContext\upr\upl\mathcal{V}\upr\Bigl(\inv \rightarrow \dlboxf{$\pi~$\loopScope{x}{$e$; \kwContinue}$~\omega$}{\bigl(\\
      &\quad(\pv{x}\keyeq\keybooltrue \rightarrow \phi) \wedge (\pv{x}\keyeq\keyboolfalse \rightarrow \inv)\bigr)}\Bigr),\Delta
    \end{aligned}
  \end{align*}
  then the following holds:
  \[
    \models\sequent{}{\applyUp{\updateContext}{\dlboxf{$\pi~$\kwLoop $e~\omega$}{\phi}}}.
  \]
\end{theorem}
\begin{proof}
  We assume
  \begin{align}
    &\models\sequent{}{\applyUp{\updateContext}{\inv}}\label{eq:loop-asm-1}\\
    &\models\begin{aligned}
      &\Gamma \turnstyle \upl\updateContext\upr\upl\mathcal{V}\upr\Bigl(\inv \rightarrow \dlboxf{$\pi~$\loopScope{x}{$e$; \kwContinue}$~\omega$}{\bigl(\\
      &\quad(\pv{x}\keyeq\keybooltrue \rightarrow \phi) \wedge (\pv{x}\keyeq\keyboolfalse \rightarrow \inv)\bigr)}\Bigr)
    \end{aligned}\label{eq:loop-asm-2}
  \end{align}  

  Let $\Kstruct=(\domain,\interp,\States,\rho)$ be a random but
  arbitrary Kripke structure, $s\in\States$ a state, and $\beta$ a
  variable assignment. We want to show that
  \[
    \valid\sequent{}{\applyUp{\updateContext}{\dlboxf{$\pi~$\kwLoop $e~\omega$}{\phi}}}.
  \]
  Let $s_1=\valS(\updateContext)(s)$. Then we need to show
  \begin{align}
    \validSOne\dlboxf{$\pi~$\kwLoop $e~\omega$}{\phi}.\label{eq:loop-goal-1}
  \end{align}
  We infer from the assumptions that
  \begin{align}
    &\validSOne\inv\label{eq:loop-asm-3}\\
    &\label{eq:loop-asm-4}\validSOne\begin{aligned}
      &\upl\mathcal{V}\upr\Bigl(\inv \rightarrow \dlboxf{$\pi~$\loopScope{x}{$e$; \kwContinue}$~\omega$}{\bigl(\\
      &\quad(\pv{x}\keyeq\keybooltrue \rightarrow \phi) \wedge (\pv{x}\keyeq\keyboolfalse \rightarrow \inv)\bigr)}\Bigr)
    \end{aligned}
  \end{align}
  Let $s_2=\valSOne(\mathcal{V})(s_1)$. Then we can transform
  \eqref{eq:loop-asm-4} to
  \begin{equation}
    \begin{aligned}
      \label{eq:loop-asm-5}&\validSTwo \\
                           &\Bigl(\inv \rightarrow \dlboxf{$\pi~$\loopScope{x}{$e$; \kwContinue}$~\omega$}{\bigl((\pv{x}\keyeq\keybooltrue \rightarrow \phi) \wedge (\pv{x}\keyeq\keyboolfalse \rightarrow \inv)\bigr)}\Bigr)
    \end{aligned}
  \end{equation}
  If the loop does not terminate when started in $s_1$, then the above
  holds trivially. Hence, we assume that the loop is exited after $k$
  iterations. We proceed with an inductive proof to show that the
  formula $\inv$ holds before all iterations. Let $s^{(i)}$ denote the
  state before the $i$-th iteration.
  \begin{case}[Before first iteration ($i=1$)]
    Then, $s^{(i)}=s^{(1)}=s_1$. By assumption
    \zcref[noname]{eq:loop-asm-3}, we have $\validSOne\inv$, hence,
    the hypothesis holds.
  \end{case}
  \begin{case}[Induction step: Before iteration $i+1$]
    We assume $(\Kstruct,s^{(i)},\beta)\models\inv$. We execute the
    loop body~$e$ once in state $s^{(i)}$. Either iteration $i$ exits
    the loop or we continue with the loop's execution. In the former
    case, there is no iteration $i+1$ and we are done. In the latter
    case, the semantics of the loop scope and the continue expression
    (\zcref{eq:local-ls-continue}) mean that
    $(\Kstruct,s^{(i+1)},\beta)\models\inv$ holds iff the following
    holds:
    \[
      (\Kstruct,s^{(i)},\beta)\models\Bigl(\inv \rightarrow \dlboxf{$\pi$\loopScope{x}{$e$; \kwContinue}$\omega$}{\bigl((\pv{x}\keyeq\keybooltrue \rightarrow \phi) \wedge (\pv{x}\keyeq\keyboolfalse \rightarrow \inv)\bigr)}\Bigr)
    \]
    Any changes the body $e$ can affect are to
    $\pv{b}_1,\ldots,\pv{b}_n$ and the places borrowed by
    $\pv{d}_1,\ldots,\pv{d}_m$. Because these are overwritten by the
    anonymizing update~$\mathcal{V}$, and for all other program
    variables we have $s_2(\pv{y})=s^{(i+1)}(\pv{y})$, the above
    follows from assumption~\zcref[noname]{eq:loop-asm-5}. Hence,
    $(\Kstruct,s^{(i+1)},\beta)\models\inv$.
  \end{case}
  Now we know that iteration $k$ exits the loop and that
  $(\Kstruct,s^{(k)},\beta)\models\inv$. From this, it follows that
  \[
    (\Kstruct,s^{(k)},\beta)\models\Bigl(\inv \rightarrow \dlboxf{$\pi~$\loopScope{x}{$e$; \kwContinue}$~\omega$}{\bigl((\pv{x}\keyeq\keybooltrue \rightarrow \phi) \wedge (\pv{x}\keyeq\keyboolfalse \rightarrow \inv)\bigr)}\Bigr)
  \]
  Because we exit the loop in iteration~$k$, we simplify this to
  \[
    (\Kstruct,s^{(k)},\beta)\models\bigl(\inv \rightarrow
    \dlboxf{$\pi\omega$}{(\pv{x}\keyeq\keybooltrue \rightarrow
      \phi)}\bigr)
  \] which we further simplify to
  \begin{equation}
    \label{eq:loop-asm-6}
    (\Kstruct,s^{(k)},\beta)\models\dlboxf{$\pi\omega$}{\phi}
  \end{equation}
  Because we assume the loop's termination, we have
  \[\trSuff{\traces{\pi\omega,\sigma^{(k)}}}{\traces{\pi~\rust{\kwLoop
          $e$}~\omega,\sigma_1}}.\]
  By \zcref{eq:loop-asm-6,lem:suffix-impl} we have
  goal~\zcref[noname]{eq:loop-goal-1}.
\end{proof}

\[
  \seqRule{breakValue}{
    \sequent{}{\applyUp{\updateContext}{\ptstub{\{x = \rsTrue; $\mathit{se}$\}}}}
  }{\sequent{}{\applyUp{\updateContext}{\dia{\pi~\rust{\loopScope{x}{\{\kwBreak $\mathit{se}$; $\many{s}$ $e$\}}}~\omega}{\phi}}}}
\]

\begin{theorem}[Soundness of Rule
  $\ruleName{breakValue}$]\label{thm:sound-breakValue}
  Let $\Gamma,\Delta\in 2^{\DLFor}$ be sets of formulas,
  $\updateContext\in\Updates$ be an update,
  $\pv{x}: \mbox{\lstinline|bool|}$ be a program variable,
  $\pi~\rust{\loopScope{x}{\kwBreak $\mathit{se}$
      $\many{s}$ $e$}}~\omega$ be a program fragment, $\many{s}$ a
  series of statements, $e$ an expression, and $\phi\in\DLFor$ a
  formula. If
  \begin{equation*}
    \models\sequent{}{\applyUp{\updateContext}{\ptstub{\{x = \rsTrue; $\mathit{se}$\}}}}
  \end{equation*}
  then the following holds:
  \[
    \models\sequent{}{\applyUp{\updateContext}{\dia{\pi~\rust{\loopScope{x}{\{\kwBreak $\mathit{se}$; $\many{s}$ $e$\}}}~\omega}{\phi}}}.
  \]
\end{theorem}
\begin{proof}
  We assume
  \begin{equation}\label{eq:break-asm-1}
    \models\sequent{}{\applyUp{\updateContext}{\ptstub{\{x = \rsTrue; $\mathit{se}$\}}}}
  \end{equation}

  It suffices to show
  \begin{equation}\label{eq:break-goal-1}
    \validSOne\dia{\pi~\rust{\loopScope{x}{\{\kwBreak $\mathit{se}$; $\many{s}$ $e$\}}}~\omega}{\phi}.
  \end{equation}
  We infer from the assumption~\zcref[noname]{eq:break-asm-1} that
  \begin{equation}\label{eq:break-asm-2}
    \validSOne\ptstub{\{x = \rsTrue; $\mathit{se}$\}}
  \end{equation}

  We only consider the case where $\mathit{se}\in\Val$. The case where
  $\mathit{se}\in\Var$ is analogous.

  Let $\sigma_1$ be an \ac{LAGC} state with $s_1\dlLagcSim\sigma_1$ By
  \zcref{eq:local-ls-break}, we have
  \begin{equation}
    \valB{\sigma_1}{\rust{\loopScope{x}{\{\kwBreak $\mathit{se}$;
          $\many{s}$
          $e$\}}}}=\{\concatTr{\mkTr{\sigma_1}}{\cont{\rust{\{x = \rsTrue;
          $\mathit{se}$\}}}}\}
  \end{equation}
  Clearly, then,
  $\trSuff{\traces{\rust{\{x = \rsTrue;
        $\mathit{se}$\}},\sigma_1}}{\traces{\rust{\loopScope{x}{\{\kwBreak
          $\mathit{se}$; $\many{s}$ $e$\}}}, \sigma_1}}$. By
  \zcref{lem:suffix-impl}, this implies that the
  goal~\zcref[noname]{eq:break-goal-1} holds.
\end{proof}

\[
  \seqRule{continue}{
    \sequent{}{\applyUp{\updateContext}{\dlf{x = \rsFalse}{\phi}}}
  }{
    \sequent{}{\applyUp{\updateContext}{\dia{\pi~\rust{\loopScope{x}{\{\kwContinue $\many{s}$ $e$\}}}~\omega}{\phi}}}
  }
\]

\begin{theorem}[Soundness of Rule
  $\ruleName{continue}$]\label{thm:sound-continue}
  Let $\Gamma,\Delta\in 2^{\DLFor}$ be sets of formulas,
  $\updateContext\in\Updates$ be an update,
  $\pv{x}: \mbox{\lstinline|bool|}$ be a program variable,
  $\pi~\rust{\loopScope{x}{\kwContinue; $\many{s}$ $e$}}~\omega$ be a
  program fragment, $\many{s}$ a series of statements, $e$ an
  expression, and $\phi\in\DLFor$ a formula. If
  \begin{equation*}
    \models\sequent{}{\applyUp{\updateContext}{\dlf{x = \rsFalse}{\phi}}}
  \end{equation*}
  then the following holds:
  \[
    \models\sequent{}{\applyUp{\updateContext}{\dia{\pi~\rust{\loopScope{x}{\{\kwContinue $\many{s}$ $e$\}}}~\omega}{\phi}}}.
  \]
\end{theorem}
\begin{proof}
  We assume
  \begin{equation}\label{eq:continue-asm-1}
    \models\sequent{}{\applyUp{\updateContext}{\dlf{x = \rsFalse}{\phi}}}
  \end{equation}

  It suffices to show
  \begin{equation}\label{eq:continue-goal-1}
    \validSOne\dia{\pi~\rust{\loopScope{x}{\{\kwContinue $\many{s}$ $e$\}}}~\omega}{\phi}.
  \end{equation}
  We infer from the assumption~\zcref[noname]{eq:break-asm-1} that
  \begin{equation}\label{eq:continue-asm-2}
    \validSOne\dlf{x = \rsFalse}{\phi}
  \end{equation}

  Let $\sigma_1$ be an \ac{LAGC} state with $s_1\dlLagcSim\sigma_1$ By
  \zcref{eq:local-ls-continue}, we have
  \begin{equation}
    \valB{\sigma_1}{\rust{\loopScope{x}{\{\kwContinue
          $\many{s}$
          $e$\}}}}=\{\concatTr{\consTr{\update{\sigma_1}{\pv{x}}{\falseSem}}{\consTr{\continueEv}{\mkTr{\sigma_1}}}}{\cont{\rust{()}}}\}
  \end{equation}

  By \zcref{eq:rule-compLoopScope}, we know that after the continue
  event, the composition cannot continue. Hence, we infer that all
  traces in $\traces{\rust{\loopScope{x}{\{\kwContinue
        $\many{s}$ $e$\}}},\sigma_1}$ are of the form
  $\concatTr{\sh}{\consTr{\update{\sigma_1}{\pv{x}}{\falseSem}}{\consTr{\continueEv}{\mkTr{\sigma_1}}}}$
  for some shining trace~$\sh$. Let
  $\sigma_2:=\update{\sigma_1}{\pv{x}}{\falseSem}$. It is trivial to
  derive that
  $\traces{\rust{x=\rsFalse},\sigma_1}=\{\consTr{\sigma_2}{\mkTr{\sigma_1}}\}$. Since
  all traces in both sets end in state $\sigma_2$, \zcref{def:rho}
  implies the goal~\zcref[noname]{eq:continue-goal-1} holds.
\end{proof}


\end{document}